\documentclass{amsart}
\usepackage{mathpazo,fullpage,amsmath,amsthm,amssymb,xcolor,graphicx}
\newcommand{\ee}{\mathrm{e}}
\newcommand{\ii}{\mathrm{i}}
\newcommand{\dd}{\mathrm{d}}
\newtheorem{rhp}{Riemann-Hilbert Problem}
\newtheorem{lemma}{Lemma}
\newtheorem{theorem}{Theorem}

\title{On the Algebraic Solutions of the Painlev\'e-III (D7) Equation}
\author{R. J. Buckingham \and P. D. Miller}
\thanks{R. J. Buckingham was supported by National Science Foundation grant DMS-2108019.  P. D. Miller was supported by National Science Foundation grant DMS-1812625.}
\address{Department of Mathematical Sciences, University of Cincinnati, P. O. Box 210025, Cincinnati, OH 45221}
\email{buckinrt@uc.edu}
\address{Department of Mathematics, University of Michigan, East Hall, 530 Church St., Ann Arbor, MI 48109}
\email{millerpd@umich.edu}

\date{\today}

\begin{document}
\maketitle
\begin{abstract}
The D7 degeneration of the Painlev\'e-III equation has solutions that are rational functions of $x^{1/3}$ for certain parameter values.  We apply the isomonodromy method to obtain a Riemann-Hilbert representation of these solutions.  We demonstrate the utility of this representation by analyzing rigorously the behavior of the solutions in the large parameter limit.
\end{abstract}

\section{Introduction}
The six Painlev\'e equations were discovered more than a century ago by Paul Painlev\'e in his classification of all second-order first-degree ordinary differential equations algebraic in the unknown function, rational in its first derivative, and analytic in the independent variable having the \emph{Painlev\'e property} that all solutions are meromorphic away from certain fixed singularities whose locations are fully determined from the equation itself.  It turned out that all such equations could be reduced either to previously known  (linear) equations or to one of the equations usually denoted PI, PII, PIII, PIV, PV, and PVI.  See \cite[Ch.\@ XIV]{Ince56}.  All of these equations except for PI involve one or more free parameters.
As a consequence of the Painlev\'e property, solutions of these equations may be regarded as new special functions, and they occur in many applications.
Although typical solutions of Painlev\'e equations are highly-transcendental functions (indeed, they are typically called \emph{Painlev\'e transcendents}), all of the equations except for PI also admit, for certain parameter values, solutions in terms of elementary functions or classical linear special functions (e.g., Airy, Bessel, etc.). The parameter values for these solutions are related by a certain finitely-generated group action, and the group acts on the solutions via B\"acklund transformations that preserve the functional character of the solution (rational, algebraic, etc.).  

From the very beginning, the Painlev\'e equations were recognized by R. Fuchs and Garnier as isomonodromic deformations of certain second-order linear equations (see also \cite{Okamoto86}).  It is convenient to work at the level of first-order systems instead, in which case each Painlev\'e equation can be recognized as the compatibility condition for a certain \emph{Lax pair} of linear equations for an auxiliary $2\times 2$ matrix-valued unknown $\boldsymbol{\Psi}$.  The inverse problem of constructing $\boldsymbol{\Psi}$ from its monodromy data can be usefully formulated as a matrix Riemann-Hilbert problem.  The aforementioned group actions reappear in this context as \emph{Schlesinger transformations}, linear gauge transformations acting on the matrix unknown that preserve the form of the Lax pair as well as the essential monodromy data, affecting only formal monodromy exponents at various singular points.  This means that the Riemann-Hilbert representation of the whole family of special solutions can be obtained once it is known for just one particular choice of parameters and solution, usually called the \emph{seed solution}.  Generally the monodromy data for a given solution cannot be obtained explicitly; however in the case of elementary-function solutions the direct problem for the Lax pair can frequently be solved in terms of classical special functions in which case the monodromy data can be found by applying known connection formul\ae.    The isomonodromy method has been successfully applied to rational solutions of 
\begin{itemize}
\item The PII equation \cite{BuckinghamM12,BuckinghamM14,BuckinghamM15,FlaschkaN80}.  Here if one uses the Jimbo-Miwa \cite{JimboM81} Lax pair for the PII equation, the direct problem for the seed solution is solved in terms of Airy functions.  There is another approach for analyzing directly the Yablonskii-Vorob'ev polynomials that can be used to construct the rational solutions based on a Hankel determinant representation \cite{BertolaB15}; this was shown to be equivalent to the isomonodromy approach in the alternate setting of the Flaschka-Newell Lax pair in \cite{MillerSh17}.
\item The PIII equation (nondegenerate D6 type) \cite{BothnerMS18,BothnerM19}.  The direct problem for the seed solution is solved in terms of confluent hypergeometric functions (Whittaker functions).
\item The PIV equation \cite{BuckinghamM20}.  In this problem there are two distinct families of rational solutions; the direct problem for the so-called generalized Hermite rational solutions is solved in terms of elementary functions while that for the so-called generalized Okamoto rational solutions is solved in terms of Airy functions.  The generalized Hermite solutions can be analyzed by means of a Hankel determinant representation similarly to the PII case; see \cite{Buckingham18}.  Yet another approach in this case exploits a connection with the spectral theory of quantum oscillators in quartic potentials; see \cite{MasoeroR18,MasoeroR21}.
\end{itemize}
The goal of this paper is to show that the isomonodromy method applies equally well to the algebraic solutions of the D7 degeneration of the PIII equation.  These are not rational solutions, although they are rational functions of the cube root of the independent variable. 

\subsection{PIII D7 equation and its algebraic solutions}
The Painlev\'e-III equation of D7 type has the form
\begin{equation}
u''=\frac{(u')^2}{u}-\frac{u'}{x}+\frac{\alpha u^2+\beta}{x}+\frac{\delta}{u},\quad u=u(x),\quad\delta\neq 0,
\label{eq:D7}
\end{equation}
which is a degenerate case of the general (D6) Painlev\'e-III equation
\begin{equation}
u''=\frac{(u')^2}{u}-\frac{u'}{x}+\frac{\alpha u^2+\beta}{x}+\gamma u^3+\frac{\delta}{u}
\end{equation}
in which $\gamma=0$.  According to \cite[\S 32.9(i)]{DLMF}, if $\alpha=1$, $\beta=2n$ with $n\in\mathbb{Z}$, and $\delta=-1$, the D7 equation \eqref{eq:D7} has a solution that is a rational function of $x^{1/3}$.  When $n=0$, that solution is $u(x)=x^{1/3}$.  By scaling $u=\sqrt{8}\widetilde{u}$ and $x=\sqrt{8}\widetilde{x}$ and dropping the tildes, we see that $u(x)=\tfrac{1}{2}x^{1/3}$ is a solution of \eqref{eq:D7} with $\alpha=8$, $\beta=0$, and $\delta=-1$, while for parameters $\alpha=8$,  $\beta=2n$, and $\delta=-1$ there is a solution of \eqref{eq:D7} that is rational in $x^{1/3}$.

If one substitutes into \eqref{eq:D7} with $\alpha=8$, $\beta=2n$, and $\delta=-1$ the formal expression $u=\tfrac{1}{2}x^{1/3}v(x^{1/3})$ where $v(\zeta)=1+v_1\zeta^{-1}+v_2\zeta^{-2}+\cdots$, then obtains a systematic recurrence to determine the coefficients $v_j$ in order that never requires division by zero.  Hence all coefficients are uniquely determined by the value of $n$ (and in fact the coefficients of all odd powers of $\zeta$ vanish), so the solution rational in $\zeta$ is unique for given $n$.  
We denote this solution by $u=u_n(x)$.  There is a sequence $\{R_n(\zeta)\}_{n=0}^\infty$ of \emph{Ohyama polynomials} \cite{Clarkson03,OhyamaKSO06} satisfying a simple recurrence relation such that if $n>0$, $u_n(x)=R_{n+1}(x^{1/3})R_{n-1}(x^{1/3})/R_n(x^{1/3})^2$.  This representation is consistent with the fact that all poles not at the origin of solutions of \eqref{eq:D7} are double poles.  As shown in \cite[Sec.\@ 3]{Clarkson03} and Figure \ref{fig:ohyama}, when displayed in the $\zeta=x^{1/3}$ plane, the poles and zeros of $u_n(x)$ for integers $n>0$ form an approximately crystalline pattern confined within two quasi-triangular regions forming a ``bow-tie'' shape.  One can also deduce from $u_0(x)=\frac{1}{2}x^{1/3}$ and the B\"acklund transformations \eqref{eq:Baecklund-up}--\eqref{eq:Baecklund-down} below (see also \cite[Eqn.\@ 3.4]{Clarkson03}) that 
\begin{equation}
\zeta\mapsto \zeta^*\implies u_n\mapsto u_n^*\quad\text{and}\quad \zeta\mapsto -\zeta\implies u_n\mapsto -u_n,\quad \text{for all }n\in\mathbb{Z}.
\label{eq:symmetries}
\end{equation}
\begin{figure}[h]
\begin{center}
\includegraphics[width=2in]{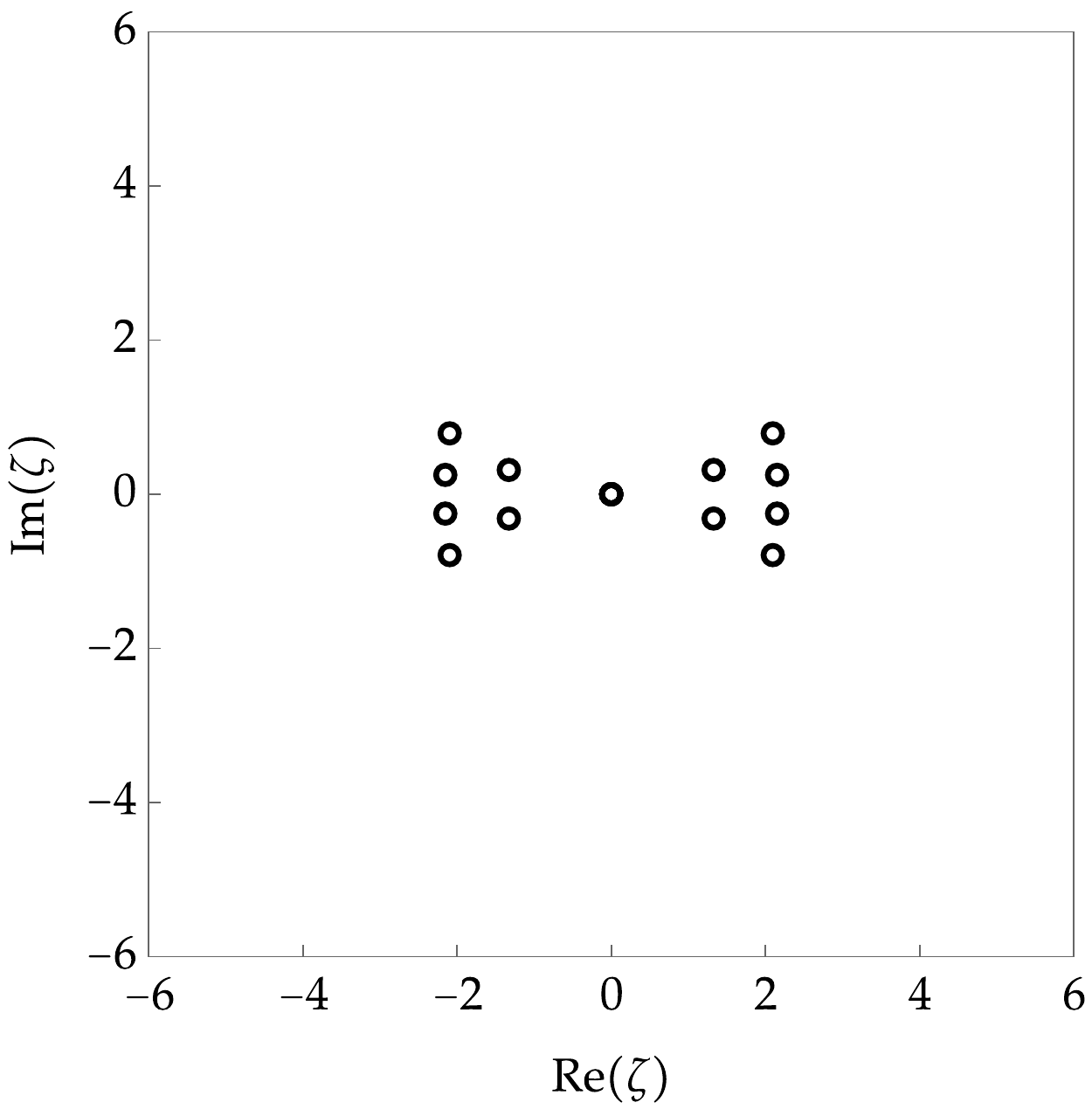}
\includegraphics[width=2in]{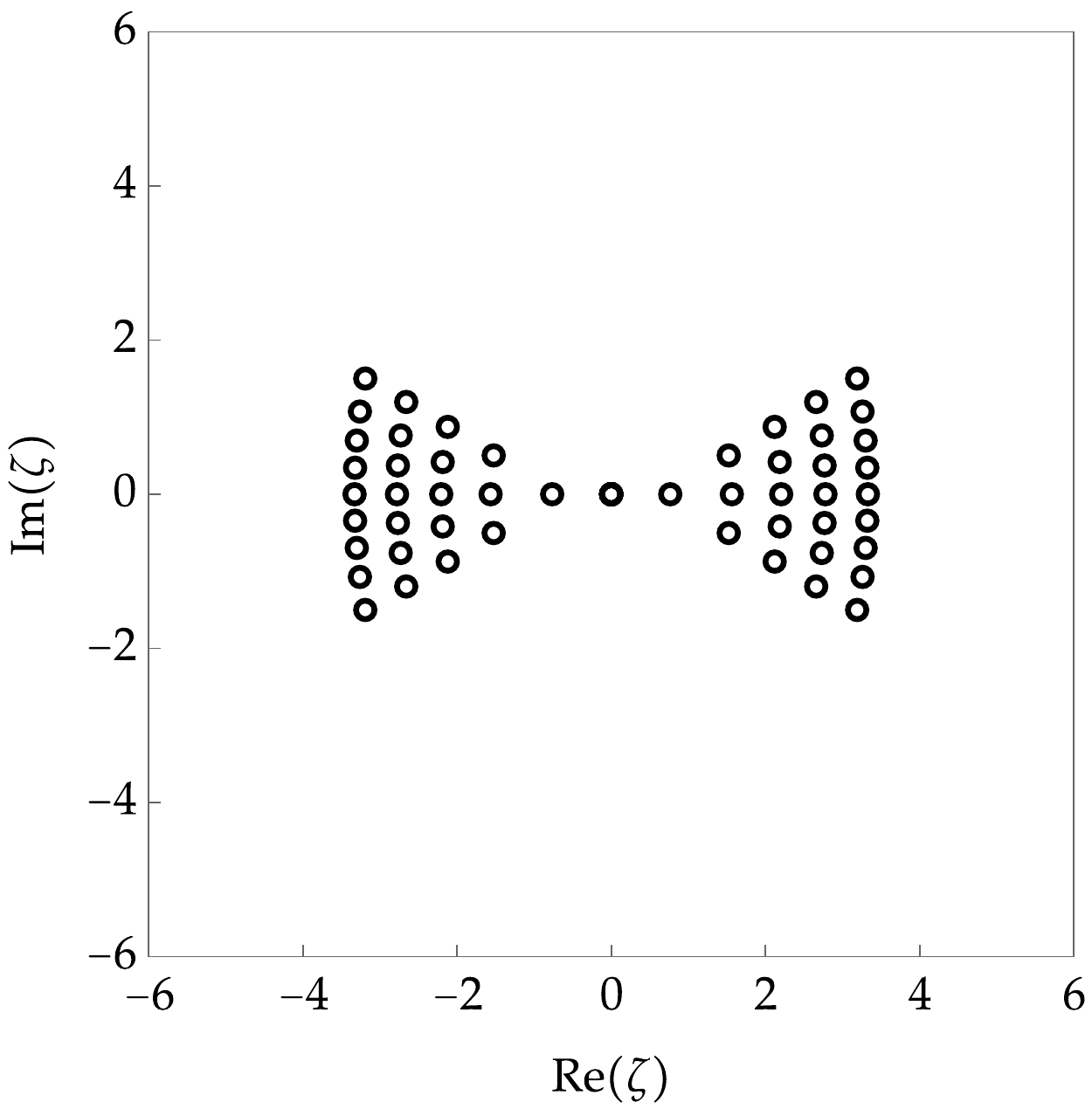}
\includegraphics[width=2in]{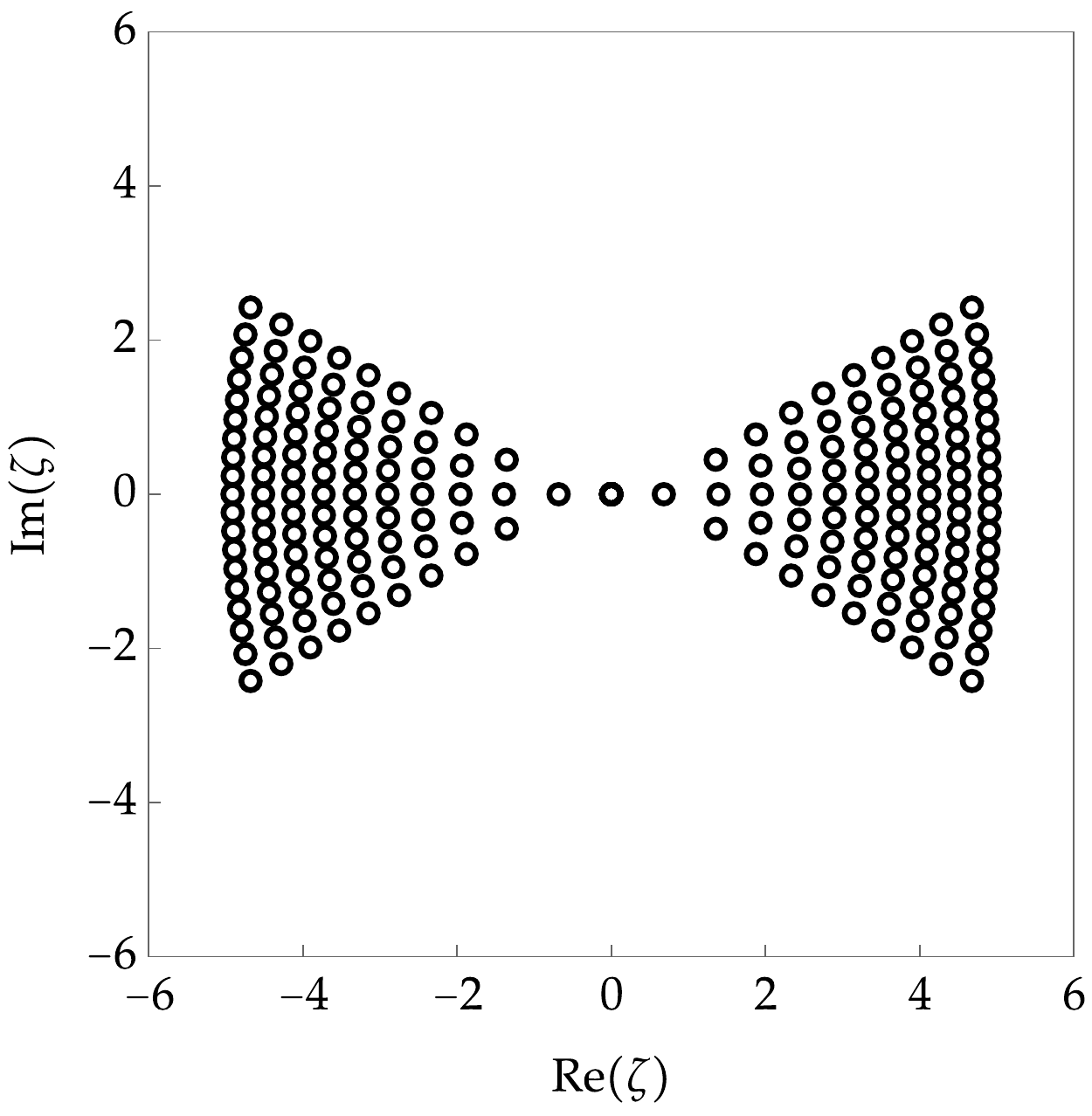}
\end{center}
\caption{Complex zeros of the Ohyama polynomials $R_n(\zeta)$.  Left:  $n=5$.  Center: $n=10$.  Right: $n=20$.}
\label{fig:ohyama}
\end{figure}

\subsection{Lax pair}
In \cite{KitaevV04}, the equation \eqref{eq:D7} is written in the form ($\alpha=-8\epsilon$, $\beta=2ab$, $\delta=b^2$)
\begin{equation}
u''=\frac{(u')^2}{u}-\frac{u'}{x}+\frac{-8\epsilon u^2+2ab}{x}+\frac{b^2}{u},\quad u=u(x).
\label{eq:D7KV}
\end{equation}
The authors of \cite{KitaevV04} introduce the following linear equations on an unknown $\boldsymbol{\Psi}=\boldsymbol{\Psi}(\lambda,x)$:
\begin{equation}
\boldsymbol{\Psi}_\lambda=\boldsymbol{\Lambda}(\lambda,x)\boldsymbol{\Psi}\quad\text{and}\quad
\boldsymbol{\Psi}_x=\mathbf{X}(\lambda,x)\boldsymbol{\Psi}
\label{eq:Lax}
\end{equation}
where the coefficient matrices are defined by 
\begin{equation}
\boldsymbol{\Lambda}(\lambda,x):=-\ii x\sigma_3 -\frac{\ii a}{2\lambda}\sigma_3-\frac{x}{\lambda}\mathbf{J}(x)+\frac{\ii x}{2\lambda^2}\mathbf{K}(x)\quad\text{and}\quad
\mathbf{X}(\lambda,x):=-\ii\lambda\sigma_3 +\frac{\ii a}{2x}\sigma_3 -\mathbf{J}(x)-\frac{\ii}{2\lambda}\mathbf{K}(x),
\label{eq:Lambda-X-matrices}
\end{equation}
with
\begin{equation}
\mathbf{J}(x):=\frac{\epsilon x}{4u(x)}\begin{bmatrix}0 & p(x)\\
q(x) & 0\end{bmatrix}\quad\text{and}\quad
\mathbf{K}(x):=\frac{u(x)}{x}\begin{bmatrix}\epsilon & \ee^{\ii\varphi(x)}\\-\ee^{-\ii\varphi(x)} & -\epsilon\end{bmatrix}.
\label{eq:M1M2}
\end{equation}
The compatibility condition of the overdetermined system \eqref{eq:Lax} is $\boldsymbol{\Lambda}_x-\mathbf{X}_\lambda + [\boldsymbol{\Lambda},\mathbf{X}]=\mathbf{0}$, which upon separating out the coefficients of $\lambda^{-1}$ and $\lambda^{-2}$ amounts to the following system of equations:
\begin{align}
\label{eq:first}
u'(x)-\frac{u(x)}{x}-\frac{x}{2}\left[p(x)\ee^{-\ii\varphi(x)}+q(x)\ee^{\ii\varphi(x)}\right]&=0,\\
\label{eq:second}
-\frac{\epsilon x^2 p'(x)}{4u(x)}+\frac{\epsilon x^2 p(x)u'(x)}{4u(x)^2} - \frac{\epsilon xp(x)}{2u(x)}+\frac{\ii \epsilon a x p(x)}{2u(x)}-2u(x)\ee^{\ii\varphi(x)}&=0,\\
\label{eq:p-explicit}
-\frac{1}{2}u(x)\varphi'(x)\ee^{\ii\varphi(x)}+\frac{1}{2}\ii u'(x)\ee^{\ii\varphi(x)}-\frac{\ii}{2x}u(x)\ee^{\ii\varphi(x)}-\frac{\ii\epsilon^2}{2}xp(x)&=0,\\
\label{eq:fourth}
-\frac{\epsilon x^2 q'(x)}{4u(x)}+\frac{\epsilon x^2 q(x)u'(x)}{4u(x)^2}-\frac{\epsilon x q(x)}{2u(x)}-\frac{\ii\epsilon a x q(x)}{2u(x)}-2u(x)\ee^{-\ii\varphi(x)}&=0,\\
\label{eq:q-explicit}
-\frac{1}{2}u(x)\varphi'(x)\ee^{-\ii\varphi(x)}-\frac{1}{2}\ii u'(x)\ee^{-\ii\varphi(x)}+\frac{\ii}{2x}u(x)\ee^{-\ii\varphi(x)} +\frac{\ii\epsilon^2}{2}xq(x)&=0.
\end{align}
We eliminate $p(x)$ and $q(x)$ explicitly using \eqref{eq:p-explicit} and \eqref{eq:q-explicit}:
\begin{equation}
\begin{split}
p(x)&=\frac{1}{\epsilon^2}\left[\frac{\ii u(x)\varphi'(x)\ee^{\ii\varphi(x)}}{x}+\frac{u'(x)\ee^{\ii\varphi(x)}}{x}-\frac{u(x)\ee^{\ii\varphi(x)}}{x^2}\right]=\frac{1}{\epsilon^2}\frac{\dd}{\dd x}\left(\frac{u(x)\ee^{\ii\varphi(x)}}{x}\right)\\
q(x)&=\frac{1}{\epsilon^2}\left[-\frac{\ii u(x)\varphi'(x)\ee^{-\ii\varphi(x)}}{x}+\frac{u'(x)\ee^{-\ii\varphi(x)}}{x}-\frac{u(x)\ee^{-\ii\varphi(x)}}{x^2}\right]=\frac{1}{\epsilon^2}\frac{\dd}{\dd x}\left(\frac{u(x)\ee^{-\ii\varphi(x)}}{x}\right).
\end{split}
\label{eq:pq}
\end{equation}
Then \eqref{eq:first} becomes
\begin{equation}
\left(1-\frac{1}{\epsilon^2}\right)\left(u'(x)-\frac{u(x)}{x}\right)=0,
\label{eq:epssquared}
\end{equation}
and the sum and difference of \eqref{eq:second} and \eqref{eq:fourth} become, respectively,
\begin{align}
\label{eq:penultimate}
-8\epsilon  u(x)^3-u(x)u'(x)+xu'(x)^2-2au(x)^2\varphi'(x)+xu(x)^2\varphi'(x)^2-xu(x)u''(x)&=0\\
\label{eq:integrable}
u(x)\varphi''(x)+ u'(x)\varphi'(x)-\frac{2a u'(x)}{x}+\frac{2a u(x)}{x^2}=0.
\end{align}
We note that \eqref{eq:integrable} can be written in the form
\begin{equation}
\frac{\dd}{\dd x}\left(u(x)\varphi'(x) -\frac{2a u(x)}{x}\right)=0\implies u(x)\varphi'(x)=\frac{2a u(x)}{x}+b
\label{eq:phiprime}
\end{equation}
where $b$ is an integration constant.  Using this to eliminate $u(x)\varphi'(x)$ from \eqref{eq:penultimate} gives the equation \eqref{eq:D7KV} on $u(x)$.  If we assume that $\epsilon^2=1$, then \eqref{eq:epssquared} places no further conditions on $u(x)$ (otherwise conditions on the parameters $\epsilon$, $a$, and $b$ are required so that \eqref{eq:D7KV} admits a solution of the form $u(x)=Ax$ for $A\neq 0$).  

Using $\epsilon^2=1$, we note that, in terms of the matrix elements of $\mathbf{J}(x)$ and $\mathbf{K}(x)$, the potentials are given by
\begin{equation}
\begin{split}
u(x)&=\epsilon x K_{11}(x)=-\epsilon x K_{22}(x)\\
\ee^{\ii\varphi(x)}&=\epsilon\frac{K_{12}(x)}{K_{11}(x)}=-\epsilon\frac{K_{12}(x)}{K_{22}(x)}\\
\ee^{-\ii\varphi(x)}&=-\epsilon\frac{K_{21}(x)}{K_{11}(x)} = \epsilon\frac{K_{21}(x)}{K_{22}(x)}\\
p(x)&=4J_{12}(x)K_{11}(x)=-4J_{12}(x)K_{22}(x)\\
q(x)&=4J_{21}(x)K_{11}(x)=-4J_{21}(x)K_{22}(x).
\end{split}
\label{eq:potentials-from-M12}
\end{equation}
The equivalence of the two expressions in each case is guaranteed from $\mathrm{tr}(\mathbf{K}(x))=0$, and the compatibility of the expressions for $\ee^{\pm\ii\varphi(x)}$ is implied by $\det(\mathbf{K}(x))=0$.  If we are given the matrices $\mathbf{J}(x)$ (off-diagonal) and $\mathbf{K}(x)$ (singular and nondiagonalizable), we cannot determine the value of $\epsilon=\pm 1$.  However, from \eqref{eq:potentials-from-M12} we can see that $\epsilon\mapsto -\epsilon$ changes the signs of $u(x)$ and $\ee^{\pm\ii\varphi(x)}$ but leaves $p(x)$ and $q(x)$ invariant.  It follows that the equations \eqref{eq:first}--\eqref{eq:q-explicit} are invariant under $\epsilon\mapsto -\epsilon$, and from \eqref{eq:phiprime} we see that the integration constant $b$ is proportional to $\epsilon$.  This is then consistent with the obvious symmetry of \eqref{eq:D7KV}:  $u\mapsto -u$, $\epsilon\mapsto -\epsilon$, $b\mapsto -b$.

Therefore, for the rest of this work, we assume without loss of generality that $\epsilon=-1$ in the parametrization of the matrices $\mathbf{J}(x)$ and $\mathbf{K}(x)$ in the Lax pair \eqref{eq:Lax}--\eqref{eq:M1M2}.  Then eliminating $u'(x)$ between \eqref{eq:p-explicit} and \eqref{eq:q-explicit} allows the integration constant $b$ to be expressed without derivatives as
\begin{equation}
b = -\frac{2au(x)}{x}-\frac{1}{2}\ii xp(x)\ee^{-\ii\varphi(x)} +\frac{1}{2}\ii x q(x)\ee^{\ii\varphi(x)}.
\label{eq:b-no-derivatives}
\end{equation}

\subsection{Outline of the paper}
In Section~\ref{sec:direct-monodromy-seed} we show how simultaneous solutions of the Lax pair equations for the simplest algebraic solution of \eqref{eq:D7KV}, $u_0(x)=\frac{1}{2}x^{1/3}$ for $\epsilon=-1$, $a=0$, and $b=\ii$, can be explicitly obtained in terms of Airy functions.  Then we build canonical bases of simultaneous solutions near the singular points at $\lambda=0$ and $\lambda=\infty$ and apply connection formul\ae\ to determine relations among them.  Next, in Section~\ref{sec:RHPs} we formulate the inverse monodromy problem for the rational solution $u_n(x)$ for $a=-\ii n$ and $n\in\mathbb{Z}$ as a Riemann-Hilbert problem.  We demonstrate that the problem is solved for $n=0$ by the seed Lax pair solutions constructed in Section~\ref{sec:direct-monodromy-seed}.  We then derive differential equations from its solution in general, recovering  the PIII (D7) equation in the form \eqref{eq:D7KV} for $\epsilon=-1$, $a=-\ii n$, and $b=\ii$.  Next we show that the solutions are related for consecutive values of $n$ by Schlesinger transformations and hence the solutions of \eqref{eq:D7KV} are related by B\"acklund transformations.  Since the latter preserve the algebraic character of the solution, this shows that the Riemann-Hilbert problem captures the algebraic solutions of \eqref{eq:D7KV}; we summarize this result in Theorem~\ref{thm:RHP-representation} below.  We conclude this section by making a convenient transformation of the Riemann-Hilbert problem that has the effect of simplifying the data for the problem.  Then, in Section~\ref{sec:application} we use the Riemann-Hilbert representation of the algebraic solution $u_n(x)$ to consider the asymptotic behavior of the solution for large $n$.  After some preliminary rescaling of the independent variable $x$ and the spectral variable $\lambda$ to balance exponents, we present for background some formal arguments applying similar scalings to the PIII (D7) equation itself.  This formal approach suggests two types of approximations:  a slowly-varying ``equilibrium'' approximation and an approximation based on the Weierstra\ss\ elliptic function.  Then, we return to the rescaled Riemann-Hilbert problem and introduce an appropriate $g$-function via a family of spectral curves.  The family of spectral curves mirrors the family of formally-approximate solutions of the PIII (D7) equation in a remarkable fashion.  Finally, we carry out a rigorous analysis of the function $u_n(x)$ for large $n$ and rescaled $x>0$ sufficiently large and recover one of the equilibrium solutions predicted by the formal theory (see Theorem~\ref{thm:positive-exterior} below).  With some extra steps, the method allows this result to be continued to the exterior of the ``bow-tie'' domain of the $x^{1/3}$ plane and hence makes rigorous some of the observations made in \cite{Clarkson03}.  See Theorem~\ref{thm:general-exterior} below.

\subsection{Notation}
Throughout our paper, square matrices are indicated by boldface capital letters with the exception of the identity $\mathbb{I}$ and the Pauli matrices:
\begin{equation}
\sigma_1:=\begin{bmatrix}0&1\\1&0\end{bmatrix},\quad\sigma_2:=\begin{bmatrix}0&-\ii\\\ii & 0\end{bmatrix},\quad\sigma_3:=\begin{bmatrix}1&0\\0&-1\end{bmatrix}.
\end{equation}
For a function (matrix or scalar-valued) analytic off an oriented contour arc, we use subscripts $+$ (resp., $-$) to indicate boundary values at a point of the arc from the left (resp., right).

\section{The direct monodromy problem for the seed solution}
\label{sec:direct-monodromy-seed}
\subsection{Lax pair for the seed solution}
To obtain the algebraic solutions of the Painlev\'e-III D7 equation in the form \eqref{eq:D7KV}, we assume that $\epsilon=-1$, $a=-\ii n$, and $b=\ii$ for $n\in\mathbb{Z}$.  In the case that $n=0$ we have the algebraic solution $u(x)=\tfrac{1}{2}x^{1/3}$.  We call this the \emph{seed solution} for the algebraic solution family, and $n=0$ is the \emph{seed parameter value}.  Taking $a=0$, $b=\ii$, and $u(x)=\tfrac{1}{2}x^{1/3}$ in \eqref{eq:phiprime} gives $\varphi'(x)=2\ii x^{-1/3}$ so that  $\varphi(x)=3\ii x^{2/3}+\varphi_0$ where $\varphi_0$ is an integration constant that we will take to be zero.  From \eqref{eq:pq} and $\epsilon=-1$ we then get 
\begin{equation}
p(x)=-\left(\frac{1}{x}+\frac{1}{3}x^{-5/3}\right)\ee^{-3x^{2/3}}\quad\text{and}\quad
q(x)=-\left(-\frac{1}{x}+\frac{1}{3}x^{-5/3}\right)\ee^{3x^{2/3}}.
\label{eq:pq0}
\end{equation}
It follows that for the seed solution, the matrices $\mathbf{J}(x)$ and $\mathbf{K}(x)$ in \eqref{eq:M1M2} are given by
\begin{equation}
\mathbf{J}(x)=\begin{bmatrix}0 & \left(\tfrac{1}{6}x^{-1}+\tfrac{1}{2}x^{-1/3}\right)\ee^{-3x^{2/3}}\\
\left(\tfrac{1}{6}x^{-1}-\tfrac{1}{2}x^{-1/3}\right)\ee^{3x^{2/3}}&0\end{bmatrix},\;
\mathbf{K}(x)=\frac{1}{2}x^{-2/3}\begin{bmatrix}-1 & \ee^{-3x^{2/3}}\\-\ee^{3x^{2/3}} & 1\end{bmatrix}.
\label{eq:M1M2-seed}
\end{equation}
We make a gauge transformation to remove the exponential factors on the off-diagonal elements of the coefficient matrices:  $\boldsymbol{\Psi}=\ee^{-3x^{2/3}\sigma_3/2}\boldsymbol{\Phi}$, which implies that also $\boldsymbol{\Psi}_x=\ee^{-3x^{2/3}\sigma_3/2}(\boldsymbol{\Phi}_x-x^{-1/3}\sigma_3\boldsymbol{\Phi})$.  In terms of $\boldsymbol{\Phi}$ then, the (compatible) Lax pair equations for the seed solution read
\begin{equation}
\boldsymbol{\Phi}_\lambda=\left(-\ii x\sigma_3-\frac{x}{\lambda}\begin{bmatrix}0 & \tfrac{1}{6}x^{-1}+\tfrac{1}{2}x^{-1/3}\\\tfrac{1}{6}x^{-1}-\tfrac{1}{2}x^{-1/3} & 0\end{bmatrix} +\frac{\ii x^{1/3}}{4\lambda^2}\begin{bmatrix}-1 & 1\\-1 & 1\end{bmatrix}\right)\boldsymbol{\Phi},
\label{eq:seedLaxlambda}
\end{equation}
\begin{equation}
\boldsymbol{\Phi}_x=\left(-\ii\lambda\sigma_3-\begin{bmatrix}-x^{-1/3} & \tfrac{1}{6}x^{-1}+\tfrac{1}{2}x^{-1/3}\\
\tfrac{1}{6}x^{-1}-\tfrac{1}{2}x^{-1/3} & x^{-1/3}\end{bmatrix}-\frac{\ii x^{-2/3}}{4\lambda}\begin{bmatrix}-1 & 1\\-1 & 1\end{bmatrix}\right)\boldsymbol{\Phi}.
\label{eq:seedLaxx}
\end{equation}
Our strategy to obtain a fundamental simultaneous solution matrix for \eqref{eq:seedLaxlambda}--\eqref{eq:seedLaxx} is to deal first with the equation \eqref{eq:seedLaxx}.  For this purpose, it is convenient to simplify the latter equation as much as possible.  We observe that multiples of the same diagonal coefficient matrix appear in the terms most singular at $\lambda=\infty$ in both equations \eqref{eq:seedLaxlambda}--\eqref{eq:seedLaxx},  and at the same time the coefficient matrix of the terms most singular at $\lambda=0$ is nilpotent (that this coefficient matrix is defective for all $x$ is what distinguishes the D7 Lax pair from the most general D6 Lax pair given by Jimbo and Miwa) which is a simplification to be retained.  Hence 
we will remove the leading terms at $\lambda=\infty$ from \eqref{eq:seedLaxx} by introducing the shearing transformation
\begin{equation}
\lambda=X^{-1}\Lambda,\; x=X\iff \Lambda=x\lambda,\; X=x.
\end{equation}
Then the partial derivatives with respect to $\lambda$ and $x$ transform as follows:
\begin{equation}
\frac{\partial}{\partial\Lambda}=\frac{1}{X}\frac{\partial}{\partial\lambda}\quad\text{and}\quad\frac{\partial}{\partial X}=\frac{\partial}{\partial x}-\frac{\Lambda}{X^2}\frac{\partial}{\partial\lambda}.
\end{equation}
In terms of the new variables $\Lambda$ and $X$, the compatible system \eqref{eq:seedLaxlambda}--\eqref{eq:seedLaxx} becomes
\begin{equation}
\boldsymbol{\Phi}_\Lambda = \left(-\ii \sigma_3-\frac{X}{\Lambda}\begin{bmatrix}0 & \tfrac{1}{6}X^{-1}+\tfrac{1}{2}X^{-1/3}\\
\tfrac{1}{6}X^{-1}-\tfrac{1}{2}X^{-1/3} & 0\end{bmatrix}+\frac{\ii X^{4/3}}{4\Lambda^2}\begin{bmatrix}-1 & 1\\-1 & 1\end{bmatrix}\right)\boldsymbol{\Phi}
\label{eq:seedLaxLambda}
\end{equation}
and
\begin{equation}
\boldsymbol{\Phi}_X=\left(X^{-1/3}\sigma_3-\frac{\ii X^{1/3}}{2\Lambda}\begin{bmatrix}-1&1\\-1&1\end{bmatrix}\right)\boldsymbol{\Phi}.
\label{eq:seedLaxX}
\end{equation}
(We note that the shearing transformation had the added benefit of removing the off-diagonal terms from the coefficient of $\lambda^0$.)
To solve \eqref{eq:seedLaxX} for fixed $\Lambda\in\mathbb{C}$, we make a change of independent variable:
\begin{equation}
Z=X^{1/3} \implies \frac{\dd}{\dd X} = \frac{1}{3}X^{-2/3}\frac{\dd}{\dd Z} =\frac{1}{3}Z^{-2}\frac{\dd}{\dd Z}.
\end{equation}
Then \eqref{eq:seedLaxX} becomes
\begin{equation}
\boldsymbol{\Phi}_Z=\left(3Z\sigma_3-\frac{3\ii Z^3}{2\Lambda}\begin{bmatrix}-1&1\\-1&1\end{bmatrix}\right)\boldsymbol{\Phi}.
\end{equation}
Introducing the constant gauge transformation
\begin{equation}
\boldsymbol{\Omega}:=\mathbf{G}\boldsymbol{\Phi},\quad \mathbf{G}:=\frac{1}{\sqrt{2}}\begin{bmatrix}1&-1\\1&1\end{bmatrix}
\label{eq:gauge-Omega-Phi}
\end{equation}
and noting that 
\begin{equation}
\mathbf{G}\sigma_3\mathbf{G}^{-1}=\sigma_1\quad\text{and}\quad
\mathbf{G}\begin{bmatrix}-1&1\\-1&1\end{bmatrix}\mathbf{G}^{-1}=\begin{bmatrix}0&0\\-2&0\end{bmatrix},
\label{eq:conjugation-identities-I}
\end{equation}
we get a purely off-diagonal system
\begin{equation}
\boldsymbol{\Omega}_Z=\left(3Z\sigma_1-\frac{3\ii Z^3}{2\Lambda}\begin{bmatrix}0&0\\-2&0\end{bmatrix}\right)\boldsymbol{\Omega}=\begin{bmatrix}0 & 3Z\\3Z+3\ii\Lambda^{-1}Z^3 & 0\end{bmatrix}\boldsymbol{\Omega}.
\label{eq:Omega-Z}
\end{equation}
Therefore, denoting by $D(Z)$ and $S(Z)$ the first-row and second-row elements respectively of either column of $\boldsymbol{\Omega}$, 
\begin{equation}
S_Z=(3Z+3\ii\Lambda^{-1}Z^3)D\quad\text{and}\quad D_Z=3ZS.
\label{eq:SD-system}
\end{equation}
Eliminating $S$ using the second equation gives a second-order equation on $D$:
\begin{equation}
Z^2D_{ZZ}-ZD_Z-4Z^4(\tfrac{9}{4}+\tfrac{9\ii}{4}\Lambda^{-1}Z^2)D=0.
\label{eq:DODE}
\end{equation}
It seems like a good idea to consider the substitution $W:=\tfrac{9}{4}+\tfrac{9\ii}{4}\Lambda^{-1}Z^2$.  Then a simple computation shows that also
\begin{equation}
\frac{\dd}{\dd Z}=\frac{9\ii Z}{2\Lambda}\frac{\dd}{\dd W},\quad \frac{\dd^2}{\dd Z^2}=\frac{9\ii}{2\Lambda}\frac{\dd}{\dd W}-\frac{81Z^2}{4\Lambda^2}\frac{\dd^2}{\dd W^2}\implies Z^2D_{ZZ}-ZD_Z=-\frac{81Z^4}{4\Lambda^2}D_{WW}.
\end{equation}
Therefore, the equation \eqref{eq:DODE} in fact becomes
\begin{equation}
D_{WW}+\frac{16\Lambda^2}{81}WD=0
\end{equation}
which is a scaling of Airy's equation.  Its general solution is
\begin{equation}
D=A\mathrm{Ai}\left(\left(-\frac{16\Lambda^2}{81}\right)^{1/3}W\right)+B\mathrm{Bi}\left(\left(-\frac{16\Lambda^2}{81}\right)^{1/3}W\right)
\end{equation}
where $A,B$ are independent of $W$ (they may depend on $\Lambda$, however).  

If $f(\xi)$ is any solution of Airy's equation $f''(\xi)-\xi f(\xi)=0$, then we have
\begin{equation}
D(Z)=f(\xi),\quad \xi:=\left(\frac{3}{2}\right)^{2/3}(\ii\Lambda)^{2/3}\left(1-\frac{Z^2}{\ii\Lambda}\right),
\label{eq:D-zeta}
\end{equation}
and from \eqref{eq:SD-system},
\begin{equation}
S(Z)=\frac{D_Z(Z)}{3Z}=-\left(\frac{2}{3}\right)^{1/3}(\ii\Lambda)^{-1/3}f'(\xi).
\end{equation}
So, if $f_1(\xi)$ and $f_2(\xi)$ form a fundamental pair of solutions of Airy's equation $f''(\xi)-\xi f(\xi)=0$, then the general solution of \eqref{eq:Omega-Z} is $\boldsymbol{\Omega}=\boldsymbol{\Omega}_0(Z,\Lambda)\mathbf{H}(\Lambda)$ where  $\mathbf{H}(\Lambda)$ is an arbitrary matrix-valued function of $\Lambda$ only, and where
\begin{equation}
\boldsymbol{\Omega}_0(Z,\Lambda):=\boldsymbol{\Delta}(\Lambda)\mathbf{F}(\zeta),\quad \boldsymbol{\Delta}(\Lambda):=\frac{1}{\sqrt{2}}\begin{bmatrix}1&0\\0 & 
-\left(\tfrac{2}{3}\right)^{1/3}(\ii\Lambda)^{-1/3}
\end{bmatrix},\quad\mathbf{F}(\xi):=\begin{bmatrix}f_1(\xi) & f_2(\xi)\\f_1'(\xi) & f_2'(\xi)\end{bmatrix},
\label{eq:Omega0}
\end{equation}
and $\xi$ is defined in terms of $Z$ and $\Lambda$ by \eqref{eq:D-zeta}.  (The factor of $1/\sqrt{2}$ is arbitrary and is chosen for later convenience.)

Using \eqref{eq:conjugation-identities-I} and 
\begin{equation}
\mathbf{G}\sigma_1\mathbf{G}^{-1}=-\sigma_3\quad\text{and}\quad
\mathbf{G}\begin{bmatrix}0&-1\\1&0\end{bmatrix}\mathbf{G}^{-1}=\begin{bmatrix}0&1\\-1&0\end{bmatrix},
\label{eq:conjugation-identities-II}
\end{equation}
we derive from \eqref{eq:seedLaxLambda} and \eqref{eq:gauge-Omega-Phi} that
\begin{equation}
\boldsymbol{\Omega}_\Lambda=\begin{bmatrix}\displaystyle\frac{1}{6\Lambda} & \displaystyle -\ii-\frac{X^{2/3}}{2\Lambda}  \vspace{.02in} \\
\displaystyle -\ii+\frac{X^{2/3}}{2\Lambda}-\frac{\ii X^{4/3}}{2\Lambda^2}&\displaystyle-\frac{1}{6\Lambda}\end{bmatrix}\boldsymbol{\Omega}=\begin{bmatrix}\displaystyle\frac{1}{6\Lambda} & \displaystyle -\ii-\frac{Z^2}{2\Lambda} \vspace{.02in} \\
\displaystyle -\ii+\frac{Z^2}{2\Lambda}-\frac{\ii Z^4}{2\Lambda^2} & \displaystyle -\frac{1}{6\Lambda}\end{bmatrix}\boldsymbol{\Omega}.
\label{eq:Omega-Lambda}
\end{equation}
Since the systems \eqref{eq:Omega-Z} and \eqref{eq:Omega-Lambda} are compatible, it is possible to find $\mathbf{H}(\Lambda)$ so that $\boldsymbol{\Omega}=\boldsymbol{\Omega}_0(Z,\Lambda)\mathbf{H}(\Lambda)$ solves not only \eqref{eq:Omega-Z} but also \eqref{eq:Omega-Lambda}.  Substituting into the latter shows that $\mathbf{H}(\Lambda)$ must solve the system
\begin{equation}
\mathbf{H}'(\Lambda)=\left(\boldsymbol{\Omega}_0(Z,\Lambda)^{-1}\begin{bmatrix}\displaystyle\frac{1}{6\Lambda} & \displaystyle -\ii-\frac{Z^2}{2\Lambda} \vspace{.02in} \\
\displaystyle -\ii+\frac{Z^2}{2\Lambda}-\frac{\ii Z^4}{2\Lambda^2} & \displaystyle -\frac{1}{6\Lambda}\end{bmatrix}\boldsymbol{\Omega}_0(Z,\Lambda)-\boldsymbol{\Omega}_0(Z,\Lambda)^{-1}\frac{\partial\boldsymbol{\Omega}_0}{\partial\Lambda}(Z,\Lambda)\right)\mathbf{H}(\Lambda).
\end{equation}
Using \eqref{eq:Omega0} shows that 
\begin{equation}
\begin{split}
\boldsymbol{\Omega}_0(Z,\Lambda)^{-1}\frac{\partial\boldsymbol{\Omega}_0}{\partial\Lambda}(Z,\Lambda)&=\boldsymbol{\Omega}_0(Z,\Lambda)^{-1}\boldsymbol{\Delta}'(\Lambda)\boldsymbol{\Delta}(\Lambda)^{-1}\boldsymbol{\Omega}_0(Z,\Lambda) + \frac{\partial\xi}{\partial\Lambda}\boldsymbol{\Omega}_0(Z,\Lambda)^{-1}\boldsymbol{\Delta}(\Lambda)\mathbf{F}'(\xi)\\
\boldsymbol{\Omega}_0(Z,\Lambda)^{-1}\frac{\partial\boldsymbol{\Omega}_0}{\partial Z}(Z,\Lambda)&= \frac{\partial\zeta}{\partial Z}\boldsymbol{\Omega}_0(Z,\Lambda)^{-1}\boldsymbol{\Delta}(\Lambda)\mathbf{F}'(\xi).
\end{split}
\end{equation}
Eliminating $\mathbf{F}'(\xi)$ between these two relations gives
\begin{equation}
\boldsymbol{\Omega}_0(Z,\Lambda)^{-1}\frac{\partial\boldsymbol{\Omega}_0}{\partial\Lambda}(Z,\Lambda)=\boldsymbol{\Omega}_0(Z,\Lambda)^{-1}\boldsymbol{\Delta}'(\Lambda)\boldsymbol{\Delta}(\Lambda)^{-1}\boldsymbol{\Omega}_0(Z,\Lambda)+\left(\frac{\partial\xi}{\partial Z}\right)^{-1}\frac{\partial\xi}{\partial\Lambda}\boldsymbol{\Omega}_0(Z,\Lambda)^{-1}\frac{\partial\boldsymbol{\Omega}_0}{\partial Z}(Z,\Lambda).
\end{equation}
The $Z$-derivative that remains can be eliminated because $\boldsymbol{\Omega}_0(Z,\Lambda)$ is a solution of \eqref{eq:Omega-Z}.  Using also that, from \eqref{eq:D-zeta},
\begin{equation}
\left(\frac{\partial\xi}{\partial Z}\right)^{-1}\frac{\partial\xi}{\partial\Lambda}=\frac{1}{3\ii Z}-\frac{Z}{6\Lambda},
\end{equation}
we therefore find that $\mathbf{H}'(\Lambda)=\boldsymbol{\Omega}_0(Z,\Lambda)^{-1}\mathbf{C}(Z,\Lambda)\boldsymbol{\Omega}_0(Z,\Lambda)\mathbf{H}(\Lambda)$, where
\begin{equation}
\mathbf{C}(Z,\Lambda):=\begin{bmatrix}\displaystyle\frac{1}{6\Lambda} & \displaystyle -\ii-\frac{Z^2}{2\Lambda} \vspace{.02in}\\
\displaystyle -\ii+\frac{Z^2}{2\Lambda}-\frac{\ii Z^4}{2\Lambda^2} & \displaystyle -\frac{1}{6\Lambda}\end{bmatrix}-\boldsymbol{\Delta}'(\Lambda)\boldsymbol{\Delta}(\Lambda)^{-1} -
\left(\frac{1}{3\ii Z}-\frac{Z}{6\Lambda}\right)\begin{bmatrix}0 & 3Z\\3Z+3\ii\Lambda^{-1}Z^3 & 0\end{bmatrix}.
\end{equation}
Differentiating $\boldsymbol{\Delta}(\Lambda)$ using \eqref{eq:Omega0} then shows that a dramatic simplification occurs:  $\mathbf{C}(Z,\Lambda)=(6\Lambda)^{-1}\mathbb{I}$.  Because this is a multiple of the identity, it follows that the differential equation satisfied by $\mathbf{H}(\Lambda)$ reads simply
\begin{equation}
\mathbf{H}'(\Lambda)=\frac{1}{6\Lambda}\mathbf{H}(\Lambda)\implies \mathbf{H}(\Lambda)=\Lambda^{1/6}\mathbf{H}_0
\end{equation}
where $\mathbf{H}_0$ is a matrix independent of both $Z$ and $\Lambda$.  It can be absorbed as a change of basis in specifying the fundamental solutions $f_1(\xi)$ and $f_2(\xi)$ of Airy's equation.  

Inverting the gauge transformations $\boldsymbol{\Phi}\mapsto\boldsymbol{\Omega}$ and $\boldsymbol{\Psi}\mapsto\boldsymbol{\Phi}$ and restoring the original independent variables by $Z=X^{1/3}=x^{1/3}$ and $\Lambda=x\lambda$, the general simultaneous solution matrix of the Lax pair for the seed solution is:
\begin{equation}
\boldsymbol{\Psi}(\lambda,x)=(\ii x\lambda)^{1/6}\ee^{-3x^{2/3}\sigma_3/2}\mathbf{G}^{-1}\boldsymbol{\Delta}(x\lambda)\mathbf{F}(\xi),\quad
\xi:=\left(\frac{3}{2}\right)^{2/3}(\ii x\lambda)^{2/3}\left(1-\frac{x^{2/3}}{\ii x\lambda}\right),
\end{equation}
where $\mathbf{F}(\xi)$ is built from two independent solutions $f_1(\xi)$, $f_2(\xi)$ of Airy's equation $f''(\xi)-\xi f(\xi)=0$ by \eqref{eq:Omega0}.  The extra factor of $\ii^{1/6}$ is included at no cost to give the simplification
\begin{equation}
(\ii x\lambda)^{1/6}\boldsymbol{\Delta}(x\lambda)=\frac{1}{\sqrt{2}}\begin{bmatrix}1&0\\0 & -(\tfrac{2}{3})^{1/3}\end{bmatrix}(\ii x\lambda)^{\sigma_3/6}.
\label{eq:diagonal-terms}
\end{equation}

\subsection{Three particular simultaneous solutions of the Lax pair for the seed}
We assume here for simplicity that $x>0$ and that $(\ii x\lambda)^p=x^p(\ii\lambda)^p$ for any power $p$.  In particular, $\xi$ is then well-defined as a function on the $\lambda$-plane once we select the principal branch for $(\ii\lambda)^{2/3}$.  It has a branch cut emanating vertically upwards from $\lambda=0$ along the imaginary axis, and it is positive real on the negative imaginary axis.

We will make use of the fact that for the Airy equation $f''(\xi)-\xi f(\xi)=0$ there are three complementary sectors:  $0\le\arg(\xi)\le \tfrac{2}{3}\pi$, $-\tfrac{2}{3}\pi\le\arg(\xi)\le 0$, and $-\tfrac{1}{3}\pi\le\arg(-\xi)\le\tfrac{1}{3}\pi$, on each of which there is a basis of solutions with exponential dichotomy and exhibiting no Stokes phenomenon as $\xi\to\infty$.  Those fundamental pairs are $(\mathrm{Ai}(\xi),\mathrm{Ai}(\ee^{-2\pi\ii/3}\xi))$, $(\mathrm{Ai}(\xi),\mathrm{Ai}(\ee^{2\pi\ii/3}\xi))$, and $(\mathrm{Ai}(\ee^{2\pi\ii/3}\xi),\mathrm{Ai}(\ee^{-2\pi\i/3}\xi))$, respectively.

\subsubsection{Solutions near infinity}
Here we consider the domain $|\lambda|>1$ and find simultaneous solutions with no Stokes phenomenon in two complementary sectors near $\lambda=\infty$ and that admit a simple normalization as $\lambda\to\infty$ in each of these sectors.  To be precise, we take unbounded domains $D_\infty^\pm$ defined by the conditions $|\lambda|>1$ and $\pm\mathrm{Re}(\lambda)>0$.  By the definition of $\xi$ for $x>0$ in terms of the principal branch of $(\ii\lambda)^{2/3}$, in the limit $\lambda\to\infty$ from $D_\infty^+$,  $\arg(\ii\lambda)\in (0,\pi)$ translates into $\xi\to\infty$ with $\arg(\xi)\in (0,\tfrac{2}{3}\pi)$.  Likewise as $\lambda\to\infty$ from $D_\infty^-$ we have $\xi\to\infty$ with $-\tfrac{2}{3}\pi<\arg(\xi)<0$.  Therefore, to avoid Stokes phenomenon, we will take for $f(\xi)$ scalar multiples of the solutions from the fundamental pairs $(\mathrm{Ai}(\xi),\mathrm{Ai}(\ee^{\mp 2\pi\ii/3}\xi))$ to define simultaneous solutions $\boldsymbol{\Psi}=\boldsymbol{\Psi}_\infty^\pm(\lambda,x)$ on the domains $D_\infty^\pm$ respectively.  

That is, using \eqref{eq:diagonal-terms}, we take
\begin{equation}
\boldsymbol{\Psi}_\infty^\pm(\lambda,x):=\ee^{-3x^{2/3}\sigma_3/2}\mathbf{G}^{-1}\frac{1}{\sqrt{2}}\begin{bmatrix}1&0\\0 & -(\tfrac{2}{3})^{1/3}\end{bmatrix}
(\ii x\lambda)^{\sigma_3/6}\begin{bmatrix}\mathrm{Ai}(\xi) & \mathrm{Ai}(\ee^{\mp 2\pi\ii/3}\xi)\\\mathrm{Ai}'(\xi) &\ee^{\mp 2\pi\ii/3}\mathrm{Ai}'(\ee^{\mp 2\pi\ii/3}\xi)\end{bmatrix}\begin{bmatrix}c_1^\pm&0\\0 & c_2^\pm\end{bmatrix}
\end{equation}
for suitable constants $c_j^\pm$, $j=1,2$.  This can be equivalently written as
\begin{multline}
\boldsymbol{\Psi}_\infty^\pm(\lambda,x)=\ee^{-3x^{2/3}\sigma_3/2}\begin{bmatrix}\frac{1}{2} & \frac{1}{2} \vspace{.02in} \\-\frac{1}{2} & \frac{1}{2}\end{bmatrix} \\
{}\cdot
\begin{bmatrix}c_1^\pm(\ii x\lambda)^{1/6}\mathrm{Ai}(\xi) & c_2^\pm(\ii x\lambda)^{1/6}\mathrm{Ai}(\ee^{\mp 2\pi \ii/3}\xi) \vspace{.02in} \\
-c_1^\pm(\frac{2}{3})^{1/3}(\ii x\lambda)^{-1/6}\mathrm{Ai}'(\xi) & -c_2^\pm(\frac{2}{3})^{1/3}(\ii x\lambda)^{-1/6}\ee^{\mp 2\pi \ii/3}\mathrm{Ai}'(\ee^{\mp 2\pi \ii/3}\xi)\end{bmatrix}.
\label{eq:Psis-at-infinity}
\end{multline}
As $\lambda\to\infty$ from $D_\infty^\pm$, i.e., with $\pm\arg(\ii x\lambda)\in [0,\pi]$, we have $\xi\to\infty$ with $\pm\arg(\xi)\in [0,\tfrac{2}{3}\pi]$, and also $\ee^{\mp 2\pi\ii/3}\xi\to\infty$ with $\mp\arg(\ee^{\mp 2\pi\ii/3}\xi)\in [0,\tfrac{2}{3}\pi]$.  Using the asymptotic formul\ae\ \cite[Eqns.\@ 9.7.5--6]{DLMF}
\begin{equation}
\mathrm{Ai}(z)=\frac{\ee^{-2z^{3/2}/3}}{2\sqrt{\pi}z^{1/4}}(1+\mathcal{O}(z^{-3/2}))\;\text{and}\;\mathrm{Ai}'(z)=-\frac{\ee^{-2z^{3/2}/3}z^{1/4}}{2\sqrt{\pi}}(1+\mathcal{O}(z^{-3/2})),
\label{eq:Airy-first-term}
\end{equation}
valid as $z\to\infty$ with $|\arg(z)|<\pi$ (and hence $|\arg(z^{3/2})|<\tfrac{3}{2}\pi$ and $|\arg(z^{1/4})|<\tfrac{1}{4}\pi$), along with $\xi = (\frac{3}{2})^{2/3}(\ii x\lambda)^{2/3}(1-x^{2/3}(\ii x\lambda)^{-1})$, we get
\begin{equation}
(\ii x\lambda)^{1/6}\mathrm{Ai}(\xi)
= \frac{\ee^{3x^{2/3}/2}}{2\sqrt{\pi}}(\tfrac{2}{3})^{1/6}\ee^{-\ii x\lambda}(1+\mathcal{O}(\lambda^{-1}))
\end{equation}
and
\begin{equation}
-(\tfrac{2}{3})^{1/3}(\ii x\lambda)^{-1/6}\mathrm{Ai}'(\xi)
=\frac{\ee^{3x^{2/3}/2}}{2\sqrt{\pi}}(\tfrac{2}{3})^{1/6}\ee^{-\ii x\lambda}(1+\mathcal{O}(\lambda^{-1}))
\end{equation}
as $\lambda\to\infty$ in either $D_\infty^+$ or $D_\infty^-$.  Likewise,
\begin{equation}
(\ii x\lambda)^{1/6}\mathrm{Ai}(\ee^{\mp 2\pi\ii/3}\xi)
=\ee^{\pm\ii\pi/6}\frac{\ee^{-3x^{2/3}/2}}{2\sqrt{\pi}}(\tfrac{2}{3})^{1/6}\ee^{\ii x\lambda}(1+\mathcal{O}(\lambda^{-1}))
\end{equation}
and
\begin{equation}
-(\tfrac{2}{3})^{1/3}(\ii x\lambda)^{-1/6}\ee^{\mp 2\pi\ii/3}\mathrm{Ai}'(\ee^{\mp 2\pi\ii/3}\xi)
=\ee^{\mp 5\pi\ii/6}\frac{\ee^{-3x^{2/3}/2}}{2\sqrt{\pi}}(\tfrac{2}{3})^{1/6}\ee^{\ii x\lambda}(1+\mathcal{O}(\lambda^{-1}))
\end{equation}
as $\lambda\to\infty$ from $D_\infty^\pm$.  Therefore, from \eqref{eq:Psis-at-infinity}, we get
\begin{equation}
\boldsymbol{\Psi}_\infty^\pm(\lambda,x)
=\frac{1}{2\sqrt{\pi}}(\tfrac{2}{3})^{1/6}\left(\begin{bmatrix}c_1^\pm & 0\\0 & -\ee^{\pm\ii\pi/6}c_2^\pm\end{bmatrix}+\mathcal{O}(\lambda^{-1})\right)\ee^{-\ii x\lambda\sigma_3}
\end{equation}
in the same limit.  Choosing the constants to be
\begin{equation}
c_1^\pm:=2\sqrt{\pi}(\tfrac{3}{2})^{1/6}\quad\text{and}\quad c_2^\pm:=-\ee^{\mp\ii\pi/6}2\sqrt{\pi}(\tfrac{3}{2})^{1/6},
\label{eq:constants-at-infinity}
\end{equation}
the simultaneous solutions $\boldsymbol{\Psi}_\infty^\pm(\lambda,x)$ are normalized so that
\begin{equation}
\mathop{\lim_{\lambda\to\infty}}_{\lambda\in D_\infty^\pm}\boldsymbol{\Psi}_\infty^\pm(\lambda,x)\ee^{\ii x\lambda\sigma_3}=\mathbb{I}.
\end{equation}
Since the coefficient matrices $\boldsymbol{\Lambda}(\lambda,x)$ and $\mathbf{X}(\lambda,x)$ have zero trace, it then follows that $\det(\boldsymbol{\Psi}_\infty^\pm(\lambda,x))=1$.

\subsubsection{Solution near zero}
When $\lambda$ is small and $\arg(\ii x\lambda)$ varies from $-\pi$ to $\pi$, $\arg(-\xi)$ varies from $\tfrac{1}{3}\pi$ to $-\tfrac{1}{3}\pi$, so to avoid Stokes phenomenon and maintain exponential dichotomy we use scalar multiples of the basis elements $\mathrm{Ai}(\ee^{2\pi\ii/3}\xi)$ and $\mathrm{Ai}(\ee^{-2\pi\ii/3}\xi)$, so for constants $c^\pm$ we use \eqref{eq:diagonal-terms} and hence define a solution in the domain $D_0$ characterized by $|\lambda|<1$ and $\arg(\ii x\lambda)\in (-\pi,\pi)$ by
\begin{multline}
\boldsymbol{\Psi}_0(\lambda,x):=\ee^{-3x^{2/3}\sigma_3/2}\mathbf{G}^{-1}\frac{1}{\sqrt{2}}\begin{bmatrix}1&0\\0&-(\frac{2}{3})^{1/3}\end{bmatrix}(\ii x\lambda)^{\sigma_3/6}\\
{}\cdot\begin{bmatrix}\mathrm{Ai}(\ee^{2\pi\ii/3}\xi) & \mathrm{Ai}(\ee^{-2\pi\ii/3}\xi)\\
\ee^{2\pi\ii/3}\mathrm{Ai}'(\ee^{2\pi\ii/3}\xi) & \ee^{-2\pi\ii/3}\mathrm{Ai}'(\ee^{-2\pi\ii/3}\xi)\end{bmatrix}\begin{bmatrix}c^+ & 0\\0 & c^-\end{bmatrix},
\end{multline}  
which can be equivalently written as
\begin{multline}
\boldsymbol{\Psi}_0(\lambda,x)=\ee^{-3x^{2/3}\sigma_3/2}\begin{bmatrix}\frac{1}{2}&\frac{1}{2}\\-\frac{1}{2} & \frac{1}{2}\end{bmatrix}\\
{}\cdot\begin{bmatrix}c^+(\ii x\lambda)^{1/6}\mathrm{Ai}(\ee^{2\pi\ii/3}\xi) & c^-(\ii x\lambda)^{1/6}\mathrm{Ai}(\ee^{-2\pi\ii/3}\xi)\\
-c^+(\frac{2}{3})^{1/3}(\ii x\lambda)^{-1/6}\ee^{2\pi\ii/3}\mathrm{Ai}'(\ee^{2\pi\ii/3}\xi) & 
-c^-(\frac{2}{3})^{1/3}(\ii x\lambda)^{-1/6}\ee^{-2\pi\ii/3}\mathrm{Ai}'(\ee^{-2\pi\ii/3}\xi)\end{bmatrix}.
\end{multline}
Using $\ee^{\pm 2\pi\ii/3}\xi=x^{2/3}(\frac{3}{2})^{2/3}\ee^{\mp\ii\pi/3}(\ii x\lambda)^{-1/3}(1-x^{-2/3}(\ii x\lambda))$, the relevant expansions we need in this situation are then
\begin{equation}
(\ii x\lambda)^{1/6}\mathrm{Ai}(\ee^{\pm 2\pi\ii/3}\xi)
=\ee^{\pm\ii\pi/12}\frac{1}{2\sqrt{\pi}}x^{-1/6}(\tfrac{3}{2})^{-1/6}(\ii x\lambda)^{1/4}\ee^{\pm\ii x(\ii x\lambda)^{-1/2}}(1+\mathcal{O}(\lambda^{1/2})),
\end{equation}
and
\begin{equation}
-(\tfrac{2}{3})^{1/3}(\ii x\lambda)^{-1/6}\ee^{\pm 2\pi\ii/3}\mathrm{Ai}'(\ee^{\pm 2\pi i/3}\xi)
=\ee^{\pm 7\pi \ii/12}\frac{1}{2\sqrt{\pi}}x^{1/6}(\tfrac{3}{2})^{-1/6}(\ii x\lambda)^{-1/4}\ee^{\pm\ii x(\ii x\lambda)^{-1/2}}(1+\mathcal{O}(\lambda^{1/2})),
\end{equation}
as $\lambda\to 0$ with $-\pi<\arg(\ii x\lambda)<\pi$.  It follows that $\det(\boldsymbol{\Psi}_0(\lambda,x))=\tfrac{1}{4\pi\ii}c^+c^-(\tfrac{2}{3})^{1/3} + \mathcal{O}(\lambda^{1/2})$ in the same limit, so by Abel's theorem $\det(\boldsymbol{\Psi}_0(\lambda,x))=\tfrac{1}{4\pi\ii}c^+c^-(\tfrac{2}{3})^{1/3}$.  We complete the definition of $\boldsymbol{\Psi}_0(\lambda,x)$ by choosing 
\begin{equation}
c^+:=2\ii\sqrt{\pi}(\tfrac{3}{2})^{1/6}\quad\text{and}\quad c^-:=2\sqrt{\pi}(\tfrac{3}{2})^{1/6}
\end{equation}
which guarantees the identity $\det(\boldsymbol{\Psi}_0(\lambda,x))=1$.  It then follows that
\begin{multline}
\boldsymbol{\Psi}_0(\lambda,x)\ee^{-\ii x(\ii x\lambda)^{-1/2}\sigma_3} = \ee^{-3x^{2/3}\sigma_3/2}\begin{bmatrix}\frac{1}{2} & \frac{1}{2}\\-\frac{1}{2} & \frac{1}{2}\end{bmatrix}x^{\sigma_3/3}\ee^{7\pi\ii\sigma_3/12}\\
{}\cdot\left(\frac{\ii\lambda}{x}\right)^{\sigma_3/4}\begin{bmatrix}1+\mathcal{O}(\lambda^{1/2}) & 
\ee^{-2\pi\ii/3}(1+\mathcal{O}(\lambda^{1/2}))\\
\ee^{-\ii\pi/3}(1+\mathcal{O}(\lambda^{1/2})) & 
1+\mathcal{O}(\lambda^{1/2})
\end{bmatrix},\quad\lambda\to 0.
\label{eq:Psi0-expand}
\end{multline}
Here, the symbol $\mathcal{O}(\lambda^{1/2})$ denotes in each case a (different) quantity having an asymptotic expansion as $\lambda\to 0$ in positive integer powers of $\lambda^{1/2}$.  This implies the existence of the limit
\begin{equation}
\mathbf{B}_0(x):=\lim_{\lambda\to 0}\boldsymbol{\Psi}_0(\lambda,x)\ee^{-\ii x(\ii x\lambda)^{-1/2}\sigma_3}\mathbf{E}\cdot\left(\frac{\ii\lambda}{x}\right)^{-\sigma_3/4},\quad \mathbf{E}:=\frac{1}{\sqrt{2}}\begin{bmatrix}1 & \ee^{\ii\pi/3}\\\ee^{2\pi\ii/3} & 1\end{bmatrix}.
\label{eq:H-define}
\end{equation}
At this point, we can revisit the expansion of $\boldsymbol{\Psi}_0(\lambda,x)\ee^{-\ii x(\ii x\lambda)^{-1/2}\sigma_3}\mathbf{E}\cdot(\ii\lambda/x)^{-\sigma_3/4}$ as $\lambda\to 0$, and take advantage of the structure of the complete asymptotic expansions \cite[Eqns.\@ 9.7.5--6]{DLMF} of which the parentheses in \eqref{eq:Airy-first-term} represent just the first explicit terms to deduce that the half-integer powers all vanish, and in fact we have a full expansion
\begin{equation}
\boldsymbol{\Psi}_0(\lambda,x)\ee^{-\ii x(\ii x\lambda)^{-1/2}\sigma_3}\mathbf{E}\cdot\left(\frac{\ii\lambda}{x}\right)^{-\sigma_3/4}\sim\sum_{p=0}^\infty\left(\frac{\ii\lambda}{x}\right)^p\mathbf{B}_p(x),\quad \lambda\to 0.
\end{equation}

\subsection{Jump conditions for the seed}
We define (see Figure~\ref{fig:DomainsAndContours-xpos})
\begin{figure}[h]
\begin{center}
\includegraphics{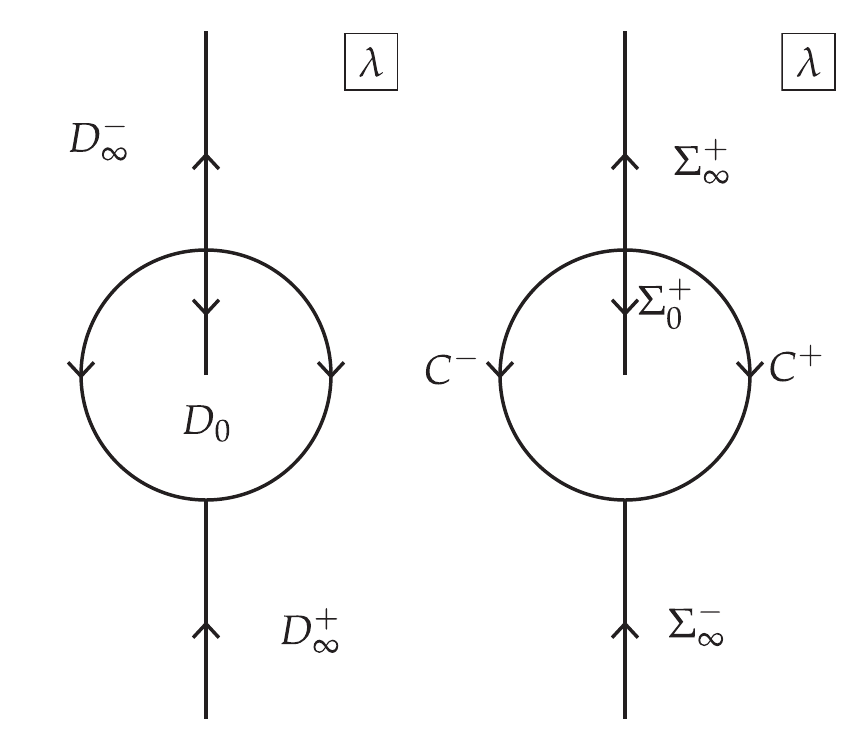}
\end{center}
\caption{Left:  the domains $D_\infty^\pm$ and $D_0$ in the $\lambda$-plane.  Right:  the arcs of the jump contour for $\boldsymbol{\Psi}(\lambda,x)$.  }
\label{fig:DomainsAndContours-xpos}
\end{figure}
\begin{equation}
\boldsymbol{\Psi}(\lambda,x):=\begin{cases}\boldsymbol{\Psi}_\infty^+(\lambda,x),&\quad\lambda\in D_\infty^+,\\
\boldsymbol{\Psi}_\infty^-(\lambda,x),&\quad\lambda\in D_\infty^-,\\
\boldsymbol{\Psi}_0(\lambda,x),&\quad\lambda\in D_0.
\end{cases}
\label{eq:Psi-seed}
\end{equation}
The jump matrices are obtained using the connection identity \cite[Eqn.\@ 9.2.12]{DLMF} $\mathrm{Ai}(z)+\ee^{2\pi\ii/3}\mathrm{Ai}(\ee^{2\pi\ii/3}z) +\ee^{-2\pi\ii/3}\mathrm{Ai}(\ee^{-2\pi\ii/3}z)=0$ and its derivative $\mathrm{Ai}'(z)+\ee^{-2\pi\ii/3}\mathrm{Ai}'(\ee^{2\pi\ii/3}z)+\ee^{2\pi\ii/3}\mathrm{Ai}'(\ee^{-2\pi\ii/3}z)=0$.  We will calculate the jump $\boldsymbol{\Psi}_+(\lambda,x)=\boldsymbol{\Psi}_-(\lambda,x)\mathbf{V}$ across each arc of the jump contour as indicated in the right-hand diagram of Figure~\ref{fig:DomainsAndContours-xpos} (subscripts $+$/$-$ indicate boundary value from the left/right by orientation).  Note that since in each domain $\boldsymbol{\Psi}(\lambda,x)$ is a simultaneous fundamental solution matrix with unit determinant for the equations $\boldsymbol{\Psi}_x=\mathbf{X}\boldsymbol{\Psi}$ and $\boldsymbol{\Psi}_\lambda=\boldsymbol{\Lambda\Psi}$, each jump matrix $\mathbf{V}$ is a constant matrix with $\det(\mathbf{V})=1$.   

To compute the jump of $\boldsymbol{\Psi}(\lambda,x)$ across $\Sigma_\infty^-$ as well as $C^\pm$, we note that $\xi$ and $(\ii x\lambda)^{\sigma_3/6}$ are both well-defined on these contours.  So the jump matrix $\mathbf{V}=\boldsymbol{\Psi}_-(\lambda,x)^{-1}\boldsymbol{\Psi}_+(\lambda,x)=\boldsymbol{\Psi}_\infty^+(\lambda,x)^{-1}\boldsymbol{\Psi}_\infty^-(\lambda,x)$ on $\Sigma_\infty^-$ satisfies
\begin{equation}
\begin{bmatrix}\mathrm{Ai}(\xi) & -\ee^{\ii\pi/6}\mathrm{Ai}(\ee^{2\pi\ii/3}\xi)\\
\mathrm{Ai}'(\xi) & -\ee^{\ii\pi/6}\ee^{2\pi\ii/3}\mathrm{Ai}'(\ee^{2\pi\ii/3}\xi)\end{bmatrix}=\begin{bmatrix}\mathrm{Ai}(\xi) & -\ee^{-\ii\pi/6}\mathrm{Ai}(\ee^{-2\pi\ii/3}\xi)\\
\mathrm{Ai}'(\xi) & -\ee^{-\ii\pi/6}\ee^{-2\pi\ii/3}\mathrm{Ai}'(\ee^{-2\pi\ii/3}\xi)\end{bmatrix}\mathbf{V}.
\end{equation}
Using the connection formul\ae\ to write the second column on the left-hand side in terms of $\mathrm{Ai}(\xi)$ and $\mathrm{Ai}'(\ee^{-2\pi\ii/3}\xi)$ and their derivatives yields
\begin{equation}
\mathbf{V}=\begin{bmatrix}1&-\ii\\0 & 1\end{bmatrix},\quad\lambda\in \Sigma_\infty^-.
\label{eq:V-Sigma-infty-minus}
\end{equation}
Similarly, the jump matrix on $C^-$ is $\mathbf{V}=\boldsymbol{\Psi}_-(\lambda,x)^{-1}\boldsymbol{\Psi}_+(\lambda,x)=\boldsymbol{\Psi}_\infty^-(\lambda,x)^{-1}\boldsymbol{\Psi}_0(\lambda,x)$, which satisfies
\begin{equation}
\begin{bmatrix}\ii\mathrm{Ai}(\ee^{2\pi\ii/3}\xi) & \mathrm{Ai}(\ee^{-2\pi\ii/3}\xi)\\
\ii\ee^{2\pi\ii/3}\mathrm{Ai}'(\ee^{2\pi\ii/3}\xi) & \ee^{-2\pi\ii/3}\mathrm{Ai}(\ee^{-2\pi\ii/3}\xi)\end{bmatrix}=
\begin{bmatrix}\mathrm{Ai}(\xi) & -\ee^{\ii\pi/6}\mathrm{Ai}(\ee^{2\pi\ii/3}\xi)\\
\mathrm{Ai}'(\xi) & -\ee^{\ii\pi/6}\ee^{2\pi\ii/3}\mathrm{Ai}'(\ee^{2\pi\ii/3}\xi)\end{bmatrix}\mathbf{V}.
\end{equation}
The same connection formul\ae\ applied to the second column of the left-hand side then yields
\begin{equation}
\mathbf{V}=\begin{bmatrix}0 & \ee^{-\ii\pi/3}\\ \ee^{-2\pi\ii/3} & \ee^{-5\ii\pi/6}\end{bmatrix},\quad \lambda\in C^-.
\label{eq:V-C-minus}
\end{equation}
Likewise, the jump matrix on $C^+$ is $\mathbf{V}=\boldsymbol{\Psi}_-(\lambda,x)^{-1}\boldsymbol{\Psi}_+(\lambda,x)=\boldsymbol{\Psi}_0(\lambda,x)^{-1}\boldsymbol{\Psi}_\infty^+(\lambda,x)$, which satisfies
\begin{equation}
\begin{bmatrix}\mathrm{Ai}(\xi) & -\ee^{-\ii\pi/6}\mathrm{Ai}(\ee^{-2\pi\ii/3}\xi)\\
\mathrm{Ai}'(\xi)&-\ee^{-\ii\pi/6}\ee^{-2\pi\ii/3}\mathrm{Ai}'(\ee^{-2\pi\ii/3}\xi)\end{bmatrix}=
\begin{bmatrix}\ii\mathrm{Ai}(\ee^{2\pi\ii/3}\xi) & \mathrm{Ai}(\ee^{-2\pi\ii/3}\xi)\\
\ii\ee^{2\pi\ii/3}\mathrm{Ai}'(\ee^{2\pi\ii/3}\xi) & \ee^{-2\pi\ii/3}\mathrm{Ai}(\ee^{-2\pi\ii/3}\xi)\end{bmatrix}\mathbf{V}.
\end{equation}
Applying the connection formul\ae\ to the first column of the left-hand side yields
\begin{equation}
\mathbf{V}=\begin{bmatrix}\ee^{-5\pi\ii/6} & 0\\ \ee^{\ii\pi/3} & \ee^{5\pi\ii/6}\end{bmatrix},\quad \lambda\in C^+.
\label{eq:V-C-plus}
\end{equation}

To compute the jump matrices on $\Sigma_0^+$ and $\Sigma_\infty^+$ we also have to take into account the jump conditions satisfied by $\xi$ and $(\ii x\lambda)^{\sigma_3/6}$:
\begin{equation}
\left[(\ii x\lambda)^{\sigma_3/6}\right]_+ = \ee^{\ii\pi\sigma_3/3}\left[(\ii x\lambda)^{\sigma_3/6}\right]_-\quad\text{and}\quad
\xi_+ = \ee^{-2\pi\ii/3}\xi_-,\quad \lambda\in\Sigma_0^+
\label{eq:scalar-jumps-0}
\end{equation}
and
\begin{equation}
\left[(\ii x\lambda)^{\sigma_3/6}\right]_+ = \ee^{-\ii\pi\sigma_3/3}\left[(\ii x\lambda)^{\sigma_3/6}\right]_-\quad\text{and}\quad
\xi_+ = \ee^{2\pi\ii/3}\xi_-,\quad \lambda\in\Sigma_\infty^+.
\label{eq:scalar-jumps-infty}
\end{equation}
The difference between these formul\ae\ arises simply from the opposite 
orientation of $\Sigma_0^+$ and $\Sigma_\infty^+$.  So, on $\Sigma_0^+$, using 
\eqref{eq:scalar-jumps-0} on the left-hand side of 
$\boldsymbol{\Psi}_+(\lambda,x)=\boldsymbol{\Psi}_-(\lambda,x)\mathbf{V}$ 
(with $\boldsymbol{\Psi}(\lambda,x)=\boldsymbol{\Psi}_0(\lambda,x)$), we have 
\begin{equation}
\begin{bmatrix}\ee^{5\pi\ii/6}\mathrm{Ai}(\xi_-) & \ee^{\ii\pi/3}\mathrm{Ai}(\ee^{2\pi\ii/3}\xi_-)\\\ee^{5\pi\ii/6}\mathrm{Ai}'(\xi_-) & \ee^{\ii\pi/3}\ee^{2\pi\ii/3}\mathrm{Ai}'(\ee^{2\pi\ii/3}\xi_-)\end{bmatrix}=
\begin{bmatrix}\ii\mathrm{Ai}(\ee^{2\pi\ii/3}\xi_-) & \mathrm{Ai}(\ee^{-2\pi\ii/3}\xi_-)\\\ii\ee^{2\pi\ii/3}\mathrm{Ai}'(\ee^{2\pi\ii/3}\xi_-) & \ee^{-2\pi\ii/3}\mathrm{Ai}'(\ee^{-2\pi\ii/3}\xi_-)\end{bmatrix}\mathbf{V}.
\end{equation}
Applying the connection formul\ae\ to the first column of the left-hand side gives
\begin{equation}
\mathbf{V}=\begin{bmatrix}1 & \ee^{-\ii\pi/6}\\ \ee^{-5\pi\ii/6} & 0
\end{bmatrix},\quad\lambda\in\Sigma_0^+.
\label{eq:V-Sigma-0}
\end{equation}
Finally, on $\Sigma_\infty^+$, we use \eqref{eq:scalar-jumps-infty} on the left-hand side of $\boldsymbol{\Psi}_+(\lambda,x)=\boldsymbol{\Psi}_-(\lambda,x)\mathbf{V}$ (this time with $\boldsymbol{\Psi}_+(\lambda,x)=\boldsymbol{\Psi}_{\infty,+}^-(\lambda,x)$ and $\boldsymbol{\Psi}_-(\lambda,x)=\boldsymbol{\Psi}_{\infty,-}^+(\lambda,x)$), we find
\begin{equation}
\begin{bmatrix}\ee^{-\ii\pi/3}\mathrm{Ai}(\ee^{2\pi\ii/3}\xi_-) & \ee^{5\pi\ii/6}\mathrm{Ai}(\ee^{-2\pi\ii/3}\xi_-)\\
\ee^{-\ii\pi/3}\ee^{2\pi\ii/3}\mathrm{Ai}'(\ee^{2\pi\ii/3}\xi_-)&\ee^{5\pi\ii/6}\ee^{-2\pi\ii/3}\mathrm{Ai}'(\ee^{-2\pi\ii/3}\xi_-)\end{bmatrix}=
\begin{bmatrix}\mathrm{Ai}(\xi_-) & \ee^{5\pi\ii/6}\mathrm{Ai}(\ee^{-2\pi\ii/3}\xi_-)\\
\mathrm{Ai}'(\xi_-) & \ee^{5\pi\ii/6}\ee^{-2\pi\ii/3}\mathrm{Ai}'(\ee^{-2\pi\ii/3}\xi_-)\end{bmatrix}\mathbf{V}.
\end{equation}
Using the connection formul\ae\ on the first column of the left-hand side gives
\begin{equation}
\mathbf{V}=\begin{bmatrix}1 & 0\\\ii & 1\end{bmatrix},\quad \lambda\in\Sigma_\infty^+.
\label{eq:V-Sigma-infty-plus}
\end{equation}

This completes the computation of the jump matrices for $\boldsymbol{\Psi}(\lambda,x)$.  It is straightforward to verify that the cyclic product of jump matrices about either of the two self-intersection points of the jump contour equals $\mathbb{I}$.

\section{Riemann-Hilbert representation for the algebraic solutions}
\label{sec:RHPs}

\subsection{Riemann-Hilbert problem formulation and basic properties}
The  algebraic solutions of \eqref{eq:D7} with $\alpha=8$, $\beta=2n\in2\mathbb{Z}$, and $\delta=-1$ (or \eqref{eq:D7KV} with $\epsilon=-1$, $a=-\ii n$, and $b=\ii$) can be obtained from the following Riemann-Hilbert problem.
\begin{rhp}[Algebraic solutions of Painlev\'e-III D7]
Given $x>0$ and $n\in\mathbb{Z}$, seek a $2\times 2$ matrix function $\lambda\mapsto \mathbf{W}^{(n)}(\lambda,x)$ with the following properties:
\begin{itemize}
\item Analyticity:  $\mathbf{W}^{(n)}(\lambda,x)$ is analytic for $\lambda\in\mathbb{C}\setminus(\Sigma_\infty^-\cup\Sigma_\infty^+\cup\Sigma_0^+\cup C^+\cup C^-)$.
\item Jump conditions:  $\mathbf{W}^{(n)}(\lambda,x)$ takes continuous boundary values from each component of its domain of analyticity, and the boundary values are related by 
\begin{equation}
\mathbf{W}^{(n)}_+(\lambda,x)=\mathbf{W}^{(n)}_-(\lambda,x)\ee^{-\ii (x\lambda-x(\ii x\lambda_-)^{-1/2})\sigma_3}(-1)^n\mathbf{V}\ee^{\ii(x\lambda-x(\ii x\lambda_+)^{-1/2})\sigma_3},\quad\lambda\in \Sigma_0^+\cup\Sigma_\infty^+,
\end{equation}
and
\begin{equation}
\mathbf{W}^{(n)}_+(\lambda,x)=\mathbf{W}^{(n)}_-(\lambda,x)\ee^{-\ii (x\lambda-x(\ii x\lambda_-)^{-1/2})\sigma_3}\mathbf{V}\ee^{\ii(x\lambda-x(\ii x\lambda_+)^{-1/2})\sigma_3},\quad\lambda\in \Sigma_\infty^-\cup C^+\cup C^-,
\end{equation}
where the constant matrix $\mathbf{V}$ is defined on each arc of the jump contour by \eqref{eq:V-Sigma-infty-minus}, \eqref{eq:V-C-minus}, \eqref{eq:V-C-plus}, \eqref{eq:V-Sigma-0}, and \eqref{eq:V-Sigma-infty-plus}.
\item Normalization:  $\mathbf{W}^{(n)}(\lambda,x)\left(\frac{\ii\lambda}{x}\right)^{n\sigma_3/2}\to\mathbb{I}$ as $\lambda\to\infty$.
\item Behavior at the origin:  the limit
\begin{equation}
\mathbf{B}^{(n)}_0(x):=\lim_{\lambda\to 0}\mathbf{W}^{(n)}(\lambda,x)\mathbf{E}\cdot \left(\frac{\ii\lambda}{x}\right)^{-(-1)^n\sigma_3/4}
\label{eq:limit-at-zero}
\end{equation}
exists, where $\mathbf{E}$ is given by \eqref{eq:H-define}.  
\end{itemize}
\label{rhp:algebraic}
\end{rhp}
The following lemma is a consequence of the construction given in Section~\ref{sec:direct-monodromy-seed}.
\begin{lemma}
The matrix $\mathbf{W}^{(0)}(\lambda,x)$ defined in terms of $\mathbf{\Psi}(\lambda,x)$ (see \eqref{eq:Psi-seed}) by
\begin{equation}
\mathbf{W}^{(0)}(\lambda,x):= \boldsymbol{\Psi}(\lambda,x)\ee^{\ii (x\lambda-x(\ii x\lambda)^{-1/2})\sigma_3}
\label{eq:Y-Psi}
\end{equation}
is a solution of Riemann-Hilbert Problem~\ref{rhp:algebraic} with $n=0$.
\label{lem:seed-solve}
\end{lemma}

\begin{lemma}
Given $x>0$ and $n\in\mathbb{Z}$, there is at most one solution of Riemann-Hilbert Problem~\ref{rhp:algebraic} and any solution has unit determinant.
\label{lem:uniqueness}
\end{lemma}
\begin{proof}
Suppose there exist two solutions, denoted $\mathbf{W}(\lambda)$ and $\widetilde{\mathbf{W}}(\lambda)$.  It follows from the conditions of the problem and Liouville's theorem that $\det(\mathbf{W}(\lambda))=\det(\widetilde{\mathbf{W}}(\lambda))\equiv 1$.  Consider the matrix $\mathbf{R}(\lambda):=\mathbf{W}(\lambda)\widetilde{\mathbf{W}}(\lambda)^{-1}$.  It similarly follows from the conditions of the problem that $\mathbf{R}(\lambda)$ is entire and tends to the identity matrix as $\lambda\to\infty$, so by Liouville's theorem again, $\mathbf{R}(\lambda)\equiv\mathbb{I}$.
\end{proof}

\begin{lemma}
Given $x>0$ and $n\in\mathbb{Z}$, if the solution $\mathbf{W}^{(n)}(\lambda,x)$ of Riemann-Hilbert Problem~\ref{rhp:algebraic} exists, then the normalization condition holds in the stronger sense that there exist coefficient matrices $\{\mathbf{A}_p^{(n)}(x)\}_{p=1}^\infty$ such that the complete asymptotic expansion
\begin{equation}
\mathbf{W}^{(n)}(\lambda,x)\ee^{\ii x(\ii x\lambda)^{-1/2}\sigma_3}\left(\frac{\ii\lambda}{x}\right)^{n\sigma_3/2}\sim\mathbb{I}+\sum_{p=1}^\infty \left(\frac{\ii\lambda}{x}\right)^{-p}\mathbf{A}_p^{(n)}(x),\quad \lambda\to\infty
\label{eq:lambda-infinity-stronger}
\end{equation}
is uniformly valid with respect to $\arg(\ii\lambda)$, and the expansion is differentiable term-by-term with respect to $\lambda$ and $x$.  Also, the limit in \eqref{eq:limit-at-zero} is the leading term in another complete asymptotic expansion involving other coefficient matrices $\{\mathbf{B}_p^{(n)}(x)\}_{p=0}^\infty$:
\begin{equation}
\mathbf{W}^{(n)}(\lambda,x)\ee^{-\ii x\lambda\sigma_3}\mathbf{E}\cdot\left(\frac{\ii\lambda}{x}\right)^{-(-1)^n\sigma_3/4}\sim\sum_{p=0}^\infty\left(\frac{\ii\lambda}{x}\right)^p\mathbf{B}_p^{(n)}(x),\quad\lambda\to 0
\label{eq:lambda-0-stronger}
\end{equation}
holding uniformly with respect to $\arg(\ii\lambda)$ and enjoying similar differentiability properties.
\label{lem:expansions}
\end{lemma}
\begin{proof}
We consider the function $\lambda\mapsto\mathbf{S}(\lambda)$ defined by 
\begin{equation}
\mathbf{S}^{(n)}(\lambda):=\begin{cases}\displaystyle \mathbf{W}^{(n)}(\lambda,x)\ee^{\ii x(\ii x\lambda)^{-1/2}\sigma_3}\left(\frac{\ii\lambda}{x}\right)^{n\sigma_3/2},&\quad |\lambda|>1,\\
\displaystyle \mathbf{W}^{(n)}(\lambda,x)\ee^{-\ii x\lambda\sigma_3}\mathbf{E}\cdot\left(\frac{\ii\lambda}{x}\right)^{-(-1)^n\sigma_3/4},&\quad |\lambda|<1.
\end{cases}
\label{eq:Sn-define}
\end{equation}
Then, $\mathbf{S}^{(n)}(\lambda)$ satisfies the conditions of an equivalent Riemann-Hilbert problem, and it is easy to check that the jump matrices for $\mathbf{S}^{(n)}(\lambda)$ tend exponentially rapidly to the identity as $\lambda\to\infty$ in $\Sigma_\infty^+\cup\Sigma_\infty^-$ and as $\lambda\to 0$ in $\Sigma_0^+$.  This kind of problem is amenable to analytic (with respect to $x$) Fredholm theory applied to an equivalent singular integral equation, and the result that $\mathbf{S}^{(n)}(\lambda)$ has asymptotic power series expansions as $\lambda\to \infty$ and $\lambda\to 0$ then follows.
\end{proof}
In the next two subsections we show how and why Riemann-Hilbert Problem~\ref{rhp:algebraic} encodes the algebraic solutions $u=u_n(x)$ of \eqref{eq:D7KV}.
\subsection{Differential equations satisfied by the solution of Riemann-Hilbert Problem~\ref{rhp:algebraic}}
\label{sec:diff-eqs}
Suppose $\mathbf{W}^{(n)}(\lambda,x)$ solves Riemann-Hilbert Problem~\ref{rhp:algebraic}, and consider the related matrix 
\begin{equation}
\boldsymbol{\Psi}^{(n)}(\lambda,x):=\mathbf{W}^{(n)}(\lambda,x)\ee^{-\ii(x\lambda-x(\ii x\lambda)^{-1/2})\sigma_3}
\end{equation}
(cf.\ \eqref{eq:Y-Psi}).  It is straightforward to check that $\boldsymbol{\Psi}^{(n)}(\lambda,x)$ satisfies jump conditions across each arc of the jump contour that are independent of both $x$ and $\lambda$:  $\boldsymbol{\Psi}^{(n)}_+(\lambda,x)=\boldsymbol{\Psi}^{(n)}_-(\lambda,x)(-1)^n\mathbf{V}$ for $\lambda\in\Sigma_0^+\cup\Sigma_\infty^+$ and $\boldsymbol{\Psi}^{(n)}_+(\lambda,x)=\boldsymbol{\Psi}^{(n)}_-(\lambda,x)\mathbf{V}$ for $\lambda\in\Sigma_\infty^-\cup C^+\cup C^-$.  It follows that the matrices
\begin{equation}
\boldsymbol{\Lambda}^{(n)}(\lambda,x):=\frac{\partial\boldsymbol{\Psi}^{(n)}}{\partial \lambda}(\lambda,x)\boldsymbol{\Psi}^{(n)}(\lambda,x)^{-1}\quad\text{and}\quad
\mathbf{X}^{(n)}(\lambda,x):=\frac{\partial\boldsymbol{\Psi}^{(n)}}{\partial x}(\lambda,x)\boldsymbol{\Psi}^{(n)}(\lambda,x)^{-1}
\end{equation}
are both analytic for $\lambda\in\mathbb{C}\setminus\{0\}$.  They can be expressed in terms of the coefficients $\{\mathbf{A}_p^{(n)}(x)\}_{p=1}^\infty$ and $\{\mathbf{B}_p^{(n)}(x)\}_{p=0}^\infty$ of Lemma~\ref{lem:expansions} as follows.  First, we write $\boldsymbol{\Psi}^{(n)}(\lambda,x)$ in terms of $\mathbf{W}^{(n)}(\lambda,x)$ and then the matrix $\mathbf{S}^{(n)}(\lambda,x)$ defined by \eqref{eq:Sn-define}.  Then, assuming that $|\lambda|>1$, we have
\begin{equation}
\begin{split}
\boldsymbol{\Lambda}^{(n)}(\lambda,x)&=\frac{\partial\mathbf{S}^{(n)}}{\partial\lambda}(\lambda,x)\mathbf{S}^{(n)}(\lambda,x)^{-1} -\frac{n}{2\lambda}\mathbf{S}^{(n)}(\lambda,x)\sigma_3\mathbf{S}^{(n)}(\lambda,x)^{-1}-\ii x\mathbf{S}^{(n)}(\lambda,x)\sigma_3\mathbf{S}^{(n)}(\lambda,x)^{-1}\\
\mathbf{X}^{(n)}(\lambda,x)&=\frac{\partial\mathbf{S}^{(n)}}{\partial x}(\lambda,x)\mathbf{S}^{(n)}(\lambda,x)^{-1} +\frac{n}{2x}\mathbf{S}^{(n)}(\lambda,x)\sigma_3\mathbf{S}^{(n)}(\lambda,x)^{-1} -\ii\lambda\mathbf{S}^{(n)}(\lambda,x)\sigma_3\mathbf{S}^{(n)}(\lambda,x)^{-1}.
\end{split}
\end{equation}
Using \eqref{eq:lambda-infinity-stronger} from Lemma~\ref{lem:expansions} (the left-hand side of which is exactly $\mathbf{S}^{(n)}(\lambda,x)$ for $|\lambda|>1$) then gives
\begin{equation}
\begin{split}
\boldsymbol{\Lambda}^{(n)}(\lambda,x)&=-\ii x\sigma_3 -\frac{n}{2\lambda}\sigma_3-\frac{x^2}{\lambda}[\mathbf{A}^{(n)}_1(x),\sigma_3] + \mathcal{O}(\lambda^{-2}),\quad\lambda\to\infty,\quad\text{and}\\
\mathbf{X}^{(n)}(\lambda,x)&=-\ii \lambda\sigma_3 +\frac{n}{2x}\sigma_3 -x[\mathbf{A}^{(n)}_1(x),\sigma_3] + \mathcal{O}(\lambda^{-1}),\quad\lambda\to\infty.
\end{split}
\label{eq:LambdaXinfinity}
\end{equation}

To analyze the same matrices in the limit $\lambda\to 0$, we use \eqref{eq:Sn-define} to obtain, for $|\lambda|<1$,
\begin{equation}
\begin{split}
\boldsymbol{\Lambda}^{(n)}(\lambda,x)&=\frac{\partial\mathbf{S}^{(n)}}{\partial\lambda}(\lambda,x)\mathbf{S}^{(n)}(\lambda,x)^{-1} +\frac{(-1)^n}{4\lambda}\mathbf{S}^{(n)}(\lambda,x)\sigma_3\mathbf{S}^{(n)}(\lambda,x)^{-1}\\
&\qquad\qquad{}+\frac{1}{2x}\left(\frac{\ii\lambda}{x}\right)^{-3/2}\mathbf{S}^{(n)}(\lambda,x)
\left(\frac{\ii\lambda}{x}\right)^{(-1)^n\sigma_3/4}\mathbf{E}^{-1}\sigma_3\mathbf{E}\left(\frac{\ii\lambda}{x}\right)^{-(-1)^n\sigma_3/4}\mathbf{S}^{(n)}(\lambda,x)^{-1}\\
\mathbf{X}^{(n)}(\lambda,x)&=\frac{\partial\mathbf{S}^{(n)}}{\partial x}(\lambda,x)\mathbf{S}^{(n)}(\lambda,x)^{-1}-\frac{(-1)^n}{4x}\mathbf{S}^{(n)}(\lambda,x)\sigma_3\mathbf{S}^{(n)}(\lambda,x)^{-1}\\
&\qquad\qquad{}+\frac{\ii}{2x}\left(\frac{\ii\lambda}{x}\right)^{-1/2}\mathbf{S}^{(n)}(\lambda,x)\left(\frac{\ii\lambda}{x}\right)^{(-1)^n\sigma_3/4}\mathbf{E}^{-1}\sigma_3\mathbf{E}\left(\frac{\ii\lambda}{x}\right)^{-(-1)^n\sigma_3/4}\mathbf{S}^{(n)}(\lambda,x)^{-1}.
\end{split}
\end{equation}
Using \eqref{eq:lambda-0-stronger} from Lemma~\ref{lem:expansions}, the left-hand side of which is exactly $\mathbf{S}^{(n)}(\lambda,x)$ for $|\lambda|<1$, the terms on the first line of each expression are $\mathcal{O}(\lambda^{-1})$ and $\mathcal{O}(1)$, respectively, as $\lambda\to 0$.  For the terms on the second line of each, we use the identity
\begin{equation}
\mathbf{E}^{-1}\sigma_3\mathbf{E} = \begin{bmatrix}0 & \ee^{\ii\pi/3}\\\ee^{-\ii\pi/3} & 0\end{bmatrix}.
\end{equation}
Therefore, using \eqref{eq:lambda-0-stronger} from Lemma~\ref{lem:expansions} again shows that if $n$ is even,
\begin{equation}
\begin{split}
\boldsymbol{\Lambda}^{(n)}(\lambda,x)&=-\ee^{-\ii\pi/3}\frac{x}{2\lambda^2}\mathbf{B}^{(n)}_0(x)\begin{bmatrix}0&0\\1&0\end{bmatrix}\mathbf{B}^{(n)}_0(x)^{-1} + \mathcal{O}(\lambda^{-1}),\quad\lambda\to 0,\quad\text{and}\\
\mathbf{X}^{(n)}(\lambda,x)&=\ee^{-\ii\pi/3}\frac{1}{2\lambda}\mathbf{B}^{(n)}_0(x)\begin{bmatrix}0&0\\1&0\end{bmatrix}\mathbf{B}^{(n)}_0(x)^{-1} + \mathcal{O}(1),\quad\lambda\to 0,
\end{split}
\label{eq:LambdaX0-even}
\end{equation}
while if instead $n$ is odd,
\begin{equation}
\begin{split}
\boldsymbol{\Lambda}^{(n)}(\lambda,x)&=-\ee^{\ii\pi/3}\frac{x}{2\lambda^2}\mathbf{B}^{(n)}_0(x)\begin{bmatrix}0&1\\0&0\end{bmatrix}\mathbf{B}^{(n)}_0(x)^{-1} + \mathcal{O}(\lambda^{-1}),\quad\lambda\to 0,\quad\text{and}\\
\mathbf{X}^{(n)}(\lambda,x)&=\ee^{\ii\pi/3}\frac{1}{2\lambda}\mathbf{B}^{(n)}_0(x)\begin{bmatrix}0&1\\0&0\end{bmatrix}\mathbf{B}^{(n)}_0(x)^{-1} + \mathcal{O}(1),\quad\lambda\to 0.
\end{split}
\label{eq:LambdaX0-odd}
\end{equation}

The Laurent expansions \eqref{eq:LambdaXinfinity} and \eqref{eq:LambdaX0-even}--\eqref{eq:LambdaX0-odd} then fully determine the matrices $\boldsymbol{\Lambda}^{(n)}(\lambda,x)$ and $\mathbf{X}^{(n)}(\lambda,x)$:
\begin{equation}
\begin{split}
\boldsymbol{\Lambda}^{(n)}(\lambda,x)&=-\ii x\sigma_3-\frac{n}{2\lambda}\sigma_3-\frac{x}{\lambda}\mathbf{J}^{(n)}(x) +\frac{\ii x}{2\lambda^2}\mathbf{K}^{(n)}(x),\\
\mathbf{X}^{(n)}(\lambda,x)&=-\ii\lambda\sigma_3+\frac{n}{2x}\sigma_3-\mathbf{J}^{(n)}(x)-\frac{\ii}{2\lambda}\mathbf{K}^{(n)}(x),
\end{split}
\end{equation}
where
\begin{equation}
\mathbf{J}^{(n)}(x):=x[\mathbf{A}_1^{(n)}(x),\sigma_3]\quad\text{and}\quad\mathbf{K}^{(n)}(x):=\begin{cases}
\ee^{\ii\pi/6}\mathbf{B}_0^{(n)}(x)\begin{bmatrix}0&0\\1&0\end{bmatrix}\mathbf{B}_0^{(n)}(x)^{-1},&\quad \text{$n$ even},\\
\ee^{5\pi\ii/6}\mathbf{B}_0^{(n)}(x)\begin{bmatrix}0&1\\0&0\end{bmatrix}\mathbf{B}_0^{(n)}(x)^{-1},&\quad\text{$n$ odd}.
\end{cases}
\label{eq:M1nM2n}
\end{equation}
Since $\mathbf{J}^{(n)}(x)$ is an off-diagonal matrix and $\mathbf{K}^{(n)}(x)$ is singular and nondiagonalizable but nonzero, it follows that whenever $x$ is such that Riemann-Hilbert Problem~\ref{rhp:algebraic} is solvable, if potentials $u=u_n(x)$, $\ee^{\pm\ii\varphi}=\ee^{\pm\ii\varphi_n(x)}$, $p=p_n(x)$, and $q=q_n(x)$ are defined in terms of the matrix elements of $\mathbf{J}=\mathbf{J}^{(n)}(x)$ and $\mathbf{K}=\mathbf{K}^{(n)}(x)$ by \eqref{eq:potentials-from-M12} (taking $\epsilon=-1$ by our convention), then in particular $u_n(x)$ is a solution of the Painlev\'e-III D7 equation in the form \eqref{eq:D7KV} for $a=-\ii n$ and $b$ determined from \eqref{eq:phiprime}.  
Using $\det(\mathbf{B}_0^{(n)}(x))=1$ (arising by combining Lemmas~\ref{lem:uniqueness} and \ref{lem:expansions}), the formula for $u_n(x)$ in terms of the matrix $\mathbf{B}_0^{(n)}(x)$ obtained from $\mathbf{W}^{(n)}(\lambda,x)$ by \eqref{eq:lambda-0-stronger} from Lemma~\ref{lem:expansions} is:
\begin{equation}
u_n(x)=\begin{cases}\ee^{-5\pi\ii/6}xB_{0,12}^{(n)}(x)B_{0,22}^{(n)}(x),&\quad \text{$n$ even,}\\
\ee^{5\pi\ii/6}xB_{0,11}^{(n)}(x)B_{0,21}^{(n)}(x),&\quad\text{$n$ odd.}
\end{cases}
\label{eq:un-formula-Y}
\end{equation}

\subsection{Solution of Riemann-Hilbert Problem~\ref{rhp:algebraic} by Schlesinger transformations}
\label{sec:Schlesinger}
\subsubsection{Schlesinger transformations for the Painlev\'e-III (D7) Lax pair}
Schlesinger transformations for the Lax pair \eqref{eq:Lax} and their induced B\"acklund transformations are discussed in \cite[Section 6.1]{KitaevV04}.
Define a gauge transformation matrix $\mathbf{G}(\lambda,x)$ by
\begin{equation}
\mathbf{G}(\lambda,x):=(\ii\lambda)^{1/2}\overline{\mathbf{G}}(x)+(\ii\lambda)^{-1/2}\underline{\mathbf{G}}(x),
\label{eq:Gauge-form}
\end{equation}
where $\overline{\mathbf{G}}(x)$ and $\underline{\mathbf{G}}(x)$ are matrices to be determined so that when $\boldsymbol{\Psi}$ is a simultaneous fundamental solution matrix of the Lax pair equations \eqref{eq:Lax}, then $\widetilde{\boldsymbol{\Psi}}:=\mathbf{G}(\lambda,x)\boldsymbol{\Psi}$ is as well, but with a different value of $a$ and different potentials $u(x)$, $p(x)$, $q(x)$, and $\varphi(x)$.  We also want to normalize $\mathbf{G}(\lambda,x)$ so that 
$\det(\mathbf{G}(\lambda,x))\equiv 1$.  Clearly, $\widetilde{\boldsymbol{\Psi}}(\lambda,x)$ is a simultaneous fundamental solution matrix of $\widetilde{\boldsymbol{\Psi}}_\lambda = \widetilde{\boldsymbol{\Lambda}}(\lambda,x)\widetilde{\boldsymbol{\Psi}}$ and $\widetilde{\boldsymbol{\Psi}}_x = \widetilde{\mathbf{X}}(\lambda,x)\widetilde{\boldsymbol{\Psi}}$, where
\begin{equation}
\begin{split}
\widetilde{\boldsymbol{\Lambda}}(\lambda,x)&:=\widetilde{\boldsymbol{\Psi}}_\lambda(\lambda,x)\widetilde{\boldsymbol{\Psi}}(\lambda,x)^{-1}\\
& = \frac{\partial\mathbf{G}}{\partial\lambda}(\lambda,x)\mathbf{G}(\lambda,x)^{-1} + \mathbf{G}(\lambda,x)\boldsymbol{\Lambda}(\lambda,x)\mathbf{G}(\lambda,x)^{-1}\quad\text{and}\\
\widetilde{\mathbf{X}}(\lambda,x)&:=\widetilde{\boldsymbol{\Psi}}_x(\lambda,x)\widetilde{\boldsymbol{\Psi}}(\lambda,x)^{-1}\\
& = \frac{\partial\mathbf{G}}{\partial x}(\lambda,x)\mathbf{G}(\lambda,x)^{-1} + \mathbf{G}(\lambda,x)\mathbf{X}(\lambda,x)\mathbf{G}(\lambda,x)^{-1}.
\end{split}
\label{eq:gauge-Lambda-X}
\end{equation}
Writing $a=-\ii n$ for $n\in\mathbb{C}$ (in general), we want to pick the coefficients $\overline{\mathbf{G}}(x)$ and $\underline{\mathbf{G}}(x)$ so that the matrices
\begin{equation}
\boldsymbol{\Lambda}(\lambda,x):=-\ii x\sigma_3 -\frac{n}{2\lambda}\sigma_3-\frac{x}{\lambda}\mathbf{J}(x)+\frac{\ii x}{2\lambda^2}\mathbf{K}(x)\quad\text{and}\quad
\mathbf{X}(\lambda,x):=-\ii\lambda\sigma_3 +\frac{n}{2x}\sigma_3 -\mathbf{J}(x)-\frac{\ii}{2\lambda}\mathbf{K}(x),
\label{eq:Lambda-X-matrices-rewrite}
\end{equation}
where $\mathbf{J}(x)$ is off-diagonal and $\mathbf{K}(x)$ is non-diagonalizable with zero eigenvalues, are transformed into corresponding matrices
\begin{equation}
\widetilde{\boldsymbol{\Lambda}}(\lambda,x):=-\ii x\sigma_3 -\frac{\widetilde{n}}{2\lambda}\sigma_3-\frac{x}{\lambda}\widetilde{\mathbf{J}}(x)+\frac{\ii x}{2\lambda^2}\widetilde{\mathbf{K}}(x)\quad\text{and}\quad
\widetilde{\mathbf{X}}(\lambda,x):=-\ii\lambda\sigma_3 +\frac{\widetilde{n}}{2x}\sigma_3 -\widetilde{\mathbf{J}}(x)-\frac{\ii}{2\lambda}\widetilde{\mathbf{K}}(x),
\label{eq:tilde-Lambda-X-matrices}
\end{equation}
with $\widetilde{\mathbf{J}}(x)$ off-diagonal and $\widetilde{\mathbf{K}}(x)$ non-diagonalizable with zero eigenvalues, where $\widetilde{n}$ is another complex parameter.
Using \eqref{eq:Gauge-form}, \eqref{eq:Lambda-X-matrices-rewrite}, and \eqref{eq:tilde-Lambda-X-matrices}, we see that both sides of each of the equations (equivalent to \eqref{eq:gauge-Lambda-X}) 
\begin{equation}
\begin{split}
\lambda^2(\ii\lambda)^{1/2}\widetilde{\boldsymbol{\Lambda}}(\lambda,x)\mathbf{G}(\lambda,x) &= \lambda^2(\ii\lambda)^{1/2}\frac{\partial\mathbf{G}}{\partial\lambda}(\lambda,x) + \lambda^2(\ii\lambda)^{1/2}\mathbf{G}(\lambda,x)\boldsymbol{\Lambda}(\lambda,x)\\
\lambda(\ii\lambda)^{1/2}\widetilde{\mathbf{X}}(\lambda,x)\mathbf{G}(\lambda,x)&=\lambda(\ii\lambda)^{1/2}\frac{\partial\mathbf{G}}{\partial x}(\lambda,x)+\lambda(\ii\lambda)^{1/2}\mathbf{G}(\lambda,x)\mathbf{X}(\lambda,x)
\end{split}
\label{eq:polynomial-gauge-form}
\end{equation}
are cubic polynomials in $\lambda$.  The coefficients of $\lambda^3$ from both equations are balanced exactly when 
\begin{equation}
\sigma_3\overline{\mathbf{G}}(x)=\overline{\mathbf{G}}(x)\sigma_3,
\label{eq:cubic-terms}
\end{equation}
i.e.,  $\overline{\mathbf{G}}(x)$ is a diagonal matrix.  The coefficients of $\lambda^2$ give the equations
\begin{equation}
\begin{split}
-\frac{\ii\widetilde{n}}{2}\sigma_3\overline{\mathbf{G}}(x)-\ii x\widetilde{\mathbf{J}}(x)\overline{\mathbf{G}}(x)-\ii x\sigma_3\underline{\mathbf{G}}(x) &=\frac{1}{2}\ii\overline{\mathbf{G}}(x) -\frac{\ii n}{2}\overline{\mathbf{G}}(x)\sigma_3-\ii x\overline{\mathbf{G}}(x)\mathbf{J}(x) -\ii x\underline{\mathbf{G}}(x)\sigma_3\\
\frac{\ii\widetilde{n}}{2x}\sigma_3\overline{\mathbf{G}}(x)-\ii\widetilde{\mathbf{J}}(x)\overline{\mathbf{G}}(x)-\ii\sigma_3\underline{\mathbf{G}}(x)&=\ii\overline{\mathbf{G}}'(x) +
\frac{\ii n}{2x}\overline{\mathbf{G}}(x)\sigma_3 -\ii\overline{\mathbf{G}}(x)\mathbf{J}(x) -\ii\underline{\mathbf{G}}(x)\sigma_3.
\end{split}
\label{eq:quadratic-terms}
\end{equation}
Using the fact that $\mathbf{J}(x)$ and $\widetilde{\mathbf{J}}(x)$ are off-diagonal, the diagonal terms of these equations are equivalent to the equations
\begin{equation}
\begin{split}
(\widetilde{n}-n)\sigma_3\overline{\mathbf{G}}(x)&=-\overline{\mathbf{G}}(x)\\
-(\widetilde{n}-n)\sigma_3\overline{\mathbf{G}}(x)&=-2x\overline{\mathbf{G}}'(x).
\end{split}
\label{eq:quadratic-terms-diagonal}
\end{equation}
Adding these together yields $2x\overline{\mathbf{G}}'(x)=-\overline{\mathbf{G}}(x)$ which implies that $\overline{\mathbf{G}}(x)=x^{-1/2}\overline{\mathbf{G}}_0$ where $\overline{\mathbf{G}}_0$ is a constant diagonal matrix.  Then, for $x\neq 0$, the first of these equations is the algebraic relation $(\widetilde{n}-n)\sigma_3\overline{\mathbf{G}}_0 = -\overline{\mathbf{G}}_0$.  Observe that if both diagonal elements of $\overline{\mathbf{G}}_0$ are nonzero, we arrive at the contradiction that both $\widetilde{n}=n+1$ and $\widetilde{n}=n-1$.  Hence nontrivial solutions for $\overline{\mathbf{G}}(x)$ are:  either
\begin{equation}
\overline{\mathbf{G}}(x)=\overline{\mathbf{G}}^\uparrow(x):=\begin{bmatrix}0 & 0\\0 & x^{-1/2}\end{bmatrix}\quad\text{and}\quad\widetilde{n}=n+1
\end{equation}
or
\begin{equation}
\overline{\mathbf{G}}(x)=\overline{\mathbf{G}}^\downarrow(x):=\begin{bmatrix}x^{-1/2} & 0\\0 & 0\end{bmatrix}\quad\text{and}\quad \widetilde{n}=n-1.
\end{equation}
These are unique up to irrelevant scalings by nonzero constants.  

At this point, we enforce $\det(\mathbf{G}(\lambda,x))=1$, which implies two alternate forms for the gauge transformation $\mathbf{G}(\lambda,x)$:
\begin{equation}
\mathbf{G}(\lambda,x)=\mathbf{G}^\uparrow(\lambda,x)=\begin{bmatrix} x^{1/2}(\ii\lambda)^{-1/2}& B^\uparrow(x)(\ii\lambda)^{-1/2}\\C^\uparrow(x)(\ii\lambda)^{-1/2} & x^{-1/2}(\ii\lambda)^{1/2}+x^{-1/2}B^\uparrow(x)C^\uparrow(x)(\ii\lambda)^{-1/2}\end{bmatrix}\quad\text{and}\quad \widetilde{n}=n+1
\label{eq:G-up-coeffs}
\end{equation}
and 
\begin{equation}
\mathbf{G}(\lambda,x)=\mathbf{G}^\downarrow(\lambda,x)=\begin{bmatrix}x^{-1/2}(\ii\lambda)^{1/2} + x^{-1/2}B^\downarrow(x)C^\downarrow(x)(\ii\lambda)^{-1/2} & B^\downarrow(x)(\ii\lambda)^{-1/2}\\
C^\downarrow(x)(\ii\lambda)^{-1/2} & x^{1/2}(\ii\lambda)^{-1/2}\end{bmatrix}\quad\text{and}\quad\widetilde{n}=n-1.
\label{eq:G-down-coeffs}
\end{equation}
Using $\det(\mathbf{G}(\lambda,x))=1$, we now solve \eqref{eq:polynomial-gauge-form} for $\widetilde{\boldsymbol{\Lambda}}(\lambda,x)$, which has the form of a Laurent polynomial in $\lambda$ involving powers ranging from $\lambda^0$ through $\lambda^{-3}$.  In order that the coefficient of $\lambda^0$ is exactly $-\ii x\sigma_3$ as required by the form \eqref{eq:tilde-Lambda-X-matrices}, it is necessary to set
\begin{equation}
C^\uparrow(x):=\frac{x^{1/2}q(x)}{8u(x)}\quad\text{and}\quad B^\downarrow(x):=-\frac{x^{1/2}p(x)}{8u(x)}.
\label{eq:SchlesingerCoeffs1}
\end{equation}
In order that the coefficient of $\lambda^{-3}$ vanishes as required by the form \eqref{eq:tilde-Lambda-X-matrices} we must then set
\begin{equation}
B^\uparrow(x)=-x^{1/2}\ee^{\ii\varphi(x)}\quad\text{and}\quad C^\downarrow(x)=-x^{1/2}\ee^{-\ii\varphi(x)}.
\label{eq:SchlesingerCoeffs2}
\end{equation}
These relations use our convention of $\epsilon=-1$ in the parametrization of the matrices $\mathbf{J}(x)$ and $\mathbf{K}(x)$ in \eqref{eq:M1M2}.
It is then clear that if one defines matrices $\widetilde{\mathbf{J}}(x)$ and $\widetilde{\mathbf{K}}(x)$ from the coefficients of $\lambda^{-1}$ and $\lambda^{-2}$ respectively after taking the correct incremented/decremented value of $\widetilde{n}$ in the form of $\widetilde{\boldsymbol{\Lambda}}(\lambda,x)$ in \eqref{eq:tilde-Lambda-X-matrices}, then $\widetilde{\mathbf{J}}(x)$ is indeed an off-diagonal matrix, and $\widetilde{\mathbf{K}}(x)$ is a non-diagonalizable matrix with zero eigenvalues.  Finally, we solve \eqref{eq:polynomial-gauge-form} for $\widetilde{\mathbf{X}}(\lambda,x)$ and compare with the form given in \eqref{eq:tilde-Lambda-X-matrices} for the computed coefficients $\widetilde{\mathbf{J}}(x)$ and $\widetilde{\mathbf{K}}(x)$.  We observe agreement due to the differential equations \eqref{eq:first}--\eqref{eq:q-explicit}.

For the transformation $\mathbf{G}^\uparrow(\lambda,x)$, the transformed coefficients are:
\begin{equation}
\widetilde{\mathbf{J}}(x)=\mathbf{J}^\uparrow(x):=\begin{bmatrix}0 & 2 x\ee^{\ii\varphi(x)}\\
\frac{x\ee^{\ii\varphi(x)}q(x)^2}{32u(x)^2}+\frac{(n+1)q(x)}{8xu(x)}-\frac{\ee^{-\ii\varphi(x)}u(x)}{2x^2} & 0\end{bmatrix}
\label{eq:M1up}
\end{equation}
and
\begin{equation}
\widetilde{\mathbf{K}}(x)=\mathbf{K}^\uparrow(x):=
\frac{\ii b x\ee^{\ii\varphi(x)}}{64u(x)^3}\begin{bmatrix}8u(x)q(x) & -64u(x)^2\\q(x)^2 & -8u(x)q(x)\end{bmatrix}.
\label{eq:M2up}
\end{equation}
For the transformation $\mathbf{G}^\downarrow(\lambda,x)$, the transformed coefficients are:
\begin{equation}
\widetilde{\mathbf{J}}(x)=\mathbf{J}^\downarrow(x):=\begin{bmatrix}
0 & -\frac{x\ee^{-\ii\varphi(x)}p(x)^2}{32u(x)^2}+\frac{(n-1)p(x)}{8xu(x)}+\frac{\ee^{\ii\varphi(x)}u(x)}{2x^2} \\
-2x\ee^{-\ii\varphi(x)} & 0\end{bmatrix}
\label{eq:M1down}
\end{equation}
and
\begin{equation}
\widetilde{\mathbf{K}}(x)=\mathbf{K}^\downarrow(x):=\frac{\ii b x\ee^{-\ii\varphi(x)}}{64u(x)^3}\begin{bmatrix}
-8u(x)p(x) & -p(x)^2\\64u(x)^2 & 8u(x)p(x)\end{bmatrix}.
\label{eq:M2down}
\end{equation}
In these formul\ae, $b$ is the constant expressed in terms of the potentials via \eqref{eq:b-no-derivatives} with $a=-\ii n$.

\subsubsection{The induction argument}
We now show that solutions of Riemann-Hilbert Problem~\ref{rhp:algebraic} for consecutive integer values of $n$ are related by suitable Schlesinger transformations.
\paragraph{\it The case of $n$ even.}
Suppose that $n$ is even and that for some $x>0$, Riemann-Hilbert Problem~\ref{rhp:algebraic} has a solution $\mathbf{W}^{(n)}(\lambda,x)$.  Let $\mathbf{G}^{(n)\uparrow}(\lambda,x)$ and $\mathbf{G}^{(n)\downarrow}(\lambda,x)$ denote the Schlesinger transformations associated with the potentials $u=u_n(x)$, $\ee^{\pm\ii\varphi}=\ee^{\pm\ii\varphi_n(x)}$, $p=p_n(x)$, and $q=q_n(x)$ for $\epsilon=-1$.  We will now show that $\mathbf{W}^{(n+1)}(\lambda,x)=\mathbf{G}^{(n)\uparrow}(\lambda,x)\mathbf{W}^{(n)}(\lambda,x)$ and that $\mathbf{W}^{(n-1)}(\lambda,x)=\mathbf{G}^{(n)\downarrow}(\lambda,x)\mathbf{W}^{(n)}(\lambda,x)$.  Obviously the transformed matrices are analytic exactly where $\mathbf{W}^{(n)}(\lambda,x)$ is, and it is easy to see that the jump conditions are satisfied in both cases because the Schlesinger transformation induces a sign change on the (positive imaginary) branch cut of $(\ii\lambda)^{\pm 1/2}$ but otherwise leaves the jump conditions invariant.  So it remains to check the normalization condition at $\lambda=\infty$ and the condition at $\lambda=0$.

We start by combining \eqref{eq:potentials-from-M12} with \eqref{eq:M1nM2n} to get, for $n$ even,
\begin{equation}
\begin{split}
u_n(x)=\epsilon \ee^{\ii\pi/6}xB_{0,12}^{(n)}(x)B_{0,22}^{(n)}(x), & \quad
\ee^{\pm\ii\varphi_n(x)}=-\epsilon \left[\frac{B_{0,12}^{(n)}(x)}{B_{0,22}^{(n)}(x)}\right]^{\pm 1},\\
p_n(x)= -8\ee^{\ii\pi/6}xA_{1,12}^{(n)}(x)B_{0,12}^{(n)}(x)B_{0,22}^{(n)}(x), & \quad
q_n(x)=8\ee^{\ii\pi/6}xA_{1,21}^{(n)}(x)B_{0,12}^{(n)}(x)B_{0,22}^{(n)}(x).
\end{split}
\end{equation}
Using these in \eqref{eq:SchlesingerCoeffs1}--\eqref{eq:SchlesingerCoeffs2} and substituting into \eqref{eq:G-up-coeffs} and \eqref{eq:G-down-coeffs} gives
\begin{equation}
\mathbf{G}^{(n)\uparrow}(\lambda,x)=\left(\begin{bmatrix}1& 0\\
\epsilon A_{1,21}^{(n)}(x) &1\end{bmatrix}+\left(\frac{\ii\lambda}{x}\right)^{-1}\begin{bmatrix}0& \epsilon B_{0,12}^{(n)}(x)B_{0,22}^{(n)}(x)^{-1}\\
0 & A_{1,21}^{(n)}(x)B_{0,12}^{(n)}(x)B_{0,22}^{(n)}(x)^{-1}\end{bmatrix}\right)\left(\frac{\ii\lambda}{x}\right)^{-\sigma_3/2},
\end{equation}
\begin{equation}
\mathbf{G}^{(n)\downarrow}(\lambda,x)=\left(\begin{bmatrix}1 & \epsilon A_{1,12}^{(n)}(x)\\0 & 1\end{bmatrix}+\left(\frac{\ii\lambda}{x}\right)^{-1}
\begin{bmatrix} A_{1,12}^{(n)}(x)B_{0,22}^{(n)}(x)B_{0,12}^{(n)}(x)^{-1} & 0\\\epsilon B_{0,22}^{(n)}(x)B_{0,12}^{(n)}(x)^{-1} & 0\end{bmatrix}\right)
\left(\frac{\ii\lambda}{x}\right)^{\sigma_3/2}.
\end{equation}
Now using \eqref{eq:lambda-infinity-stronger} from Lemma~\ref{lem:expansions} gives
\begin{multline}
\mathbf{G}^{(n)\uparrow}(\lambda,x)\mathbf{W}^{(n)}(\lambda,x)\ee^{\ii x(\ii x\lambda)^{-1/2}\sigma_3}\left(\frac{\ii\lambda}{x}\right)^{(n+1)\sigma_3}\\
\begin{aligned}
&=
\left(\begin{bmatrix}1&0\\\epsilon A_{1,21}^{(n)}(x) & 1\end{bmatrix}+\mathcal{O}(\lambda^{-1})\right)\left(\frac{\ii\lambda}{x}\right)^{-\sigma_3/2}
\left(\mathbb{I}+\left(\frac{\ii\lambda}{x}\right)^{-1}\mathbf{A}_1^{(n)}(x)+\mathcal{O}(\lambda^{-2})\right)\left(\frac{\ii\lambda}{x}\right)^{\sigma_3/2}\\
&=\left(\begin{bmatrix}1&0\\\epsilon A_{1,21}^{(n)}(x) & 1\end{bmatrix}+\mathcal{O}(\lambda^{-1})\right)\left(\begin{bmatrix}1&0\\A_{1,21}^{(n)}(x) & 1\end{bmatrix}+\mathcal{O}(\lambda^{-1})\right)\\
&=\mathbb{I}+\mathcal{O}(\lambda^{-1}),\quad\lambda\to\infty, \quad\text{and}
\end{aligned}
\label{eq:G-up-infinity-even}
\end{multline}
\begin{multline}
\mathbf{G}^{(n)\downarrow}(\lambda,x)\mathbf{W}^{(n)}(\lambda,x)\ee^{\ii x(\ii x\lambda)^{-1/2}\sigma_3}\left(\frac{\ii\lambda}{x}\right)^{(n-1)\sigma_3}\\
\begin{aligned}
&=
\left(\begin{bmatrix}1 & \epsilon A_{1,12}^{(n)}(x)\\0 & 1\end{bmatrix}+\mathcal{O}(\lambda^{-1})\right)\left(\frac{\ii\lambda}{x}\right)^{\sigma_3/2}
\left(\mathbb{I}+\left(\frac{\ii\lambda}{x}\right)^{-1}\mathbf{A}_1^{(n)}(x) + \mathcal{O}(\lambda^{-2})\right)\left(\frac{\ii\lambda}{x}\right)^{-\sigma_3/2}\\
&= \left(\begin{bmatrix}1 & \epsilon A_{1,12}^{(n)}(x)\\0 & 1\end{bmatrix}+\mathcal{O}(\lambda^{-1})\right)\left(\begin{bmatrix}1 & A_{1,12}^{(n)}(x)\\0 & 1\end{bmatrix} + \mathcal{O}(\lambda^{-1})\right)\\
&=\mathbb{I}+\mathcal{O}(\lambda^{-1}),\quad \lambda\to\infty,
\end{aligned}
\label{eq:G-down-infinity-even}
\end{multline}
using $\epsilon =-1$.  Therefore $\mathbf{G}^{(n)\uparrow}(\lambda,x)\mathbf{W}^{(n)}(\lambda,x)$ and $\mathbf{G}^{(n)\downarrow}(\lambda,x)\mathbf{W}^{(n)}(\lambda,x)$ behave respectively as $\mathbf{W}^{(n+1)}(\lambda,x)$ and $\mathbf{W}^{(n-1)}(\lambda,x)$ are required to according to Riemann-Hilbert Problem~\ref{rhp:algebraic}.  On the other hand, using \eqref{eq:lambda-0-stronger} from Lemma~\ref{lem:expansions} for $n$ even gives
\begin{multline}
\mathbf{G}^{(n)\uparrow}(\lambda,x)\mathbf{W}^{(n)}(\lambda,x)\ee^{-\ii x\lambda}\mathbf{E}\cdot\left(\frac{\ii\lambda}{x}\right)^{\sigma_3/4}\\
\begin{aligned}
&=\left(\left(\frac{\ii\lambda}{x}\right)^{-1}\begin{bmatrix}0& \epsilon B_{0,12}^{(n)}(x)B_{0,22}^{(n)}(x)^{-1}\\
0 & A_{1,21}^{(n)}(x)B_{0,12}^{(n)}(x)B_{0,22}^{(n)}(x)^{-1}\end{bmatrix}+\begin{bmatrix}1& 0\\
\epsilon A_{1,21}^{(n)}(x) &1\end{bmatrix}\right)\\
&\quad\quad\quad{}\cdot\left(\frac{\ii\lambda}{x}\right)^{-\sigma_3/2}\left(\mathbf{B}_0^{(n)}(x)+\left(\frac{\ii\lambda}{x}\right)\mathbf{B}_1^{(n)}(x) + \mathcal{O}(\lambda^{-2})\right)\left(\frac{\ii\lambda}{x}\right)^{\sigma_3/2}\\
&=\left(\left(\frac{\ii\lambda}{x}\right)^{-1}\begin{bmatrix}0& \epsilon B_{0,12}^{(n)}(x)B_{0,22}^{(n)}(x)^{-1}\\
0 & A_{1,21}^{(n)}(x)B_{0,12}^{(n)}(x)B_{0,22}^{(n)}(x)^{-1}\end{bmatrix}+\begin{bmatrix}1& 0\\
\epsilon A_{1,21}^{(n)}(x) &1\end{bmatrix}\right)\\
&\quad\quad\quad{}\cdot\left(\left(\frac{\ii\lambda}{x}\right)^{-1}\begin{bmatrix}0 & B_{0,12}^{(n)}(x)\\0 & 0\end{bmatrix} + 
\begin{bmatrix}B_{0,11}^{(n)}(x) & B_{1,12}^{(n)}(x)\\0 & B_{0,22}^{(n)}(x)\end{bmatrix}+\mathcal{O}(\lambda)\right)\\
&=\mathcal{O}(1),\quad\lambda\to 0,
\end{aligned}
\end{multline}
and
\begin{multline}
\mathbf{G}^{(n)\downarrow}(\lambda,x)\mathbf{W}^{(n)}(\lambda,x)\ee^{-\ii x\lambda}\mathbf{E}\cdot\left(\frac{\ii\lambda}{x}\right)^{\sigma_3/4}\\
\begin{aligned}
&=\left(\left(\frac{\ii\lambda}{x}\right)^{-1}\begin{bmatrix}A_{1,12}^{(n)}(x)B_{0,22}^{(n)}(x)B_{0,12}^{(n)}(x)^{-1} & 0\\\epsilon B_{0,22}^{(n)}(x)B_{0,12}^{(n)}(x)^{-1} & 0
\end{bmatrix}+\begin{bmatrix}1 & \epsilon A_{1,12}^{(n)}(x)\\0 & 1\end{bmatrix}\right)\\
&\quad\quad\quad{}\cdot\left(\frac{\ii\lambda}{x}\right)^{\sigma_3/2}\left(\mathbf{B}_0^{(n)}(x) +\left(\frac{\ii\lambda}{x}\right)\mathbf{B}_1^{(n)}(x) + \mathcal{O}(\lambda^{-2})\right)\left(\frac{\ii\lambda}{x}\right)^{\sigma_3/2}\\
&=\left(\left(\frac{\ii\lambda}{x}\right)^{-1}\begin{bmatrix}A_{1,12}^{(n)}(x)B_{0,22}^{(n)}(x)B_{0,12}^{(n)}(x)^{-1} & 0\\\epsilon B_{0,22}^{(n)}(x)B_{0,12}^{(n)}(x)^{-1} & 0
\end{bmatrix}+\begin{bmatrix}1 & \epsilon A_{1,12}^{(n)}(x)\\0 & 1\end{bmatrix}\right)\\
&\quad\quad\quad{}\cdot\left(\left(\frac{\ii\lambda}{x}\right)^{-1}\begin{bmatrix}0&0\\0 & B_{0,22}^{(n)}(x)\end{bmatrix}+
\begin{bmatrix}0 & B_{0,12}^{(n)}(x)\\B_{0,21}^{(n)}(x) & B_{1,22}^{(n)}(x)\end{bmatrix}+\mathcal{O}(\lambda)\right)\\
&=\mathcal{O}(1),\quad\lambda\to 0,
\end{aligned}
\end{multline}
again using $\epsilon=-1$.  This proves that $\mathbf{G}^{(n)\uparrow}(\lambda,x)\mathbf{W}^{(n)}(\lambda,x)$ and $\mathbf{G}^{(n)\downarrow}(\lambda,x)\mathbf{W}^{(n)}(\lambda,x)$ behave, respectively, as $\mathbf{W}^{(n+1)}(\lambda,x)$ and $\mathbf{W}^{(n-1)}(\lambda,x)$ are required to according to the conditions of Riemann-Hilbert Problem~\ref{rhp:algebraic}, taking into account that $n$ is even.  

Therefore if $\mathbf{W}^{(n)}(\lambda,x)$ is the (unique, by Lemma~\ref{lem:uniqueness}) solution of Riemann-Hilbert Problem~\ref{rhp:algebraic} for  $n$ even, then 
$\mathbf{W}^{(n+1)}(\lambda,x):=\mathbf{G}^{(n)\uparrow}(\lambda,x)\mathbf{W}^{(n)}(\lambda,x)$ and $\mathbf{W}^{(n-1)}(\lambda,x):=\mathbf{G}^{(n)\downarrow}(\lambda,x)\mathbf{W}^{(n)}(\lambda,x)$ satisfy all the conditions of the same problem for $n\mapsto n+1$ and $n\mapsto n-1$ respectively, so they are the unique solutions of those problems.

\paragraph{\it The case of $n$ odd.}
Now suppose that $n$ is odd and that for some $x>0$, Riemann-Hilbert Problem~\ref{rhp:algebraic} has a solution $\mathbf{W}^{(n)}(\lambda,x)$.
Again let $\mathbf{G}^{(n)\uparrow}(\lambda,x)$ and $\mathbf{G}^{(n)\downarrow}(\lambda,x)$ denote the Schlesinger transformations associated with the potentials $u=u_n(x)$, $\ee^{\pm\ii\varphi}=\ee^{\pm\ii\varphi_n(x)}$, $p=p_n(x)$, and $q=q_n(x)$ for $\epsilon =-1$.   As before, we just have to check that the transformed matrices $\mathbf{G}^{(n)\uparrow}(\lambda,x)\mathbf{W}^{(n)}(\lambda,x)$ and $\mathbf{G}^{(n)\downarrow}(\lambda,x)\mathbf{W}^{(n)}(\lambda,x)$ behave as required as $\lambda\to\infty$ and $\lambda\to 0$.  Combining \eqref{eq:potentials-from-M12} with \eqref{eq:M1nM2n} for $n$ odd gives
\begin{equation}
\begin{split}
u_n(x)=-\epsilon \ee^{5\pi\ii/6}xB_{0,11}^{(n)}(x)B_{0,21}^{(n)}(x), & \quad
\ee^{\pm\ii\varphi_n(x)}=-\epsilon\left[\frac{B^{(n)}_{0,11}(x)}{B^{(n)}_{0,21}(x)}\right]^{\pm 1},\\
p_n(x)=8\ee^{5\pi\ii/6}xA_{1,12}^{(n)}(x)B_{0,11}^{(n)}(x)B_{0,21}^{(n)}(x), & \quad
q_n(x)=-8\ee^{5\pi\ii/6}xA_{1,21}^{(n)}(x)B_{0,11}^{(n)}(x)B_{0,21}^{(n)}(x).
\end{split}
\end{equation}
Then using \eqref{eq:G-up-coeffs}--\eqref{eq:SchlesingerCoeffs2} gives
\begin{equation}
\mathbf{G}^{(n)\uparrow}(\lambda,x)=\left(\begin{bmatrix}1&0\\\epsilon A_{1,21}^{(n)}(x) & 1\end{bmatrix} +\left(\frac{\ii\lambda}{x}\right)^{-1}
\begin{bmatrix}0&\epsilon B_{0,11}^{(n)}(x)B_{0,21}^{(n)}(x)^{-1}\\0 & A_{1,21}^{(n)}(x)B_{0,11}^{(n)}(x)B_{0,21}^{(n)}(x)^{-1}\end{bmatrix}\right)\left(\frac{\ii\lambda}{x}\right)^{-\sigma_3/2}
\end{equation}
and
\begin{equation}
\mathbf{G}^{(n)\downarrow}(\lambda,x)=\left(\begin{bmatrix}1&\epsilon A_{1,12}^{(n)}(x)\\0 & 1\end{bmatrix}+\left(\frac{\ii\lambda}{x}\right)^{-1}
\begin{bmatrix}A_{1,12}^{(n)}(x)B_{0,21}^{(n)}(x)B_{0,11}^{(n)}(x)^{-1} & 0\\\epsilon B_{0,21}^{(n)}(x)B_{0,11}^{(n)}(x)^{-1} & 0\end{bmatrix}\right)\left(\frac{\ii\lambda}{x}\right)^{\sigma_3/2}.
\end{equation}
Since the Schlesinger transformations for $n$ odd agree with those for $n$ even in their leading terms for large $\lambda$, exactly the same calculations \eqref{eq:G-up-infinity-even}--\eqref{eq:G-down-infinity-even} apply for $n$ odd (although the Laudau symbols stand for different expressions in the even and odd cases).  Therefore using $\epsilon =-1$, $\mathbf{G}^{(n)\uparrow}(\lambda,x)\mathbf{W}^{(n)}(\lambda,x)$ and $\mathbf{G}^{(n)\downarrow}(\lambda,x)\mathbf{W}^{(n)}(\lambda,x)$ behave respectively as $\mathbf{W}^{(n+1)}(\lambda,x)$ and $\mathbf{W}^{(n-1)}(\lambda,x)$ are required to in the limit $\lambda\to\infty$ according to Riemann-Hilbert Problem~\ref{rhp:algebraic}.
Now using \eqref{eq:lambda-0-stronger} from Lemma~\ref{lem:expansions} for $n$ odd gives
\begin{multline}
\mathbf{G}^{(n)\uparrow}(\lambda,x)\mathbf{W}^{(n)}(\lambda,x)\ee^{-\ii x\lambda}\mathbf{E}\cdot\left(\frac{\ii\lambda}{x}\right)^{-\sigma_3/4}\\
\begin{aligned}
&=\left(\left(\frac{\ii\lambda}{x}\right)^{-1}
\begin{bmatrix}0&\epsilon B_{0,11}^{(n)}(x)B_{0,21}^{(n)}(x)^{-1}\\0 & A_{1,21}^{(n)}(x)B_{0,11}^{(n)}(x)B_{0,21}^{(n)}(x)^{-1}\end{bmatrix}+\begin{bmatrix}1&0\\\epsilon A_{1,21}^{(n)}(x) & 1\end{bmatrix}\right)\\
&\quad\quad\quad{}\cdot\left(\frac{\ii\lambda}{x}\right)^{-\sigma_3/2}\left(\mathbf{B}_0^{(n)}(x)+\left(\frac{\ii\lambda}{x}\right)\mathbf{B}_1^{(n)}(x) + \mathcal{O}(\lambda^{-2})\right)\left(\frac{\ii\lambda}{x}\right)^{-\sigma_3/2}\\
&=\left(\left(\frac{\ii\lambda}{x}\right)^{-1}
\begin{bmatrix}0&\epsilon B_{0,11}^{(n)}(x)B_{0,21}^{(n)}(x)^{-1}\\0 & A_{1,21}^{(n)}(x)B_{0,11}^{(n)}(x)B_{0,21}^{(n)}(x)^{-1}\end{bmatrix}+\begin{bmatrix}1&0\\\epsilon A_{1,21}^{(n)}(x) & 1\end{bmatrix}\right)\\
&\quad\quad\quad{}\cdot\left(\left(\frac{\ii\lambda}{x}\right)^{-1}\begin{bmatrix}B_{0,11}^{(n)}(x) & 0\\0 & 0\end{bmatrix} + 
\begin{bmatrix}B_{1,11}^{(n)}(x) & B_{0,12}^{(n)}(x)\\B_{0,21}^{(n)}(x) & 0\end{bmatrix} + \mathcal{O}(\lambda)\right)\\
&=\mathcal{O}(1),\quad\lambda\to 0,\quad\text{and}
\end{aligned}
\end{multline}
\begin{multline}
\mathbf{G}^{(n)\downarrow}(\lambda,x)\mathbf{W}^{(n)}(\lambda,x)\ee^{-\ii x\lambda}\mathbf{E}\cdot\left(\frac{\ii\lambda}{x}\right)^{-\sigma_3/4}\\
\begin{aligned}
&=\left(\left(\frac{\ii\lambda}{x}\right)^{-1}
\begin{bmatrix}A_{1,12}^{(n)}(x)B_{0,21}^{(n)}(x)B_{0,11}^{(n)}(x)^{-1} & 0\\\epsilon B_{0,21}^{(n)}(x)B_{0,11}^{(n)}(x)^{-1} & 0\end{bmatrix}+\begin{bmatrix}1&\epsilon A_{1,12}^{(n)}(x)\\0 & 1\end{bmatrix}\right)\\
&\quad\quad\quad{}\cdot\left(\frac{\ii\lambda}{x}\right)^{\sigma_3/2}\left(\mathbf{B}_0^{(n)}(x) + \left(\frac{\ii\lambda}{x}\right)\mathbf{B}_1^{(n)}(x) + \mathcal{O}(\lambda^2)\right)\left(\frac{\ii\lambda}{x}\right)^{-\sigma_3/2}\\
&=\left(\left(\frac{\ii\lambda}{x}\right)^{-1}
\begin{bmatrix}A_{1,12}^{(n)}(x)B_{0,21}^{(n)}(x)B_{0,11}^{(n)}(x)^{-1} & 0\\\epsilon B_{0,21}^{(n)}(x)B_{0,11}^{(n)}(x)^{-1} & 0\end{bmatrix}+\begin{bmatrix}1&\epsilon A_{1,12}^{(n)}(x)\\0 & 1\end{bmatrix}\right)\\
&\quad\quad\quad{}\cdot\left(\left(\frac{\ii\lambda}{x}\right)^{-1}\begin{bmatrix}0 & 0\\B_{0,21}^{(n)}(x) & 0\end{bmatrix}+\begin{bmatrix}B_{0,11}^{(n)}(x) & 0\\
B_{1,21}^{(n)}(x) & B_{0,22}^{(n)}(x)\end{bmatrix}+\mathcal{O}(\lambda)\right)\\
&=\mathcal{O}(1),\quad\lambda\to 0,
\end{aligned}
\end{multline}
using $\epsilon=-1$.  So, $\mathbf{G}^{(n)\uparrow}(\lambda,x)\mathbf{W}^{(n)}(\lambda,x)$ and $\mathbf{G}^{(n)\downarrow}(\lambda,x)\mathbf{W}^{(n)}(\lambda,x)$ respectively behave the same as $\lambda\to 0$ as $\mathbf{W}^{(n+1)}(\lambda,x)$ and $\mathbf{W}^{(n-1)}(\lambda,x)$.  Along with the behavior as $\lambda\to\infty$ and the analyticity and jump properties, we conclude that $\mathbf{W}^{(n+1)}(\lambda,x):=\mathbf{G}^{(n)\uparrow}(\lambda,x)\mathbf{W}^{(n)}(\lambda,x)$ and $\mathbf{W}^{(n-1)}(\lambda,x):=\mathbf{G}^{(n)\downarrow}(\lambda,x)\mathbf{W}^{(n)}(\lambda,x)$ are the unique solutions of Riemann-Hilbert Problem~\ref{rhp:algebraic} for $n\mapsto n+1$ and for $n\mapsto n-1$ respectively.

\subsubsection{B\"acklund transformations}
Since for any $n\in\mathbb{Z}$, $\mathbf{J}^{(n+1)}(x)=\mathbf{J}^{(n)\uparrow}(x)$, $\mathbf{K}^{(n+1)}(x)=\mathbf{K}^{(n)\uparrow}(x)$, $\mathbf{J}^{(n-1)}(x)=\mathbf{J}^{(n)\downarrow}(x)$, and $\mathbf{K}^{(n-1)}(x)=\mathbf{K}^{(n)\downarrow}(x)$ we derive from \eqref{eq:potentials-from-M12} with $\epsilon =-1$ and \eqref{eq:M1up}--\eqref{eq:M2up} the explicit B\"acklund transformations
\begin{equation}
\begin{split}
u_{n+1}(x)&=-\frac{\ii b_n x^2\ee^{\ii\varphi_n(x)}q_n(x)}{8u_n(x)^2},\quad
\ee^{\pm\ii\varphi_{n+1}(x)}=\left[\frac{8u_n(x)}{q_n(x)}\right]^{\pm 1},\\
p_{n+1}(x)&=\frac{\ii b_n x^2\ee^{2\ii\varphi_n(x)}q_n(x)}{u_n(x)^2},\\
q_{n+1}(x)&=\frac{\ii b_nx\ee^{\ii\varphi_n(x)}q_n(x)}{2u_n(x)^2}\left[\frac{x\ee^{\ii\varphi_n(x)}q_n(x)^2}{32u_n(x)^2}+\frac{(n+1)q_n(x)}{8xu_n(x)}-\frac{\ee^{-\ii\varphi_n(x)}u_n(x)}{2x^2}\right],
\end{split}
\label{eq:Baecklund-up}
\end{equation}
and using \eqref{eq:M1down}--\eqref{eq:M2down} instead of \eqref{eq:M1up}--\eqref{eq:M2up} we get
\begin{equation}
\begin{split}
u_{n-1}(x)&= \frac{\ii b_n x^2\ee^{-\ii\varphi_n(x)}p_n(x)}{8u_n(x)^2},\quad
\ee^{\pm\ii\varphi_{n-1}(x)}= -\left[\frac{p_n(x)}{8u_n(x)}\right]^{\pm 1},\\
p_{n-1}(x)&= \frac{\ii b_n x \ee^{-\ii\varphi_n(x)}p_n(x)}{2u_n(x)^2}\left[\frac{x\ee^{-\ii\varphi_n(x)}p_n(x)^2}{32u_n(x)^2}-\frac{(n-1)p_n(x)}{8xu_n(x)}-\frac{\ee^{\ii\varphi_n(x)}u_n(x)}{2x^2}\right],\\
q_{n-1}(x)&=\frac{\ii b_n x^2\ee^{-2\ii\varphi_n(x)}p_n(x)}{u_n(x)^2}.
\end{split}
\label{eq:Baecklund-down}
\end{equation}
In these expressions, $b_n$ is the constant given, for any $n\in\mathbb{Z}$, by
\begin{equation}
b_n=\frac{2\ii n u_n(x)}{x}-\frac{1}{2}\ii x p_n(x)\ee^{-\ii\varphi_n(x)} +\frac{1}{2}\ii x q_n(x)\ee^{\ii\varphi_n(x)}.
\end{equation}
It then follows from \eqref{eq:Baecklund-up} that $b_{n+1}=b_n$ and from \eqref{eq:Baecklund-down} that $b_{n-1}=b_n$ as well.  Since $u_0(x)=\tfrac{1}{2}x^{1/3}$ satisfies \eqref{eq:D7KV} with $\epsilon=-1$, $a=-\ii n$ for $n=0$ and $b=b_0=\ii$, it is then clear that the function $u_n(x)$ extracted from Riemann-Hilbert Problem~\ref{rhp:algebraic} using \eqref{eq:un-formula-Y} satisfies \eqref{eq:D7KV} with $\epsilon=-1$, $a=-\ii n$, and $b=b_n=\ii$ for all $n\in\mathbb{Z}$.

Now, when $n=0$, $u_0(x)$ is obviously a rational function of $x^{1/3}$, while from \eqref{eq:pq0} we see that $\ee^{\ii\varphi_0(x)}\ee^{3x^{2/3}}$, $\ee^{-\ii\varphi_0(x)}\ee^{-3x^{2/3}}$, $p_0(x)\ee^{3x^{2/3}}$, and $q_0(x)\ee^{-3x^{2/3}}$ are rational in $x^{1/3}$ as well.  It follows inductively from \eqref{eq:Baecklund-up} and \eqref{eq:Baecklund-down} that for all $n\in\mathbb{Z}$,  $u_n(x)$, $\ee^{\ii\varphi_n(x)}\ee^{3x^{2/3}}$, $\ee^{-\ii\varphi_n(x)}\ee^{-3x^{2/3}}$, $p_n(x)\ee^{3x^{2/3}}$, and $q_n(x)\ee^{-3x^{2/3}}$ are rational functions of $x^{1/3}$.

We summarize the results of Sections~\ref{sec:diff-eqs}--\ref{sec:Schlesinger} in the following theorem.
\begin{theorem}
Suppose that Riemann-Hilbert Problem~\ref{rhp:algebraic} is solvable for given $n\in\mathbb{Z}$ and $x>0$.  Then the function $u_n(x)$ defined by \eqref{eq:un-formula-Y} is the unique solution of \eqref{eq:D7KV} with $\epsilon=-1$, $a=-\ii n$, and $b=\ii$ that is a rational function of $x^{1/3}$.
\label{thm:RHP-representation}
\end{theorem}

\subsection{Change of branch cut}
\label{sec:branch-cut}
Although the principal branch of $(\ii\lambda)^p$ for various powers $p$ (cut on the positive imaginary axis) is most convenient for describing the canonical solutions in different Stokes sectors at $\lambda=0$ and $\lambda=\infty$, for subsequent asymptotic analysis in the limit of large $n$ it will be better to reformulate a version of Riemann-Hilbert Problem~\ref{rhp:algebraic} involving power functions with branch cuts on the negative imaginary axis instead.  To this end, we start with $\mathbf{W}^{(n)}(\lambda,x)$ solving Riemann-Hilbert Problem~\ref{rhp:algebraic} and, letting $D_0^\pm\subset D_0$ denote the region between $C^\pm$ and the imaginary axis, we set
\begin{equation}
\mathbf{Y}^{(n)}(\lambda,x):=\ii^{n\sigma_3}\mathbf{W}^{(n)}(\lambda,x)\ee^{-\ii (x\lambda-x(\ii x\lambda)^{-1/2})\sigma_3}
\begin{cases}
\ee^{(\ii x\lambda-x(-\ii x\lambda)^{-1/2})\sigma_3},&\quad\lambda\in D_\infty^+,\\
(-1)^n\ee^{(\ii x\lambda-x(-\ii x\lambda)^{-1/2})\sigma_3},&\quad\lambda\in D_\infty^-,\\
\ee^{(\ii x\lambda-x(-\ii x\lambda)^{-1/2})\sigma_3}\ee^{-5\pi\ii\sigma_3/6},&\quad\lambda\in D_0^+,\\
(-1)^n(-\ii\sigma_2)\ee^{(\ii x\lambda-x(-\ii x\lambda)^{-1/2})\sigma_3}\ee^{\ii\pi\sigma_3/3},&\quad\lambda\in D_0^-.
\end{cases}
\label{eq:W-Y}
\end{equation}
Letting $\Sigma_0^-$ denote the segment of the imaginary axis between $\lambda=0$ and $\lambda=-\ii$, oriented downwards, 
\begin{equation}
\mathbf{Y}^{(n)}_+(\lambda,x)=\mathbf{Y}^{(n)}_-(\lambda,x) \ee^{-(\ii x\lambda-x(-\ii x\lambda)^{-1/2})\sigma_3}\begin{bmatrix}1&0\\\ii & 1\end{bmatrix}
\ee^{(\ii x\lambda-x(-\ii x\lambda)^{-1/2})\sigma_3},\quad\lambda\in\Sigma_\infty^+,
\end{equation}
\begin{equation}
\mathbf{Y}^{(n)}_+(\lambda,x)=\mathbf{Y}^{(n)}_-(\lambda,x)\ee^{-(\ii x\lambda -x(-\ii x\lambda_-)^{-1/2})\sigma_3}(-1)^n\begin{bmatrix}1&-\ii\\0&1\end{bmatrix}\ee^{(\ii x\lambda-x(-\ii x\lambda_+)^{-1/2})\sigma_3},\quad\lambda\in\Sigma_\infty^-,
\end{equation}
\begin{equation}
\mathbf{Y}^{(n)}_+(\lambda,x)=\mathbf{Y}^{(n)}_-(\lambda,x)\ee^{-(\ii x\lambda-x(-\ii x\lambda)^{-1/2})\sigma_3}\begin{bmatrix}1 & 0\\\ii & 1 \end{bmatrix} \ee^{(\ii x\lambda-x(-\ii x\lambda)^{-1/2})\sigma_3},\quad\lambda\in\Sigma_0^+,
\end{equation}
\begin{equation}
\begin{split}
\mathbf{Y}^{(n)}_+(\lambda,x)&=\mathbf{Y}^{(n)}_-(\lambda,x)\ee^{-(\ii x\lambda-x(-\ii x\lambda_-)^{-1/2})\sigma_3}(-1)^n\ii\sigma_1
\ee^{(\ii x\lambda-x(-\ii x\lambda_+)^{-1/2})\sigma_3}\\
&=\mathbf{Y}^{(n)}(\lambda,x)\ee^{-\ii x\lambda\sigma_3}(-1)^n\ii\sigma_1\ee^{\ii x\lambda\sigma_3},\quad\lambda\in\Sigma_0^-,
\end{split}
\end{equation}
\begin{equation}
\mathbf{Y}^{(n)}_+(\lambda,x)=\mathbf{Y}^{(n)}_-(\lambda,x)\ee^{-(\ii x\lambda -x(-\ii x\lambda)^{-1/2})\sigma_3}\begin{bmatrix}1 & 0\\
-\ii & 1\end{bmatrix}\ee^{(\ii x\lambda-x(-\ii x\lambda)^{-1/2})\sigma_3},\quad\lambda\in C^+,
\end{equation}
\begin{equation}
\mathbf{Y}^{(n)}_+(\lambda,x)=\mathbf{Y}^{(n)}_-(\lambda,x)\ee^{-(\ii x\lambda-x(-\ii x\lambda)^{-1/2})\sigma_3}\begin{bmatrix}1 & 0\\-\ii & 1\end{bmatrix}\ee^{(\ii x\lambda-x(-\ii x\lambda)^{-1/2})\sigma_3},\quad\lambda\in C^-.
\end{equation}
Using the identity relating principal branches:
\begin{equation}
(\ii\lambda)^p = \ee^{\pm\ii\pi p}(-\ii\lambda)^p,\quad \pm\mathrm{Im}(\ii\lambda)>0,
\end{equation}
we see from \eqref{eq:lambda-infinity-stronger} that regardless of whether $\lambda\to\infty$ from $D_\infty^+$ or $D_\infty^-$, 
\begin{equation}
\mathbf{Y}^{(n)}(\lambda,x)\ee^{x(-\ii x\lambda)^{-1/2}\sigma_3}\left(-\frac{\ii\lambda}{x}\right)^{n\sigma_3/2}\sim\mathbb{I}+\sum_{p=1}^\infty
\left(-\frac{\ii\lambda}{x}\right)^{-p}(-1)^p\ii^{n\sigma_3}\mathbf{A}_p^{(n)}(x)\ii^{-n\sigma_3},\quad\lambda\to\infty.
\end{equation}
Similarly, using \eqref{eq:lambda-0-stronger}, regardless of whether $\lambda\to 0$ from $D_0^+$ or $D_0^-$,
\begin{equation}
\mathbf{Y}^{(n)}(\lambda,x)\ee^{-\ii x\lambda\sigma_3}\widetilde{\mathbf{E}}\cdot\left(-\frac{\ii\lambda}{x}\right)^{-(-1)^n\sigma_3/4}\sim\sum_{p=0}^\infty
\left(-\frac{\ii\lambda}{x}\right)^p(-1)^p\ii^{n\sigma_3}\mathbf{B}_p^{(n)}(x)\ee^{\ii\pi(-1)^n\sigma_3/4},\quad\lambda\to 0,\quad\text{where}
\label{eq:Yn-zero-expansion}
\end{equation}
\begin{equation}
\widetilde{\mathbf{E}}:=\ee^{5\pi\ii\sigma_3/6}\mathbf{E}=\ee^{-\ii\pi\sigma_3/3}\ii\sigma_2\mathbf{E}\ii^{\sigma_3}=\frac{1}{\sqrt{2}}\begin{bmatrix}\ee^{5\pi\ii/6} & \ee^{-5\pi\ii/6}\\ \ee^{-\ii\pi/6} & \ee^{-5\pi\ii/6}\end{bmatrix}.
\end{equation}

Along with the analyticity of $\mathbf{Y}^{(n)}(\lambda,x)$ for $\lambda\in\mathbb{C}\setminus(\Sigma_\infty^+\cup\Sigma_\infty^-\cup\Sigma_0^+\cup\Sigma_0^-\cup C^+\cup C^-)$ and continuity of boundary values implied by the substitution \eqref{eq:W-Y}, these conditions amount to an equivalent Riemann-Hilbert problem for the matrix $\mathbf{Y}^{(n)}(\lambda,x)$.

\section{Application:  asymptotic analysis of the algebraic solution $u_n(x)$ for large $n$}
\label{sec:application}
\subsection{Rescaling}
\label{sec:rescaling}
Note that $u\mapsto \ii u$ and $x\mapsto \ii x$ takes a solution of \eqref{eq:D7KV} to a solution of the same equation for the same values of $b$ and $\epsilon$, but with $a\mapsto -a$.  Hence it is sufficient to assume that $n$ is a large positive integer.  We initially scale the variables in Riemann-Hilbert Problem~\ref{rhp:algebraic} with $n>0$ as
\begin{equation}
x=n^{3/2}y\quad\text{and}\quad \lambda=n^{-1/2}\nu.
\end{equation}
The two terms in the exponent in the jump matrices then balance with the exponent in the normalization condition as $\lambda\to\infty$, which is proportional to $n$:
\begin{equation}
\ii x\lambda -x(-\ii x\lambda)^{-1/2}=n(\ii y\nu-y(-\ii y\nu)^{-1/2}).
\end{equation}
It is convenient to further scale the spectral parameter by the rescaled parameter $y$, to make the dominant term near $\infty$ independent of $y$.  Thus for $y>0$, we set $\nu=y^{-1}\eta$ and obtain 
\begin{equation}
\ii x\lambda-x(-\ii x\lambda)^{-1/2}=n\Phi(\eta,y),\quad\Phi(\eta,y):=\ii\eta-y(-\ii\eta)^{-1/2}.
\label{eq:Phi-define}
\end{equation}
The quantity $\ii\lambda/x$ appearing in the normalization condition then reads
\begin{equation}
\frac{\ii\lambda}{x} = \frac{1}{n^2}\frac{\ii\nu}{y} = \frac{1}{n^2y^2}\ii\eta,\quad\text{so}\quad
\left(-\frac{\ii\lambda}{x}\right)^{n\sigma_3/2} = (ny)^{-n\sigma_3}(-\ii\eta)^{n\sigma_3/2}.
\end{equation}

We define a new unknown  that is a function of the variables $y,\eta$ instead of $x,\lambda$ by setting
\begin{equation}
\mathbf{Z}^{(n)}(\eta,y):=(ny)^{-n\sigma_3}\mathbf{Y}^{(n)}(n^{3/2}y,n^{-1/2}y^{-1}\eta).
\end{equation}
Then, after an unimportant rescaling of the jump contour for $\mathbf{Y}^{(n)}(\lambda,x)$ to fix it in the $\eta$-plane, $\mathbf{Z}^{(n)}(\eta,y)$ satisfies the following conditions:
\begin{itemize}
\item $\mathbf{Z}^{(n)}(\eta,y)$ is analytic for $\eta\in\mathbb{C}\setminus (\Sigma_\infty^+\cup\Sigma_\infty^-\cup\Sigma_0^+\cup\Sigma_0^-\cup C^+\cup C^-)$.
\item $\mathbf{Z}^{(n)}(\eta,y)$ takes continuous boundary values on the jump contour related by the jump conditions 
\begin{equation}
\mathbf{Z}^{(n)}_+(\eta,y)=\mathbf{Z}^{(n)}_-(\eta,y)\ee^{-n\Phi(\eta,y)\sigma_3}\begin{bmatrix}1&0\\\ii&1\end{bmatrix}\ee^{n\Phi(\eta,y)\sigma_3},\quad\eta\in\Sigma_\infty^+\cup\Sigma_0^+,
\label{eq:Z-jump-first}
\end{equation}
\begin{equation}
\mathbf{Z}_+^{(n)}(\eta,y)=\mathbf{Z}^{(n)}_-(\eta,y)\ee^{-n\Phi(\eta,y)\sigma_3}\begin{bmatrix}1&0\\-\ii & 1\end{bmatrix}\ee^{n\Phi(\eta,y)\sigma_3},\quad\eta\in C^+\cup C^-,
\end{equation}
\begin{equation}
\mathbf{Z}_+^{(n)}(\eta,y)=\mathbf{Z}^{(n)}_-(\eta,y)\ee^{-n\Phi_-(\eta,y)\sigma_3}(-1)^n\begin{bmatrix}1& -\ii\\0 & 1\end{bmatrix}\ee^{n\Phi_+(\eta,y)\sigma_3},\quad\eta\in\Sigma_\infty^-,\quad\text{and}
\end{equation}
\begin{equation}
\mathbf{Z}_+^{(n)}(\eta,y)=\mathbf{Z}^{(n)}_-(\eta,y)\ee^{-n\Phi_-(\eta,y)\sigma_3}(-1)^n\ii\sigma_1\ee^{n\Phi_+(\eta,y)\sigma_3},\quad\eta\in\Sigma_0^-.
\label{eq:Z-jump-last}
\end{equation}
\item $\mathbf{Z}^{(n)}(\eta,y)(-\ii\eta)^{n\sigma_3/2}\to\mathbb{I}$ as $\eta\to\infty$.
\item The limit of $\mathbf{Z}^{(n)}(\eta,y)\ee^{-\ii n\eta\sigma_3}\widetilde{\mathbf{E}}\cdot(-\ii\eta)^{-(-1)^n\sigma_3/4}$ as $\eta\to 0$ exists.  In terms of the leading coefficient in the expansion \eqref{eq:lambda-0-stronger} (and see also \eqref{eq:Yn-zero-expansion}), 
\begin{equation}
\lim_{\eta\to 0}\mathbf{Z}^{(n)}(\eta,y)\ee^{-\ii n\eta\sigma_3}\widetilde{\mathbf{E}}\cdot(-\ii\eta)^{-(-1)^n\sigma_3/4} = (-\ii ny)^{-n\sigma_3}\mathbf{B}_0^{(n)}(n^{3/2}y)(-\ii n y)^{-(-1)^n\sigma_3/2}.  
\label{eq:Z-origin-limit}
\end{equation}
This may be regarded as the definition of $\mathbf{B}_0^{(n)}(n^{3/2}y)$ in terms of $\mathbf{Z}^{(n)}(\eta,y)$.
\end{itemize}
\subsection{Aside:  scaling formalism}
We perturb the basic scaling $y=n^{-3/2}x$ designed to balance the terms in the exponents in Riemann-Hilbert Problem~\ref{rhp:algebraic} by writing $x=n^{3/2}(y+n^{-p}z)$ for some $p>0$ to be determined.  Then we also scale $u=n^qU$ and consider what the Painlev\'e-III (D7) equation \eqref{eq:D7KV} on $u(x)$ with $\epsilon=-1$, $a=-\ii n$, and $b=\ii$ implies for $U=U(z)$ when $y\neq 0$ is a fixed parameter:
\begin{equation}
n^{q-3+2p}U''(z) = n^{q-3+2p}\frac{U'(z)^2}{U(z)}-n^{q-3+p}\frac{U'(z)}{y+n^{-p}z} + n^{2q-3/2}\frac{8U(z)^2}{y+n^{-p}z} + n^{-1/2}\frac{2}{y+n^{-p}z} - n^{-q}\frac{1}{U(z)}.
\label{eq:D7KV-rescaled}
\end{equation}
Because $p>0$, the second term on the right-hand side is negligible compared with the left-hand side and the preceding term.  All remaining terms can be balanced with these by choosing $p=1$ and $q=\frac{1}{2}$.
Then \eqref{eq:D7KV-rescaled} becomes
\begin{equation}
U''(z)=\frac{U'(z)^2}{U(z)} + \frac{8}{y}U(z)^2+\frac{2}{y} -\frac{1}{U(z)} + \mathcal{O}(n^{-1}).
\end{equation}
The \emph{approximating equation} is obtained by dropping the formally-small $\mathcal{O}(n^{-1})$ error term.  It is an autonomous nonlinear second order equation on $U(z)$ with parameter $y\neq 0$.  It turns out that under these scalings, the algebraic solutions $u_n(x)$ of \eqref{eq:D7KV} behave for large $n$ like one or the other of two types of solutions of the approximating equation, depending on the value of the parameter $y$.  Although the motivation comes from the simple scaling formalism above, these are correspondences that are proved using techniques of analysis of Riemann-Hilbert problems.

\subsubsection{Equilibrium solutions}
If $U$ is independent of $z$, the approximating equation yields equilibrium solutions solving the cubic equation
\begin{equation}
8U^3+2U-y=0.
\label{eq:U-equilibrium}
\end{equation}
For large $y$, the solutions are $U\approx \frac{1}{2}y^{1/3}$ and $U\approx\frac{1}{2}\ee^{\pm 2\pi\ii/3}y^{1/3}$.  Reversing the scalings by $U=n^{-1/2}u$ and $y=n^{-3/2}x$ (neglecting $n^{-1}z$) the large-$y$ equilibrium solution $U\approx\frac{1}{2}y^{1/3}$ reads $u\approx \frac{1}{2}x^{1/3}$, which is exactly the seed solution for $n=0$.
The three distinct solutions for large $y$ can coalesce at branch points that can be found by equating to zero the discriminant
of the cubic \eqref{eq:U-equilibrium} with respect to $U$, which is $-64(27y^2+4)$.

\subsubsection{Non-equilibrium solutions}
Multiplying the approximating equation through by $U'(z)/U(z)^2$ we obtain
\begin{equation}
\frac{U'(z)U''(z)}{U(z)^2}-\frac{U'(z)^3}{U(z)^3}-\frac{8}{y}U'(z) - \frac{2}{y}\frac{U'(z)}{U(z)^2}+\frac{U'(z)}{U(z)^3}=0
\end{equation}
or, equivalently,
\begin{equation}
\frac{\dd}{\dd z}\left[\frac{U'(z)^2}{2U(z)^2}-\frac{8}{y}U(z)+\frac{2}{y}\frac{1}{U(z)}-\frac{1}{2}\frac{1}{U(z)^2}\right]=0.
\end{equation}
Letting $E$ denote the implied integration constant, 
\begin{equation}
U'(z)^2 = \frac{16}{y}U(z)^3+2EU(z)^2-\frac{4}{y}U(z)+1.
\label{eq:nonequilibrium-elliptic}
\end{equation}
Setting
\begin{equation}
U(z)=\frac{1}{4}yw(z)-\frac{1}{24}yE,
\end{equation}
$w(z)$ solves the Weierstra\ss\ equation $w'(z)^2=4w(z)^3-g_2w(z)-g_3$ (see \cite[Ch.\@ III.5]{Markushevich05} or \cite[Ch.\@ 23]{DLMF}) with invariants
\begin{equation}
g_2=\frac{16}{y^2}+\frac{E^2}{3};\quad g_3=-\frac{16}{y^2}-\frac{8E}{3y^2} -\frac{E^3}{27}.
\end{equation}
Therefore, the general solution of the approximating equation can be written in terms of the Weierstra\ss\ elliptic function $\wp(z)$ with these invariants as
\begin{equation}
U(z)=\frac{1}{4}y\wp(z-z_0)-\frac{1}{24}yE.
\label{eq:non-equilibrium-approx}
\end{equation}
The fact that all poles of $\wp(z)$ are double is consistent with the fact that all nonzero poles of solutions of \eqref{eq:D7KV} are also double.

\subsection{$g$-function and spectral curves}
\label{sec:SpectralCurves}
Returning to the large-$n$ analysis of $\mathbf{Z}^{(n)}(\eta,y)$, we now introduce a $g$-function $\eta\mapsto g(\eta,y)$ with parameter $y$ with the following properties:
\begin{itemize}
\item $g$ is analytic for $\eta\in\mathbb{C}\setminus(\Sigma^+_\infty\cup\Sigma^-_\infty\cup\Sigma_0\cup C^+\cup C^-)$ and is bounded for bounded $\eta$.
\item $g$ takes continuous boundary values on the jump contour with the property that 
\begin{equation}
\eta\mapsto F(\eta,y):=\left(\frac{\partial g}{\partial\eta}(\eta,y)-\frac{\partial\Phi}{\partial\eta}(\eta,y)\right)^2
\label{eq:f-define}
\end{equation}
is continuous except at $\eta=0$.
\item There is a quantity $g_0(y)$ such that $g(\eta,y)=-\tfrac{1}{2}\log(-\ii\eta) + g_0(y) + o(1)$ as $\eta\to\infty$.
\item $g$ is independent of $n$.
\end{itemize}
Then we use such a function $g(\eta,y)$ to define a new unknown by 
\begin{equation}
\mathbf{M}^{(n)}(\eta,y):=\ee^{ng_0(y)\sigma_3}\mathbf{Z}^{(n)}(\eta,y)\ee^{- n g(\eta,y)\sigma_3}.
\label{eq:M-Z}
\end{equation}
Then $\eta\mapsto \mathbf{M}^{(n)}(\eta,y)$ is analytic where $\mathbf{Z}^{(n)}(\eta,y)$ is, satisfies the simplified normalization condition that $\mathbf{M}^{(n)}(\eta,y)\to\mathbb{I}$ as $\eta\to\infty$, has the property that  $\mathbf{M}^{(n)}(\eta,y)\widetilde{\mathbf{E}}^\pm(-\ii\eta)^{-(-1)^n\sigma_3/4}$ has a common limit as $\eta\to 0$ from $D_0^\pm$, and has jump conditions explicitly related to those satisfied by $\mathbf{Z}^{(n)}(\eta,y)$ that involve the boundary values of $g$.  We will explain these jump conditions later.

But first, we consider the function $F(\eta,y)$ defined by \eqref{eq:f-define}.  Obviously, this function is not only continuous for $\eta\neq 0$, but it is also analytic for $\eta\in\mathbb{C}\setminus\{0\}$.  It is therefore determined by its asymptotic behavior as $\eta\to 0$ and $\eta\to\infty$.  Since then $\partial g/\partial\eta$ has to have Laurent expansions in both limits in powers of $(-\ii\eta)^{1/2}$ for $F$ to be analytic, we interpret the conditions on $g(\eta,y)$ near $\eta=0,\infty$ in terms of $\partial g/\partial \eta$ as follows:
\begin{equation}
\frac{\partial g}{\partial\eta}(\eta,y) =\begin{cases} -\frac{1}{2\eta} + a_1(-\ii\eta)^{-3/2} + a_2(-\ii\eta)^{-2} + \cdots,&\quad \eta\to\infty\\
b_1(-\ii\eta)^{-1/2} + b_2 + \cdots,&\quad\eta\to 0.
\end{cases}
\end{equation}
It then follows from \eqref{eq:Phi-define} and \eqref{eq:f-define} that 
\begin{equation}
F(\eta,y)=\begin{cases}
\left(-\ii -\frac{1}{2\eta} + \mathcal{O}(\eta^{-3/2})\right)^2 = -1 +\frac{\ii}{\eta} + \mathcal{O}(\eta^{-2}),&\quad\eta\to\infty,\\
\left(\frac{ 1}{2}\ii y(-\ii\eta)^{-3/2} + \mathcal{O}(\eta^{-1/2})\right)^2 = -\frac{y^2}{4}(-\ii\eta)^{-3} + \mathcal{O}(\eta^{-2}),&\quad\eta\to 0.
\end{cases}
\end{equation}
Therefore, $F(\eta,y)$ is a Laurent polynomial of the form
\begin{equation}
F(\eta,y)= -1+\frac{1}{\mu}+\frac{c}{\mu^2} -\frac{y^2}{4\mu^3},\quad\mu:=-\ii\eta,
\end{equation}
involving an undetermined coefficient $c$.  Writing $F(\eta,y)$ in the form $F(\eta,y)=P(-\ii\eta,y)(-\ii\eta)^{-3}$, where $P(\mu,y)$ is the  cubic polynomial\footnote{Note that this polynomial is closely related to $\frac{16}{y}U^3+2EU^2-\frac{4}{y}U+1$ appearing in \eqref{eq:nonequilibrium-elliptic}.  Indeed, 
\[
-\frac{64 U^3}{y^3}P\left(\frac{y}{2U};y\right) = \frac{16}{y}U^3-\frac{16c}{y^2}U^2-\frac{4}{y}U+1
\]
which matches the target form if we identify the integration constants by $y^2E=-8c$.}
\begin{equation}
P(\mu,y):=-\mu^3+\mu^2+c\mu-\frac{y^2}{4}.
\label{eq:spectral-polynomial}
\end{equation} 
We now assume that $c$ is chosen so that $P(\mu,y)$ has a double root $\mu=d$ and a simple root $\mu=s$.  Expanding $P(\mu,y)=-(\mu-d)^2(\mu-s)$ and matching the coefficients with \eqref{eq:spectral-polynomial} gives three equations:
\begin{equation}
2d+s=1, \quad d^2+2ds=-c, \quad\text{and}\quad d^2s=-\tfrac{1}{4}y^2.  
\label{eq:matching-coefficients}
\end{equation}
Combining the first equation with the last, one eliminates $d$ and obtains a cubic equation for $s$:  $s(s-1)^2=-y^2$.  Given any solution of this equation, $d$ and $c$ are determined explicitly from the first and second equations.  It is clear that as long as $y>0$, there exists one real negative root and a complex-conjugate pair of roots.  We select the real solution (a guess based on symmetry to be justified later); thus $s$ is decreasing from $s=0$ as $y$ increases from $y=0$, with asymptotic behavior $s=-y^{2/3}(1+o(1))$ as $y\to+\infty$.  Then from $2d+s=1$ we have that $d$ is increasing from $d=\tfrac{1}{2}$ as $y$ increases from $y=0$, with asymptotic behavior $d=\frac{1}{2}y^{2/3}(1+o(1))$ as $y\to+\infty$.  We observe that the discriminant of the cubic $s(s-1)^2=-y^2$ with respect to $s$ is
\begin{equation}
y^2(27y^2+4)^3=0,
\label{eq:cornerpoints}
\end{equation}
which should be compared with the discriminant of \eqref{eq:U-equilibrium}.

Now we discuss the jump conditions satisfied by $\mathbf{M}^{(n)}(\eta,y)$.  These take the form shown in \eqref{eq:Z-jump-first}--\eqref{eq:Z-jump-last} with the only change being that the function $\Phi(\eta,y)$ in the exponents is replaced with $-h(\eta,y)$, where
\begin{equation}
h(\eta,y):=g(\eta,y)-\Phi(\eta,y).
\end{equation}
The effect of the conjugation of the constant jump matrices involving the exponent $h(\eta,y)$ depends on the sign of $\mathrm{Re}(h(\eta,y))$.  Under the assumption that $P(\mu,y)=-(\mu-d)^2(\mu-s)$, with $s=s(y)<0$ and $d=(1-s)/2>0$, we have from \eqref{eq:f-define} that 
\begin{equation}
\frac{\partial h}{\partial\eta}(\eta,y) =\sqrt{f(\eta,y)} =  -(\eta-\ii d)(-\ii\eta)^{-3/2}(-\ii\eta-s)^{1/2}.
\label{eq:hprime}
\end{equation}
Here, all fractional powers refer to the principal branch, and the sign of the square root was chosen to match the asymptotic behavior at $\eta=\infty$ according to the definitions of $h(\eta,y)$, $\Phi(\eta,y)$, and the large-$\eta$ behavior of $g(\eta,y)$.  This function is analytic except on the negative imaginary axis between $\eta=0$ and $\eta=\ii s$, because the sign changes of the latter two factors cancel for $\eta$ below $\ii s$ on the imaginary axis. 

Integrating using $d=(1-s)/2$, 
\begin{equation}
h(\eta,y)=(-\ii\eta-s+1)(-\ii\eta-s)^{1/2}(-\ii\eta)^{-1/2}+\frac{1}{2}\log\left(\frac{(-\ii\eta-s)^{1/2}-(-\ii\eta)^{1/2}}{(-\ii\eta-s)^{1/2}+(-\ii\eta)^{1/2}}\right) + C
\label{eq:h-log}
\end{equation}
where $C$ is an integration constant, which we take to be $C=0$.  Then one can check that $h_+(\eta,y)+h_-(\eta,y)\equiv 0$ for $\eta$ on the segment between $\eta=0$ and $\eta=\ii s$.  Also, $g(\eta,y)=h(\eta,y)+\Phi(\eta,y)$, so
\begin{equation}
g(\eta,y)=(-\ii\eta-s+1)(-\ii\eta-s)^{1/2}(-\ii\eta)^{-1/2}+\frac{1}{2}\log\left(\frac{(-\ii\eta-s)^{1/2}-(-\ii\eta)^{1/2}}{(-\ii\eta-s)^{1/2}+(-\ii\eta)^{1/2}}\right)
+\ii \eta-y(-\ii\eta)^{-1/2},
\end{equation}
which indeed has the expansion
\begin{equation}
g(\eta,y)=-\frac{1}{2}\log(-\ii\eta)+g_0(y)+\mathcal{O}(\eta^{-1/2}),\quad\eta\to\infty,\quad g_0(y):=1-\frac{3}{2}s+\frac{1}{2}\ln\left(-\frac{1}{4}s\right),\quad s=s(y)<0.
\end{equation}
Also, using $s<0$ and $d^2s=-\tfrac{1}{4}y^2$ shows that 
\begin{equation}
g(\eta,y)=-\frac{1+3s}{2(-s)^{1/2}}(-\ii\eta)^{1/2} + \mathcal{O}(\eta),\quad\eta\to 0.
\label{eq:g-expand-origin}
\end{equation}

Assuming that $y>0$ is sufficiently large, the function $\mathrm{Re}(h(\eta,y))$ is continuous except at $\eta=0$, harmonic except on the negative imaginary segment connecting the origin with $\eta=\ii s$, and has a zero level set consisting of that same segment and two curves emanating from $\eta=\ii s$ into the left and right half-planes and extending to $\infty$.  These curves are reflections of each other in the imaginary axis, and all three branches emanating from $\eta=\ii s$ are separated by angles of $2\pi/3$ at that junction point.  We have $\mathrm{Re}(h(\eta,y))>0$  above the two curves (a region containing the positive imaginary axis) and $\mathrm{Re}(h(\eta,y))<0$ below them.  See the left-hand panel of Figure~\ref{fig:sign-charts} below.

\subsection{Parametrix construction and error analysis}

To make use of this structure, we assume that the jump contour is deformed so that $\eta=\ii s$ is the junction point of $C^+$, $C^-$, $\Sigma_0^-$, and $\Sigma_\infty^-$, and so that $\Sigma_0^-$ coincides with the segment joining $\eta=\ii s$ with the origin.  Then, since the sum of the boundary values of $\partial h/\partial\eta(\eta,y)$ vanishes and $\mathrm{Re}(h(\eta,y))=0$ on $\Sigma_0^-$, the jump condition of $\mathbf{M}^{(n)}(\eta,y)$ on this arc reads
\begin{equation}
\mathbf{M}^{(n)}_+(\eta,y)=\mathbf{M}^{(n)}_-(\eta,y)\begin{bmatrix}0&\ee^{\ii\phi}\\-\ee^{-\ii\phi} & 0\end{bmatrix},\quad\phi:=\frac{\pi}{2}+n\pi -\ii n(h_+(\eta,y)+h_-(\eta,y))=\frac{\pi}{2}+n\pi,\quad\eta\in\Sigma_0^-,
\label{eq:twist-jump}
\end{equation}
and here, $\phi$ is a real phase that is constant along $\Sigma_0^-$.  Because $h(\eta,y)$ is analytic on all other arcs of the jump contour except for $\Sigma_\infty^-$, along which $h_+(\eta,y)-h_-(\eta,y)=-\ii\pi$, all other jump matrices for $\mathbf{M}^{(n)}(\eta,y)$ are exponentially small perturbations of the identity matrix when $n\to+\infty$, estimates that are uniform with respect to $\eta$ except in a neighborhood of $\eta=\ii s$.  Here it is possible to install a standard Airy parametrix to solve the jump conditions exactly locally.  To construct an outer parametrix, we solve the jump condition \eqref{eq:twist-jump} exactly and build in suitable singularities near the two endpoints of $\Sigma_0^-$ to allow matching onto the Airy parametrix at $\eta=\ii s$ and to match the required behavior of $\mathbf{M}^{(n)}(\eta,y)$ near the origin.

\subsubsection{Outer parametrix}
By definition, the outer parametrix is the matrix $\dot{\mathbf{M}}^{(n),\mathrm{out}}(\eta,y)$ with the following properties:
\begin{itemize}
\item $\dot{\mathbf{M}}^{(n),\mathrm{out}}(\eta,y)$ is analytic for $\eta\in\mathbb{C}\setminus\Sigma_0^-$, taking continuous boundary values from each side except at the endpoints.
\item The boundary values are related by the jump condition \eqref{eq:twist-jump}.
\item $\dot{\mathbf{M}}^{(n),\mathrm{out}}(\eta,y)\to\mathbb{I}$ as $\eta\to\infty$.
\item $\dot{\mathbf{M}}^{(n),\mathrm{out}}(\eta,y) = \mathcal{O}(1+|\eta|^{-1/4}+|\eta-\ii s|^{-1/4})$.
\end{itemize}
We construct $\dot{\mathbf{M}}^{(n),\mathrm{out}}(\eta,y)$ explicitly by first removing the phase factors $\ee^{\pm\ii\phi}$ from the jump condition without changing any of the other conditions; we write $\dot{\mathbf{M}}^{(n),\mathrm{out}}(\eta,y)=\ee^{\ii\phi\sigma_3/2}\mathbf{N}(\eta,y)\ee^{-\ii\phi\sigma_3/2}$ so that \eqref{eq:twist-jump} for $\dot{\mathbf{M}}^{(n),\mathrm{out}}(\eta,y)$ implies
\begin{equation}
\mathbf{N}_+(\eta,y)=\mathbf{N}_-(\eta,y)\ii\sigma_2,\quad\eta\in\Sigma_0^-.
\end{equation}
Next, we diagonalize the jump matrix $\ii\sigma_2$ by the substitution (again, not changing any of the other properties)
\begin{equation}
\mathbf{N}(\eta,y)=\frac{1}{\sqrt{2}}\begin{bmatrix}1&\ii\\\ii & 1\end{bmatrix}\mathbf{O}(\eta,y)\frac{1}{\sqrt{2}}\begin{bmatrix}1&-\ii\\-\ii & 1\end{bmatrix}.
\end{equation}
Then $\mathbf{O}(\eta,y)$ is analytic for $\eta\in\mathbb{C}\setminus\Sigma_0^-$, takes continuous boundary values except at the endpoints where $-1/4$ power singularities are admissible, tends to $\mathbb{I}$ as $\eta\to\infty$, and satisfies the diagonal jump condition $\mathbf{O}_+(\eta,y)=\mathbf{O}_-(\eta,y)\ii^{\sigma_3}$ for $\eta\in\Sigma_0^-$ (note that the dependence on $y>0$ enters via the moving endpoint $\eta=\ii c$).  The unique solution of these conditions is the diagonal matrix
\begin{equation}
\mathbf{O}(\eta,y):=(-\ii(\eta-\ii s))^{\sigma_3/4}(-\ii\eta)^{-\sigma_3/4} = \left(\frac{\eta-\ii s}{\eta}\right)^{\sigma_3/4}
\end{equation}
where in each case the principal branch power is intended.  Inverting the transformations $\dot{\mathbf{M}}^{(n),\mathrm{out}}(\eta,y)\mapsto\mathbf{N}(\eta,y)\mapsto\mathbf{O}(\eta,y)$ gives the explicit formula for the outer parametrix as
\begin{equation}
\dot{\mathbf{M}}^{(n),\mathrm{out}}(\eta,y):=\ee^{\ii\phi\sigma_3/2}\frac{1}{\sqrt{2}}\begin{bmatrix}1&\ii\\\ii & 1\end{bmatrix}\left(\frac{\eta-\ii s}{\eta}\right)^{\sigma_3/4}\frac{1}{\sqrt{2}}\begin{bmatrix}1&-\ii\\-\ii & 1\end{bmatrix}\ee^{-\ii\phi\sigma_3/2}.
\label{eq:M-outer}
\end{equation}

A calculation shows that $\dot{\mathbf{M}}^{(n),\mathrm{out}}(\eta,y)\widetilde{\mathbf{E}}\cdot(-\ii\eta)^{-(-1)^n\sigma_3/4}$ is analytic at $\eta=0$.  Indeed, letting $\mathbf{C}^{(n)}$ denote the matrix
\begin{equation}
\mathbf{C}^{(n)}:=\begin{cases}-\ii\sigma_1,&\quad\text{$n$ even,}\\
\mathbb{I},&\quad\text{$n$ odd,}\end{cases}
\end{equation}
we see that
\begin{multline}
\dot{\mathbf{M}}^{(n),\mathrm{out}}(\eta,y)\widetilde{\mathbf{E}}\cdot(-\ii\eta)^{-(-1)^n\sigma_3/4}\\
\begin{aligned}
&=\ee^{\ii\phi\sigma_3/2}\frac{1}{\sqrt{2}}\begin{bmatrix}
1 & \ii\\\ii & 1\end{bmatrix}\left(\frac{\eta-\ii s}{\eta}\right)^{\sigma_3/4}\mathbf{C}^{(n)}\cdot(-\ii\eta)^{-(-1)^n\sigma_3/4}\ee^{5\pi\ii\sigma_3/6}\ee^{-\ii\phi\sigma_3/2}\\
&=\ee^{\ii\phi\sigma_3/2}\frac{1}{\sqrt{2}}\begin{bmatrix}1&\ii\\\ii & 1\end{bmatrix}(-\ii(\eta-\ii s))^{\sigma_3/4}(-\ii\eta)^{-\sigma_3/4}\mathbf{C}^{(n)}\cdot
(-\ii\eta)^{-(-1)^n\sigma_3/4}\ee^{5\pi\ii\sigma_3/6}\ee^{-\ii\phi\sigma_3/2}\\
&=\ee^{\ii\phi\sigma_3/2}\frac{1}{\sqrt{2}}\begin{bmatrix}1&\ii\\\ii & 1\end{bmatrix}(-\ii(\eta-\ii s))^{\sigma_3/4}\mathbf{C}^{(n)}\ee^{5\pi\ii\sigma_3/6}\ee^{-\ii\phi\sigma_3/2},
\end{aligned}
\end{multline}
and having a vertical branch cut emanating downward from $\eta=\ii s$ with $s<0$, the diagonal matrix $(-\ii(\eta-\ii s))^{\sigma_3/4}$ is analytic at $\eta=0$.  Since $\mathbf{C}^{(n)}$ is either diagonal or off-diagonal, we can further simplify the resulting formula as follows:
\begin{equation}
\begin{split}
\dot{\mathbf{M}}^{(n),\mathrm{out}}(\eta,y)\widetilde{\mathbf{E}}\cdot(-\ii\eta)^{-(-1)^n\sigma_3/4}&=\ee^{\ii\phi\sigma_3/2}\frac{1}{\sqrt{2}}\begin{bmatrix}1&\ii\\
\ii & 1\end{bmatrix}\mathbf{C}^{(n)}\ee^{5\pi\ii\sigma_3/6}\ee^{-\ii\phi\sigma_3/2}(-\ii(\eta-\ii s))^{-(-1)^n\sigma_3/4}\\
&=\frac{1}{\sqrt{2}}\begin{bmatrix}1&1\\-1 & 1\end{bmatrix}\ee^{5\pi\ii\sigma_3/6}(-\ii(\eta-\ii s))^{-(-1)^n\sigma_3/4},
\end{split}
\end{equation}
where we also used the $n$-dependent definition of $\phi$ in \eqref{eq:twist-jump}.  In particular, this implies that
\begin{equation}
\lim_{\eta\to 0}\dot{\mathbf{M}}^{(n),\mathrm{out}}(\eta,y)\widetilde{\mathbf{E}}\cdot(-\ii\eta)^{-(-1)^n\sigma_3/4} = \frac{1}{\sqrt{2}}\begin{bmatrix}1&1\\-1&1\end{bmatrix}\ee^{5\pi\ii\sigma_3/6}(-s)^{-(-1)^n\sigma_3/4}.
\label{eq:outer-limit-at-origin}
\end{equation}
\subsubsection{Airy parametrix}
Since $h_+(\ii s,y)+h_-(\ii s,y)=0$,  $h_+(\eta,y)+h_-(\eta,y)$ is real-valued and strictly decreasing in the negative imaginary direction along $\Sigma_\infty^-$, and  $h'(\eta,y)$ is an analytic function on this contour arc that vanishes like a square root at $\eta=\ii s$, we may introduce a conformal coordinate $\eta\mapsto W(\eta)$ on a neighborhood $D_{\ii s}$ of $\eta=\ii s$ with the properties that $W(\ii s)=0$, $W(\eta)>0$ for $\eta\in \Sigma_\infty^-\cap D_{\ii s}$, and $W(\eta)=-(h_+(\eta,y)+h_-(\eta,y))^{2/3}$ for $\eta\in \Sigma_\infty^-\cap D_{\ii s}$.  Using $h_+(\eta,y)-h_-(\eta,y)=-\ii\pi$ for $\eta\in\Sigma_\infty^-$, we see that on $C^\pm\cap D_{\ii s}$, the exponent function $2h(\eta,y)$ is the analytic continuation from $\Sigma_\infty^-\cap D_{\ii s}$ through $D_\pm^\infty\cap D_{\ii s}$ of $2h_\mp(\eta,y)=(h_+(\eta,y)+h_-(\eta,y))\mp (h_+(\eta,y)-h_-(\eta,y)) = h_+(\eta,y)+h_-(\eta,y)\pm \ii\pi$.  Rescaling the conformal coordinate by $\zeta=n^{2/3}W(\eta)$, the jump conditions for $\mathbf{M}^{(n)}(\eta,y)$ within $D_{\ii s}$ can be written as follows:
\begin{equation}
\begin{split}
\mathbf{M}^{(n)}_+(\eta,y)&=\mathbf{M}^{(n)}_-(\eta,y)\ee^{nh_-(\eta,y)\sigma_3}\begin{bmatrix}(-1)^n & -\ii (-1)^n\\0 & (-1)^n\end{bmatrix}\ee^{-nh_+(\eta,y)\sigma_3}\\
&=\mathbf{M}^{(n)}_-(\eta,y)\ee^{n(h_+(\eta,y)+h_-(\eta,y))\sigma_3/2}\ee^{\ii n\pi\sigma_3/2}\begin{bmatrix}(-1)^n & -\ii (-1)^n\\0 & (-1)^n\end{bmatrix}\ee^{\ii n\pi\sigma_3/2}\ee^{-n(h_+(\eta,y)+h_-(\eta,y))\sigma_3/2}\\
&=\mathbf{M}^{(n)}_-(\eta,y)\begin{bmatrix}1 & -\ii (-1)^n\ee^{n(h_+(\eta,y)+h_-(\eta,y))}\\0 & 1\end{bmatrix}\\
&=\mathbf{M}^{(n)}_-(\eta,y)\begin{bmatrix}1&-\ii (-1)^n\ee^{-\zeta^{3/2}}\\0&1\end{bmatrix}\\
&=\mathbf{M}^{(n)}_-(\eta,y)\ee^{\ii (-1)^n\pi\sigma_3/4}\begin{bmatrix}1&-\ee^{-\zeta^{3/2}}\\0&1\end{bmatrix}\ee^{-\ii (-1)^n\pi\sigma_3/4},\quad
\eta\in\Sigma_\infty^-\cap D_{\ii s};
\end{split}
\end{equation}
\begin{equation}
\begin{split}
\mathbf{M}^{(n)}_+(\eta,y)&=\mathbf{M}^{(n)}_-(\eta,y)\begin{bmatrix}1&0\\-\ii\ee^{-n(h_+(\eta,y)+h_-(\eta,y))} & 1\end{bmatrix}\\
&=\mathbf{M}^{(n)}_-(\eta,y)\begin{bmatrix}1&0\\-\ii (-1)^n\ee^{\zeta^{3/2}} & 1\end{bmatrix}\\
&=\mathbf{M}^{(n)}_-(\eta,y)\ee^{\ii (-1)^n\pi\sigma_3/4}\begin{bmatrix}1&0\\\ee^{\zeta^{3/2}} & 1\end{bmatrix}\ee^{-\ii (-1)^n\pi\sigma_3/4},\quad\eta\in C^\pm\cap D_{\ii s}; \quad\text{and}
\end{split}
\end{equation}
\begin{equation}
\begin{split}
\mathbf{M}^{(n)}_+(\eta,y)&=\mathbf{M}^{(n)}_-(\eta,y)\begin{bmatrix}0 & \ii (-1)^n\\\ii (-1)^n & 0\end{bmatrix}\\
&=\mathbf{M}^{(n)}_-(\eta,y)\ee^{\ii (-1)^n\pi\sigma_3/4}\begin{bmatrix}0 & 1\\-1 & 0\end{bmatrix}\ee^{-\ii (-1)^n\pi\sigma_3/4},\quad\eta\in \Sigma_0^-\cap D_{\ii s}.
\end{split}
\end{equation}
We deform the contours $C^\pm$ within $D_{\ii s}$ to lie along the rays $\arg(\zeta)=\pm 2\pi/3$ as shown in the left-hand panel of Figure~\ref{fig:error-jump} below.  Then one can check that $\mathbf{M}(\eta,y)\ee^{\ii (-1)^n\pi\sigma_3/4}\mathbf{A}(n^{2/3}W(\eta))^{-1}$ is holomorphic in $D_{\ii s}$, where $\mathbf{A}(\zeta)$ is the standard Airy parametrix as defined (for instance) by the solution of \cite[Riemann-Hilbert Problem 4]{BothnerM19}.  Defining a matrix holomorphic in $D_{\ii s}$ and uniformly bounded as $n\to\infty$ by 
\begin{equation}
\mathbf{H}(\eta):=\ee^{\ii (-1)^n\pi\sigma_3/4}\frac{1}{\sqrt{2}}\begin{bmatrix}1&\ii\\\ii & 1\end{bmatrix}\left(\frac{\eta-\ii s}{\eta}\right)^{\sigma_3/4}W(\eta)^{-\sigma_3/4},\quad\eta\in D_{\ii s},
\end{equation}
we take as the inner parametrix
\begin{equation}
\dot{\mathbf{M}}^{(n),\mathrm{in}}(\eta,y):=\mathbf{H}(\eta)n^{-\sigma_3/6}\mathbf{A}(n^{2/3}W(\eta))\ee^{-\ii (-1)^n\pi\sigma_3/4},\quad\eta\in D_{\ii s}.
\label{eq:M-inner}
\end{equation}
Then we can compare the inner and outer parametrices on $\partial D_{\ii s}$:
\begin{equation}
\dot{\mathbf{M}}^{(n),\mathrm{in}}(\eta,y)\dot{\mathbf{M}}^{(n),\mathrm{out}}(\eta,y)^{-1}=\mathbf{H}(\eta)n^{-\sigma_3/6}\mathbf{A}(\zeta)\frac{1}{\sqrt{2}}\begin{bmatrix}1&\ii\\\ii & 1\end{bmatrix}\zeta^{-\sigma_3/4}\cdot n^{\sigma_3/6}\mathbf{H}(\eta)^{-1},\quad\eta\in\partial D_{\ii s},
\end{equation}
where $\zeta=n^{2/3}W(\eta)$ and we used the fact that $\ee^{-\ii (-1)^n\sigma_3/4}\ee^{\ii\phi\sigma_3/2}$ is a multiple (by $\pm 1$) of the identity.  Since $\zeta$ is uniformly large when $\eta\in \partial D_{\ii s}$, we use the large-$\zeta$ asymptotic of $\mathbf{A}(\zeta)$ (see for instance \cite[Eqn.\@ (113)]{BothnerM19}):
\begin{equation}
\mathbf{A}(\zeta)\frac{1}{\sqrt{2}}\begin{bmatrix}1&\ii\\\ii & 1\end{bmatrix}\zeta^{-\sigma_3/4}=\mathbb{I}+\begin{bmatrix}\mathcal{O}(\zeta^{-3}) & \mathcal{O}(\zeta^{-1})\\\mathcal{O}(\zeta^{-2}) & \mathcal{O}(\zeta^{-3})\end{bmatrix},\quad\zeta\to\infty,
\end{equation}
which implies that
\begin{equation}
\sup_{\eta\in\partial D_{\ii s}}\|\dot{\mathbf{M}}^{(n),\mathrm{in}}(\eta,y)\dot{\mathbf{M}}^{(n),\mathrm{out}}(\eta,y)^{-1}-\mathbb{I}\|=\mathcal{O}(n^{-1}),\quad n\to\infty.
\label{eq:comparison-estimate}
\end{equation}

\subsubsection{Global parametrix and error estimation}
We define a global parametrix for $\mathbf{M}(\eta,y)$ by
\begin{equation}
\dot{\mathbf{M}}^{(n)}(\eta,y):=\begin{cases}\dot{\mathbf{M}}^{(n),\mathrm{in}}(\eta,y),&\quad \eta\in D_{\ii s}\\
\dot{\mathbf{M}}^{(n),\mathrm{out}}(\eta,y),&\quad\eta\in\mathbb{C}\setminus\overline{D_{\ii s}}.
\end{cases}
\label{eq:global-parametrix}
\end{equation}
The mismatch between the global parametrix and $\mathbf{M}^{(n)}(\eta,y)$ is defined as 
\begin{equation}
\mathbf{D}^{(n)}(\eta,y):=\mathbf{M}^{(n)}(\eta,y)\dot{\mathbf{M}}^{(n)}(\eta,y)^{-1}.
\label{eq:discrepancy}
\end{equation}
Since the inner parametrix is an exact solution of the jump conditions for $\mathbf{M}^{(n)}(\eta,y)$ within $D_{\ii s}$, $\mathbf{D}^{(n)}(\eta,y)$ may be regarded as being analytic for $\eta\in D_{\ii s}$.  Similarly, because the outer parametrix is an exact solution of the jump condition for $\mathbf{M}^{(n)}(\eta,y)$ on the arc $\Sigma_0^-$, $\mathbf{D}^{(n)}(\eta,y)$ has no jump across this contour either.  Therefore $\mathbf{D}^{(n)}(\eta,y)$ is analytic in the complement of the jump contour illustrated in the right-hand panel of Figure~\ref{fig:error-jump}.
\begin{figure}[h]
\includegraphics{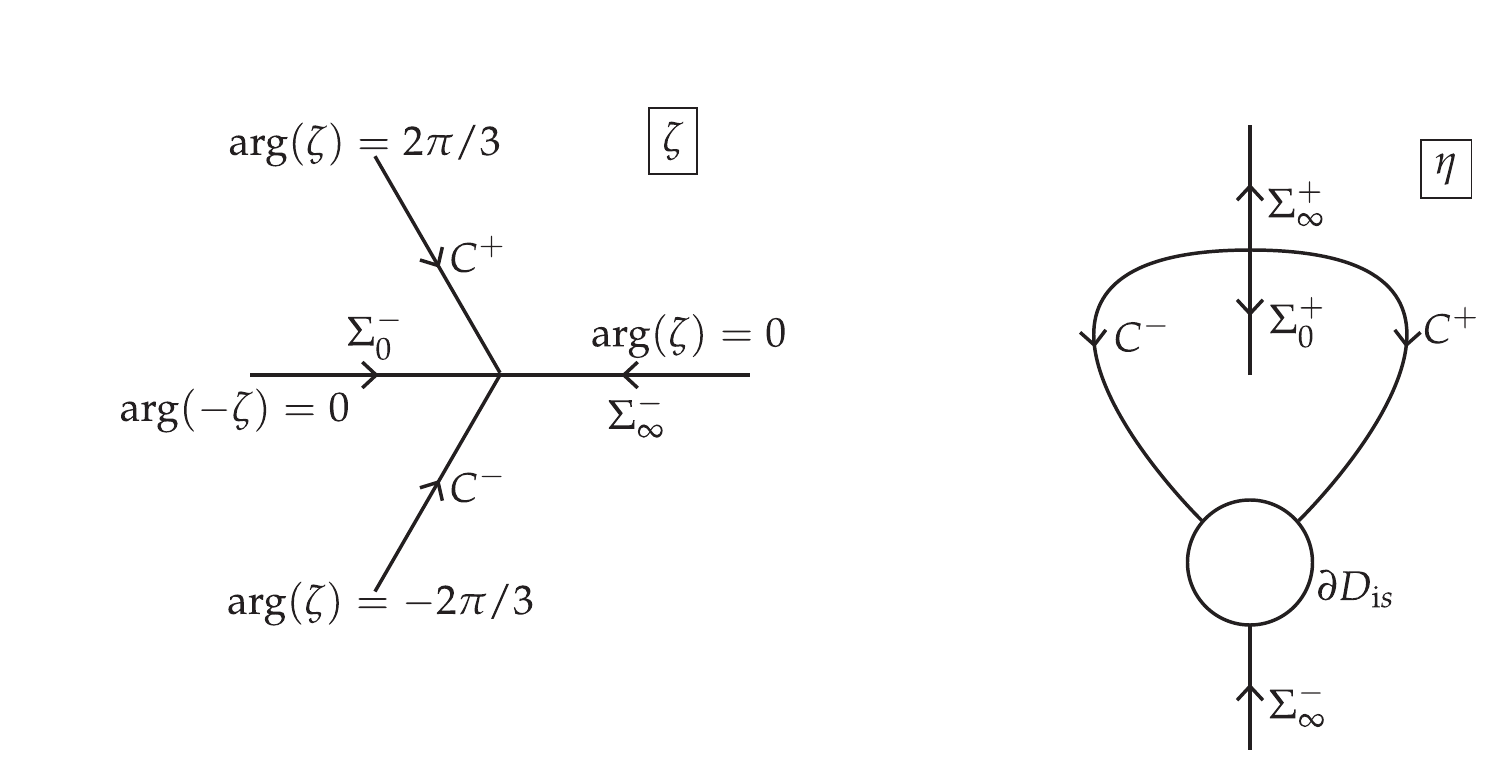}
\caption{Left:  the jump contour near $\eta=\ii s$ in the $\zeta$-plane.  Right:  the jump contour for $\mathbf{D}^{(n)}(\eta,y)$}
\label{fig:error-jump}
\end{figure}
On all arcs of the jump contour except for $\partial D_{\ii s}$ we have $\dot{\mathbf{M}}^{(n)}(\eta,y)=\dot{\mathbf{M}}^{(n),\mathrm{out}}(\eta,y)$, and the outer parametrix $\dot{\mathbf{M}}^{(n),\mathrm{out}}(\eta,y)$ is analytic; therefore on those arcs,
\begin{equation}
\mathbf{D}^{(n)}_+(\eta,y)=\mathbf{D}^{(n)}_-(\eta,y)\dot{\mathbf{M}}^{(n),\mathrm{out}}(\eta,y)\cdot[\mathbf{M}^{(n)}_-(\eta,y)^{-1}\mathbf{M}^{(n)}_+(\eta,y)]\cdot\dot{\mathbf{M}}^{(n),\mathrm{out}}(\eta,y)^{-1}.
\end{equation}
Because $\mathrm{Re}(h(\eta,y))\ge \delta>0$ holds on the parts of $C^\pm$, $\Sigma_0^+$, and $\Sigma_\infty^+$ outside of $D_{\ii s}$, while $\mathrm{Re}(h(\eta,y))\le -\delta <0$ holds on the part of $\Sigma_\infty^-$ outside of $D_{\ii s}$, we have $\mathbf{M}^{(n)}_-(\eta,y)^{-1}\mathbf{M}^{(n)}_+(\eta,y)-\mathbb{I}$ decays exponentially to zero as $n\to\infty$ on these arcs, so since the outer parametrix has unit determinant and is bounded on these contours, uniformly so as $n\to\infty$, we have $\mathbf{D}^{(n)}_+(\eta,y)=\mathbf{D}^{(n)}_-(\eta,y)(\mathbb{I}+\text{exponentially small})$ as $n\to\infty$, where the exponentially small term is measured in both $L^\infty$ and $L^2$.  On $\partial D_{\ii s}$ taken with clockwise orientation, we have
\begin{equation}
\mathbf{D}^{(n)}_+(\eta,y)=\mathbf{D}^{(n)}_-(\eta,y)\dot{\mathbf{M}}^{(n),\mathrm{in}}(\eta,y)\dot{\mathbf{M}}^{(n),\mathrm{out}}(\eta,y)^{-1},\quad\eta\in\partial D_{\ii s}
\end{equation}
because $\mathbf{M}^{(n)}(\eta,y)$ has no jump across this contour.  Therefore from \eqref{eq:comparison-estimate} it follows that $\mathbf{D}^{(n)}_+(\eta,y)=\mathbf{D}^{(n)}_-(\eta,y)(\mathbb{I}+\mathcal{O}(n^{-1}))$ holds uniformly on $\partial D_{\ii s}$ as $n\to\infty$.  Finally, we note that both factors in the definition of $\mathbf{D}^{(n)}(\eta,y)$ tend to the identity as $\eta\to\infty$, so $\mathbf{D}^{(n)}(\eta,y)$ does as well.  

The matrix $\mathbf{D}^{(n)}(\eta,y)$ therefore satisfies the conditions of a Riemann-Hilbert problem of small-norm type, and the key implication we need of this fact is that $\mathbf{D}^{(n)}(\eta,y)$ is continuous at $\eta=0$ and $\mathbf{D}^{(n)}(\eta,y)=\mathbb{I}+\mathcal{O}(n^{-1})$ holds uniformly in the $\eta$-plane, and in particular in a neighborhood of $\eta=0$.

\subsection{Asymptotic formula for $u_n(n^{3/2}y)$ with $y>0$ sufficiently large}
According to \eqref{eq:Z-origin-limit},
\begin{equation}
(-\ii n y)^{-n\sigma_3}\mathbf{B}_0^{(n)}(x)(-\ii n y)^{-(-1)^n\sigma_3/2}=\lim_{\eta\to 0}\mathbf{Z}^{(n)}(\eta,y)\ee^{-\ii n\eta\sigma_3}\widetilde{\mathbf{E}}\cdot(-\ii\eta)^{-(-1)^n\sigma_3/4},\quad x=n^{3/2}y.
\end{equation}
Using \eqref{eq:M-Z} and \eqref{eq:global-parametrix}--\eqref{eq:discrepancy} and the fact that $\eta=0$ lies in the exterior of $D_{\ii s}$, this can be written as
\begin{multline}
(-\ii n y)^{-n\sigma_3}\mathbf{B}_0^{(n)}(x)(-\ii n y)^{-(-1)^n\sigma_3/2}\\
{}=\ee^{-ng_0(y)\sigma_3}\lim_{\eta\to 0}\mathbf{D}^{(n)}(\eta,y)\dot{\mathbf{M}}^{(n),\mathrm{out}}(\eta,y)\ee^{ng(\eta,y)\sigma_3}
\ee^{-\ii n\eta\sigma_3}\widetilde{\mathbf{E}}\cdot(-\ii\eta)^{-(-1)^n\sigma_3/4}.
\end{multline}
Taking into account the continuity of $\mathbf{D}^{(n)}(\eta,y)$ at $\eta=0$ and the behavior of the outer parametrix near $\eta=0$ as given by \eqref{eq:outer-limit-at-origin}, 
we have
\begin{multline}
(-\ii n y)^{-n\sigma_3}\mathbf{B}_0^{(n)}(x)(-\ii n y)^{-(-1)^n\sigma_3/2}\\
{}=\ee^{-ng_0(y)\sigma_3}\mathbf{D}^{(n)}(0,y)\frac{1}{\sqrt{2}}\begin{bmatrix}1&1\\-1&1\end{bmatrix}\ee^{5\pi\ii\sigma_3/6}(-s)^{-(-1)^n\sigma_3/4}\\
{}\cdot \lim_{\eta\to 0}(-\ii\eta)^{(-1)^n\sigma_3/4}\widetilde{\mathbf{E}}^{-1}\ee^{ng(\eta,y)\sigma_3}\ee^{-\ii n\eta\sigma_3}\widetilde{\mathbf{E}}\cdot(-\ii\eta)^{-(-1)^n\sigma_3/4}.
\label{eq:limit-penultimate}
\end{multline}
The conjugating factors in the limit on the last line will produce a singularity proportional to $(-\ii\eta)^{-1/2}$, so it is necessary to capture terms proportional to $(-\ii\eta)^{1/2}$ in $\widetilde{\mathbf{E}}^{-1}\ee^{ng(\eta,y)\sigma_3}\ee^{-\ii n\eta\sigma_3}\widetilde{\mathbf{E}}$.  Obviously $\ee^{-\ii n\eta\sigma_3}=\mathbb{I}+\mathcal{O}(\eta)$, however according to \eqref{eq:g-expand-origin}, 
\begin{equation}
\ee^{ng(\eta,y)\sigma_3}=\mathbb{I}-n\frac{1+3s}{2(-s)^{1/2}}(-\ii\eta)^{1/2}\sigma_3 +\mathcal{O}(\eta),\quad\eta\to 0.
\end{equation}
Since $\widetilde{\mathbf{E}}^{-1}\sigma_3\widetilde{\mathbf{E}}=\ee^{-5\pi\ii\sigma_3/6}\sigma_1\ee^{5\pi\ii\sigma_3/6}$, it then follows that
\begin{equation}
(-\ii\eta)^{(-1)^n\sigma_3/4}\widetilde{\mathbf{E}}^{-1}\ee^{ng(\eta,y)\sigma_3}\ee^{-\ii n\eta\sigma_3}\widetilde{\mathbf{E}}(-\ii\eta)^{-(-1)^n\sigma_3/4} = 
\begin{cases}\begin{bmatrix}1 & 0\\-n\frac{1+3s}{2(-s)^{1/2}}\ee^{5\pi\ii/3} & 1\end{bmatrix}+\mathcal{O}(\eta^{1/2}),&\quad\text{$n$ even}\\
\begin{bmatrix}1&-n\frac{1+3s}{2(-s)^{1/2}}\ee^{-5\pi\ii/3}\\0&1\end{bmatrix}+\mathcal{O}(\eta^{1/2}),&\quad\text{$n$ odd.}
\end{cases}
\end{equation}
Finally, we may take the limit in \eqref{eq:limit-penultimate}, and we obtain
\begin{multline}
\mathbf{B}_0^{(n)}(x)=(-\ii n y)^{n\sigma_3}\ee^{-ng_0(y)\sigma_3}\mathbf{D}^{(n)}(0,y)\frac{1}{\sqrt{2}}\begin{bmatrix}1&1\\-1 & 1\end{bmatrix}\\
{}\cdot
\begin{cases}\begin{bmatrix}1&0\\-\frac{1}{2}n(1+3s) & 1\end{bmatrix}\ee^{5\pi\ii\sigma_3/6}(-s)^{-\sigma_3/4}(-\ii n y)^{\sigma_3/2},&\quad\text{$n$ even} \vspace{.02in}\\
\begin{bmatrix}1&-\frac{1}{2}n(1+3s) \\0&1\end{bmatrix}\ee^{5\pi\ii\sigma_3/6}(-s)^{\sigma_3/4}(-\ii n y)^{-\sigma_3/2},&\quad\text{$n$ odd.}
\end{cases}
\end{multline}
With $\mathbf{B}^{(n)}_0(x)$ computed, we can use \eqref{eq:un-formula-Y} to write a formula for the algebraic solution $u_n(x)$ of the Painlev\'e-III (D7) equation \eqref{eq:D7KV}.  Since the product of second-column (respectively first-column) elements of $\mathbf{B}_0^{(n)}(x)$ appears in the formula for $n$ even (respectively $n$ odd), the diagonal prefactors $(-\ii n y)^{n\sigma_3}\ee^{-ng_0(y)\sigma_3}$ do not play any role, nor does the term $-\tfrac{1}{2}n(1+3s)$.  Using the fact that $\mathbf{D}^{(n)}(0,y)=\mathbb{I}+\mathcal{O}(n^{-1})$, the result is that an asymptotic formula of the same form holds in both cases:
\begin{equation}
u_n(x) = n^{1/2}\frac{1}{2}(-s)^{1/2} + \mathcal{O}(n^{-1/2}),\quad n\to\infty,\quad s=s(y)<0.
\label{eq:un-asymp-formula}
\end{equation}
We may observe from \eqref{eq:matching-coefficients} that $\tfrac{1}{2}(-s)^{1/2}=y/(4d)$ and $d^2(1-2d)=-\tfrac{1}{4}y^2$, from which we can derive the cubic equation
\begin{equation}
8\left[\frac{1}{2}(-s)^{1/2}\right]^3 + 2\left[\frac{1}{2}(-s)^{1/2}\right]-y = 0.
\label{eq:equilibrium-cubic-from-RHP}
\end{equation}
Comparing with \eqref{eq:U-equilibrium} proves the following.
\begin{theorem}
There exists $y_\mathrm{c}>0$ such that the following asymptotic formula holds:
\begin{equation}
u_n(n^{3/2}y)=n^{1/2}U + \mathcal{O}(n^{-1/2}),\quad n\to\infty,\quad y>y_\mathrm{c},
\end{equation}
where $U=U(y)$ is the positive real solution of the equilibrium cubic \eqref{eq:U-equilibrium}, and the error estimate is valid pointwise for $y>y_\mathrm{c}$ as well as uniformly for $y\ge y_\mathrm{c}+\delta$ for any $\delta>0$.
\label{thm:positive-exterior}
\end{theorem}
\begin{figure}[h]
\includegraphics[width=0.48\linewidth]{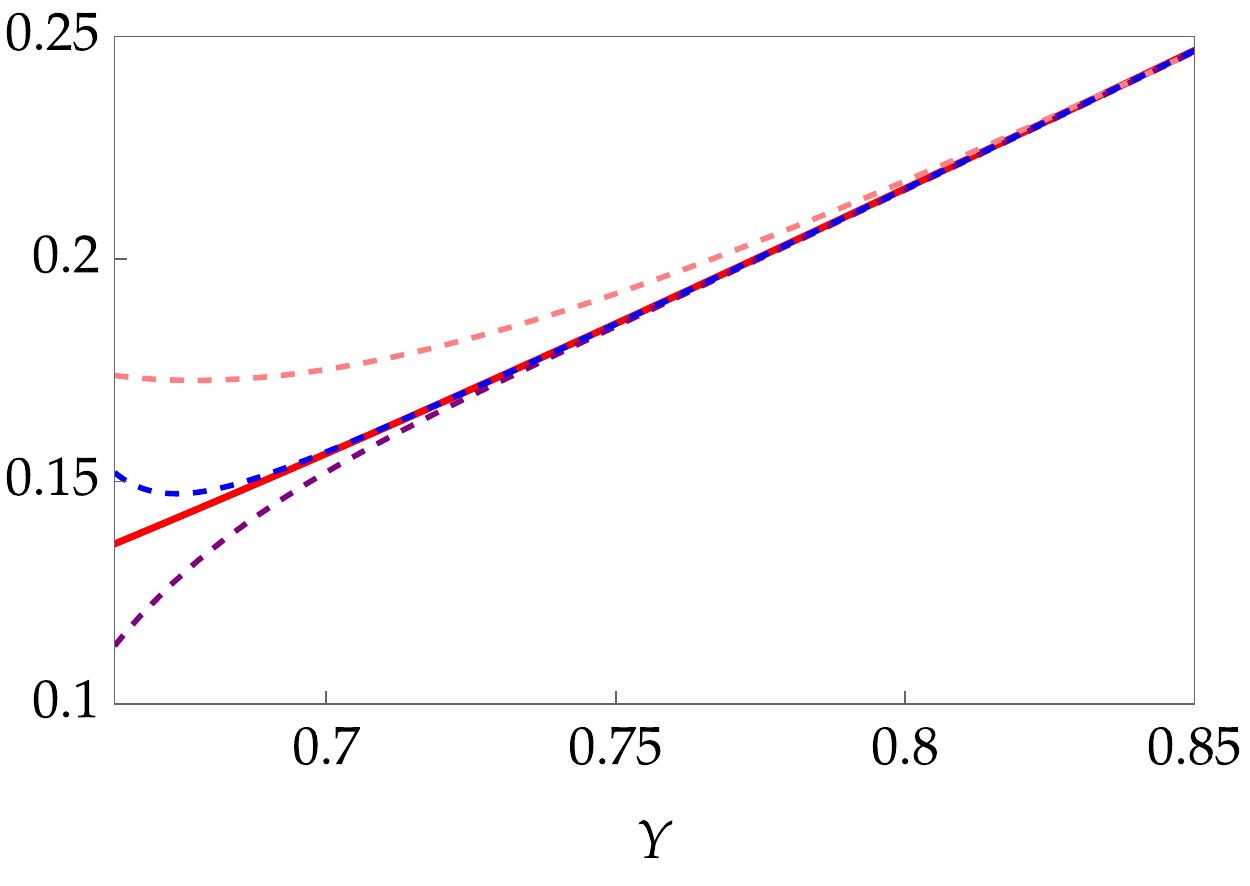}\hfill%
\includegraphics[width=0.48\linewidth]{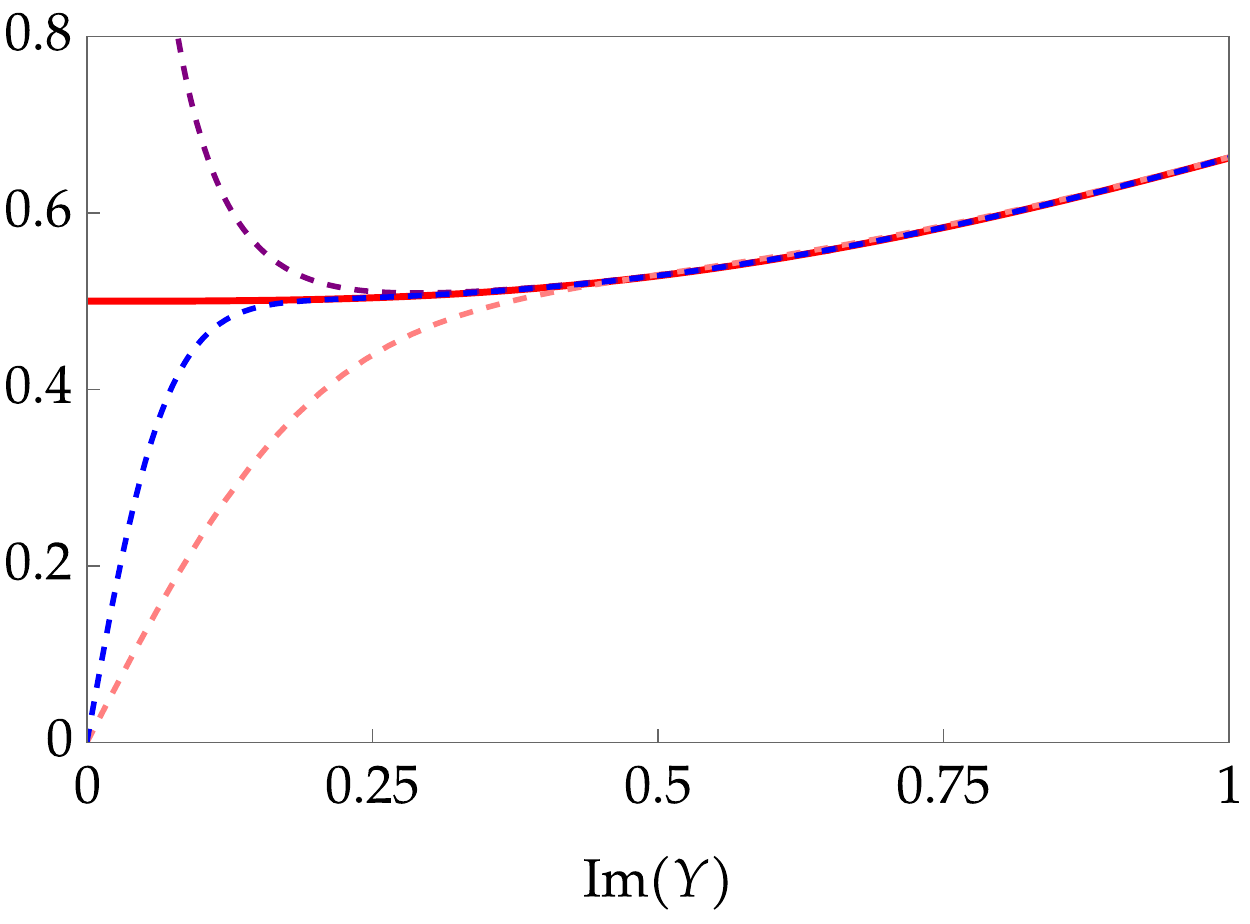}
\caption{Left:  comparing $n^{-1/2}u_n(n^{3/2}Y^3)$ for $n=2,5,10$ (pink, purple, blue dotted lines) with $U(Y^3)$ (red solid curve) for $Y>y_\mathrm{c}^{1/3}$ .  Right:  a similar plot but with imaginary parts of the purely imaginary functions plotted on the imaginary $Y$-axis.}
\label{fig:Accuracy}
\end{figure}

\subsection{Transition to elliptic behavior}
When $y>0$ decreases below the threshold value $y_\mathrm{c}\approx 0.29177$, a topological change occurs in the structure of the zero level curve such that the inequality $\mathrm{Re}(h(\eta,y))>0$ can no longer be satisfied on $\Sigma_0^+\cup\Sigma_\infty^+$.  To study $u_n(x)$ in this situation, it is necessary to dispense with the assumption that the cubic polynomial $P(\mu,y)$ defined in \eqref{eq:spectral-polynomial} has a repeated root.  The coefficient $c$ then has to be determined instead as a function of $y$ to guarantee that an outer parametrix constructed from the elliptic functions of the Riemann surface $w^2=P(\mu,y)$ remains bounded as $n\to\infty$.  This is done by imposing what is frequently called a Boutroux condition on $P(\mu,y)$.  Once this is done, the invariants $g_2$ and $g_3$ of the approximating Weierstra\ss\ equation are determined as functions of $y$ and a more involved Riemann-Hilbert analysis can be used to justify the non-equilibrium approximation of $n^{-1/2}u_n(n^{3/2}(y+n^{-1}z))$ by the right-hand side of \eqref{eq:non-equilibrium-approx}.  The details will not be given here.

The threshold value $y_\mathrm{c}$ is obtained by requiring that the double root $d=d(y)>0$ of $P(\mu,y)$ lies on the zero level of $\mathrm{Re}(h(\eta,y))$, i.e., by solving for $y$ the condition $\mathrm{Re}(h(\ii d(y),y))=0$.   The effect of a sign change on the topology of the level curve $\mathrm{Re}(h(\eta,y))=0$ as $y$ decreases through $y_\mathrm{c}$ is illustrated in the central and right-hand panels of Figure~\ref{fig:sign-charts}.
\begin{figure}[h]
\includegraphics[width=0.32\linewidth]{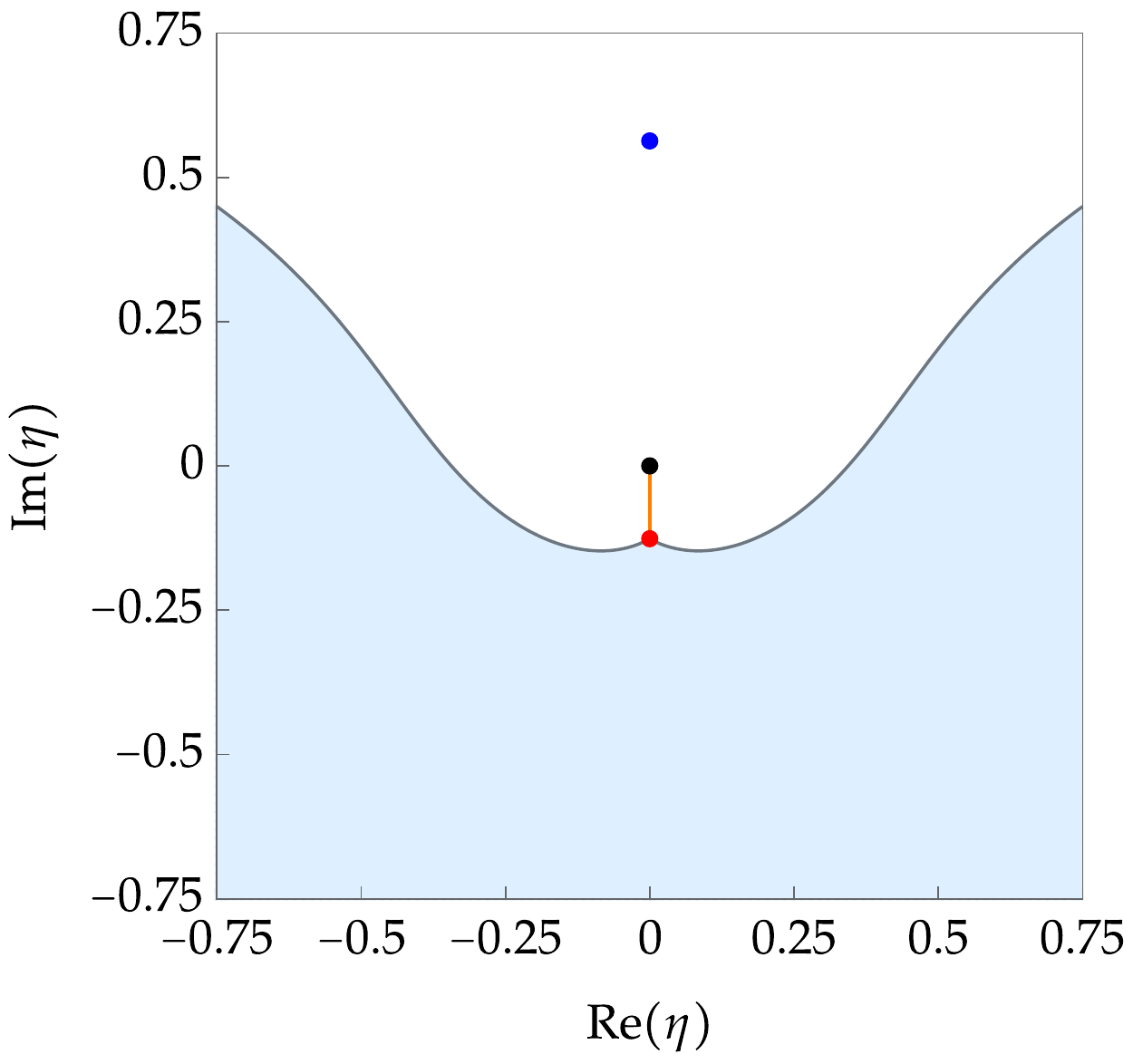}\hfill\includegraphics[width=0.32\linewidth]{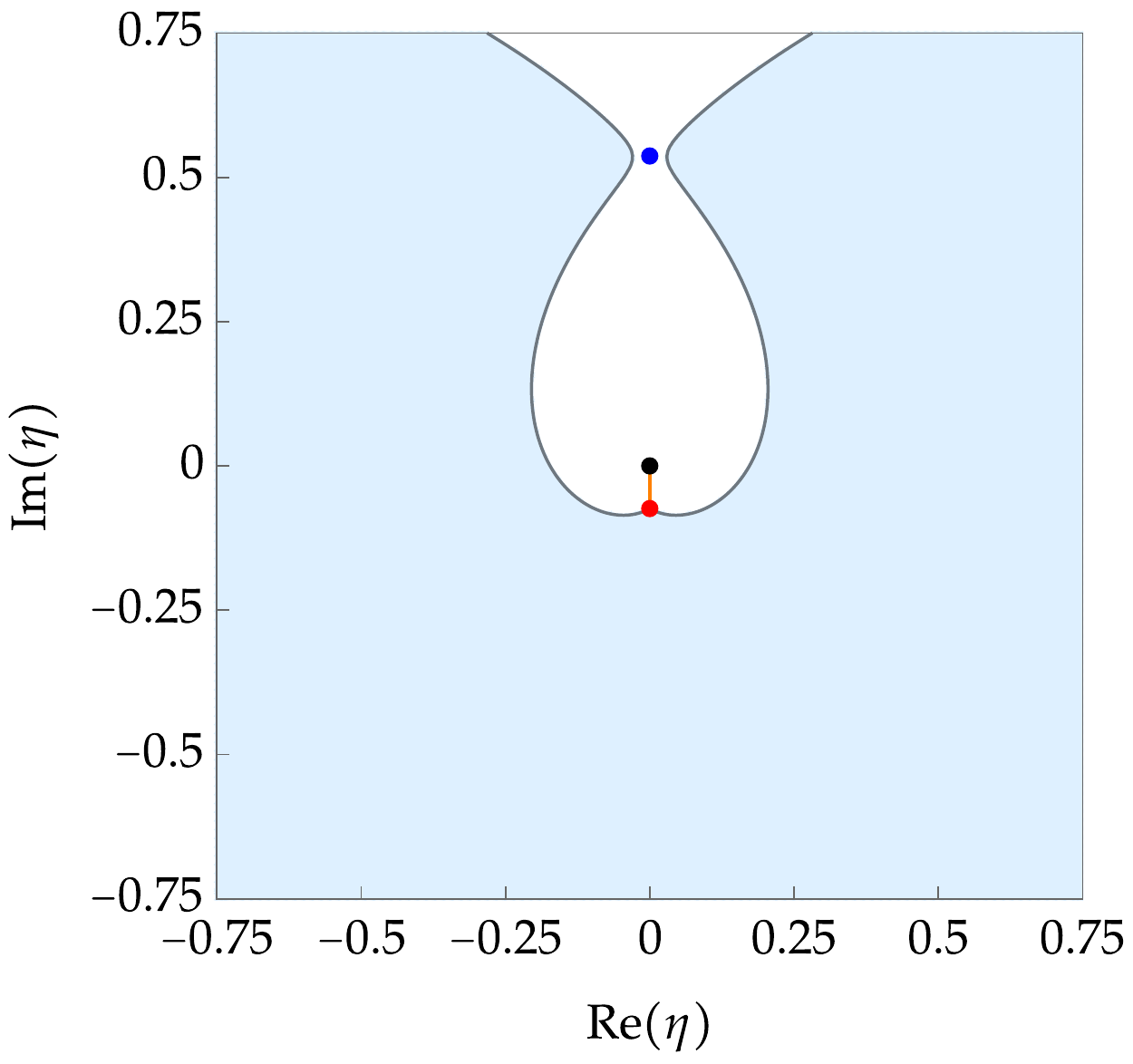}\hfill%
\includegraphics[width=0.32\linewidth]{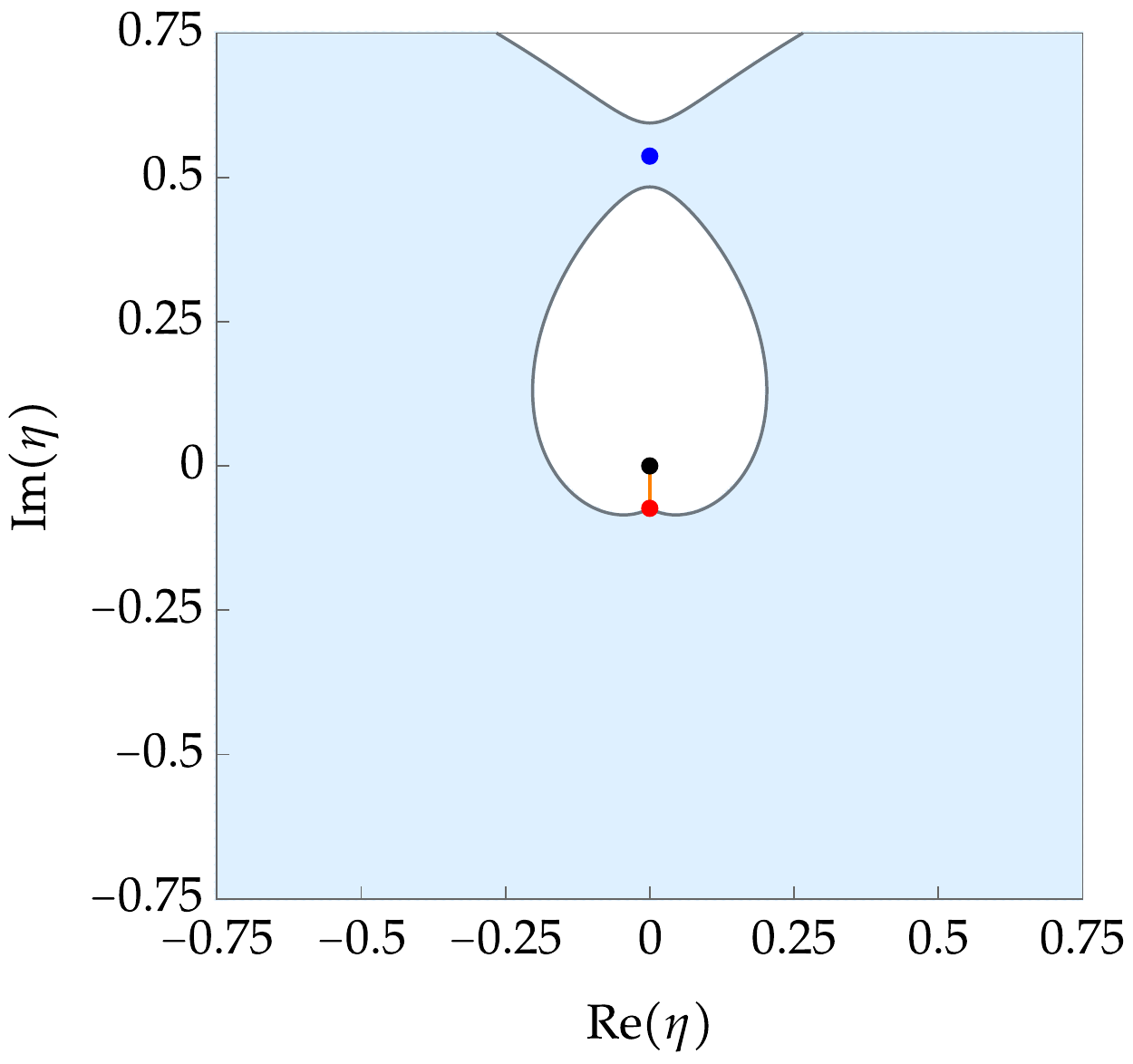}
\caption{Left:  the sign chart of $\mathrm{Re}(h(\eta,y))$ for $y=0.4$.  The branch cut of $h_\eta(\eta,y)$ is shown with an orange line, the points $\eta=\ii s,0,\ii d$ are shown with red, black, and blue dots respectively, and the region where $\mathrm{Re}(h(\eta,y))<0$ holds is indicated with light blue shading.  Center:  the same for $y=0.292$.  Right:  the same for $y=0.291$.  Note that $y=0.291<y_\mathrm{c}$, so Theorem \ref{thm:positive-exterior} does not apply here.}
\label{fig:sign-charts}
\end{figure}

\subsection{Analytic continuation for $0\le|\arg(y)|\le 3\pi$ and boundary curves}
 While Riemann-Hilbert Problem~\ref{rhp:algebraic} and the equivalent conditions for the matrix $\mathbf{Y}^{(n)}(\lambda,x)$ discussed in Section~\ref{sec:branch-cut} were originally formulated assuming that $x>0$, under a suitable deformation of the jump contours near the origin they remain a valid description of $u_n(x)$ for $\arg(x)\neq 0$.  The scalings $x=n^{3/2}y$ and $\lambda=n^{-1/2}y^{-1}\eta$ introduced in Section~\ref{sec:rescaling} can also be used, and hence it makes sense to analyze the matrix $\mathbf{Z}^{(n)}(\eta,y)$ for large $n$ when $\arg(y)\neq 0$.  For this problem, the steepest descent directions as $\eta\to\infty$ are independent of $\arg(y)$;  however the steepest descent directions into the singularity at $\eta=0$ that are vertical when $\arg(y)=0$ more generally lie tangent to the rays with angles $2\arg(y)\pmod\pi$.  So the tangents at the origin of the jump contour arcs $\Sigma_0^\pm$ have to rotate as $y$ moves off the positive real axis, in the same direction as $\arg(y)$ but twice as much. 
 
Degenerate spectral curves (for which the cubic $P(\mu,y)$ has a repeated root, see Section~\ref{sec:SpectralCurves}) exist for $\arg(y)\neq 0$.  Given the solution $s=s(y)<0$ and $d=d(y)>\frac{1}{2}$ of \eqref{eq:matching-coefficients} analytic for $y>0$, one simply continues the solution into the complex plane.  Although one cannot generally write $h(\eta,y)$ in terms of principal branches as in \eqref{eq:h-log}, the formula \eqref{eq:hprime} for $\partial h/\partial\eta$ remains valid under the interpretation that the right-hand side is analytic off an arc $\Sigma_0^-$ connecting the origin with $\eta=\ii s$, tends to $-\ii$ as $\eta\to\infty$, and its square is equal to the single-valued function $(\eta-\ii d)^2(-\ii\eta)^{-3}(-\ii\eta-s)$.  Using the condition $2d+s=1$ in \eqref{eq:matching-coefficients} to eliminate $d$ then shows that $\partial h/\partial\eta$ depends on $y$ only through $s$; the same condition then guarantees that 
\begin{equation}
L(s):=\mathrm{Re}\left(\int_{\ii s}^{\ii (1-s)/2}\frac{\partial h}{\partial\eta}\,\dd\eta\right)
\end{equation}
is well defined regardless of the path of integration taken in the multiply connected domain of analyticity of the integrand.  Topological changes in the level curve $\mathrm{Re}(h(\eta,y))=0$ occur when $y$ corresponds to a value of $s$ for which $L(s)=0$, since this condition detects the double root $\eta=\ii d=\ii(1-s)/2$ lying on the same level curve as $\eta=\ii s$.  For reliable numerical evaluation of $L(s)$ it is convenient to integrate by parts and obtain the equivalent formula
\begin{multline}
L(s)=\mathrm{Re}\left(\int_{\ii s}^0\left[2+\frac{\eta-\ii(1-s)/2}{\eta-\ii s}\right]\left(-\frac{\eta-\ii s}{\eta}\right)^{1/2}\,\dd\eta\right.\\
\left.{}-\ii
\int_{0}^{\ii (1-s)/2}\left[2+\frac{\eta-\ii(1-s)/2}{\eta-\ii s}\right]\left(\frac{\eta-\ii s}{\eta}\right)^{1/2}\,\dd\eta\right)
\end{multline}
where in each integral the path of integration is a straight line, and the fractional powers are principal branches, which are well-defined and analytic on the paths of integration.  The locus $L(s)=0$ in the $s$-plane is displayed in the left-hand panel of Figure~\ref{fig:boundary-curves}.  The point of self-intersection is exactly $s=\frac{1}{3}$ (corresponding to the coalescence of $\eta=\ii s$ and $\eta=\ii d$), and $L(s)=0$ for all $s\in[\frac{1}{3},1]$.

Since \eqref{eq:matching-coefficients} implies that $s$ satisfies the cubic equation $s(s-1)^2=-y^2$, and since we are taking the solution that for large $y>0$ satisfies $s=-y^{2/3}(1+o(1))$, upon introducing $Y=y^{1/3}$ it is straightforward to determine $s$ near $Y=\infty$ in the complex plane as an even analytic function of $Y$ with asymptotic behavior $s=-Y^2+\mathcal{O}(1)$ as $Y\to\infty$.  The discriminant of the cubic vanishes if $Y=0$ or if $Y^6=-\frac{4}{27}$, the latter giving six points on the circle centered at $Y=0$ of radius $|Y|=2^{1/3}/3^{1/2}$ at angles $\arg(Y)=\pm\pi/6, \pm\pi/2,\pm 5\pi/6$.  Defining $s$ as a function of $Y$ by analytic continuation in from $Y=\infty$ of the solution of $s(s-1)^2=-Y^6$ along radial paths we expect branch cuts to appear from some of these branch points, connecting them to the origin with straight lines.  However, this solution is positive real for large imaginary $Y$ and remains real upon continuation inwards along the imaginary $Y$-axis.  Moreover it is clear that the strict inequality $s>1$ holds for all nonzero imaginary $Y$.  Now, at a branch point we would also have $\dd (s(s-1)^2)/\dd s=0$ or $(3s-1)(s-1)=0$, but since this cannot vanish if $s>1$ it follows that the function $s=s(Y)$ of interest cannot be branched at the two purely imaginary branch points of the cubic $s(s-1)^2=-Y^6$.  It is, however, branched at the remaining four branch points, so the domain of analyticity of $s(Y)$ is the complement of the union of two crossing line segments, one with endpoints $\pm (2^{1/3}/3^{1/2})\ee^{\ii\pi/6}$ and the other with endpoints $\pm (2^{1/3}/3^{1/2})\ee^{5\pi \ii/6}$.  These two segments are shown in orange in the right-hand panel of Figure~\ref{fig:boundary-curves}.

We have already observed that $s(Y)>1$ holds for nonzero imaginary $Y$, so these points are not on the level curve $L(s(Y))=0$.  Approaching the origin along the imaginary axis one  necessarily arrives at the limiting value of $s=1$ consistent with $Y=0$.  From $s(s-1)^2=-Y^6$ one then sees from the double factor of $s-1$ on the left-hand side that as $\arg(Y)$ increases/decreases by $\pi/3$, $\arg(s-1)$ increases/decreases by $\pi$, so $s(Y)$ remains real as $Y$ moves outwards along the top/bottom edges of the branch cuts in the upper/lower half-plane, and decreases monotonically from $s=1$ to $s=\frac{1}{3}$ at the terminal branch points.  These ``outer'' edges of the branch cuts are therefore  points on the level curve $L(s(Y))=0$.  Continuing $s(Y)$ around any of the branch points to the ``inner'' edge of the branch cut, $s$ remains real but decreases further as $Y$ moves along the edge toward the origin, taking the value $s=0$ in the limit.  Since $s(Y)<\frac{1}{3}$ along the inner edges of the branch cuts, these edges are \emph{not} on the level curve $L(s(Y))=0$.   To summarize, the image of the interval $\frac{1}{3}\le s\le 1$ where $L(s)=0$ on the branch $s=s(Y)$ continued radially from $Y=\infty$ where $s(Y)=-Y^2+\mathcal{O}(1)$ consists of the top/bottom edges of the straight-line branch cuts in the upper/lower half-plane.  The images of the loop joining $s=\frac{1}{3}$ with itself and of the two non-real unbounded curves seen in the left-hand panel of Figure~\ref{fig:boundary-curves} are easy to compute numerically by evaluation of $L(s(Y))$.  The result is shown in the right-hand panel of Figure~\ref{fig:boundary-curves}, which should be compared with Figure~\ref{fig:ohyama} in which $\zeta=n^{1/2}Y$.  We denote the open subset of the $Y$-plane exterior to the bounded ``bow-tie'' component of the boundary curve by $\mathcal{E}$.  Note that $\mathcal{E}$ contains four unbounded arcs of the boundary curve; we will explain the significance of these below.

\begin{figure}[h]
\includegraphics[width=0.45\linewidth]{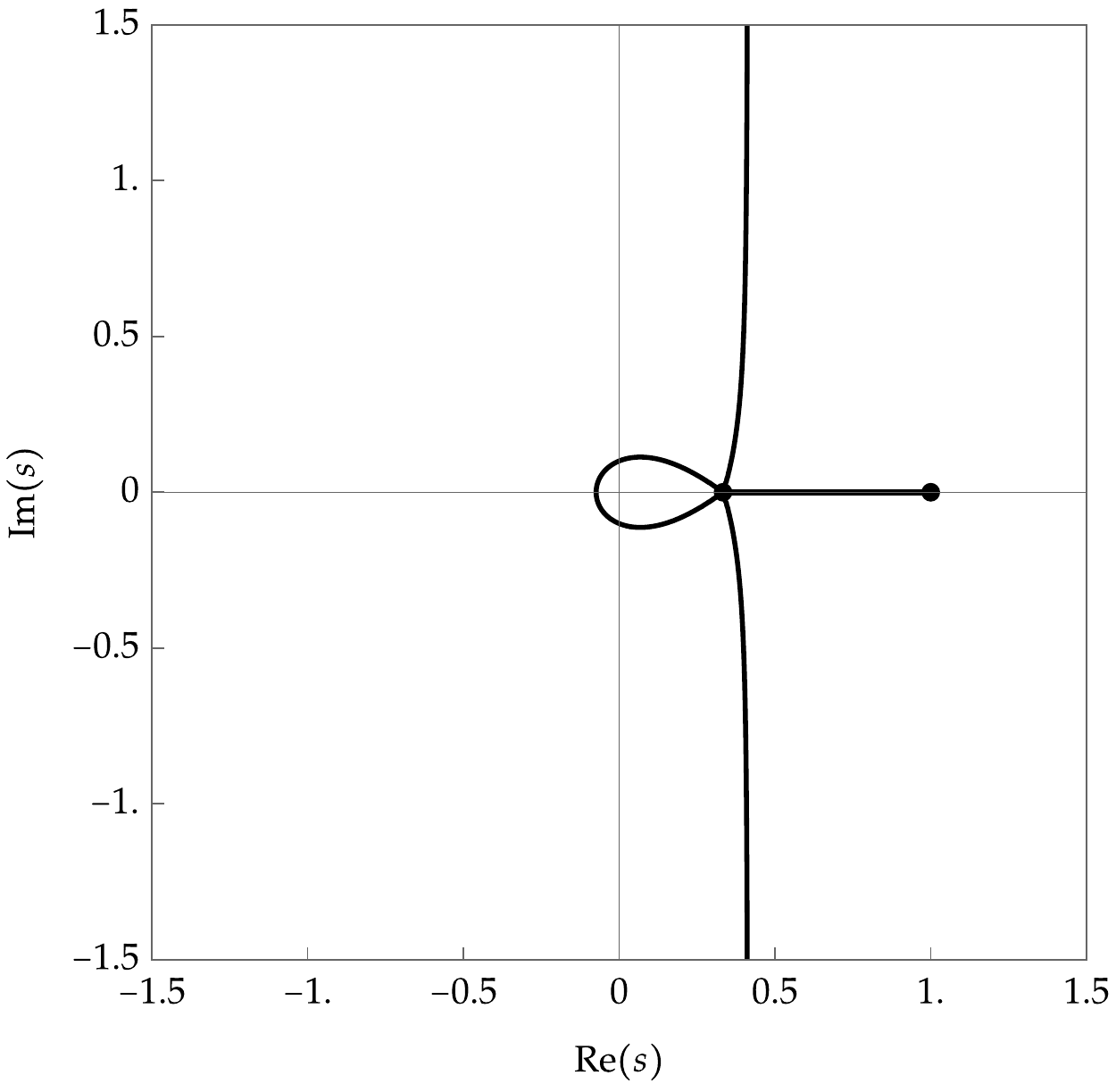}\hfill\includegraphics[width=0.45\linewidth]{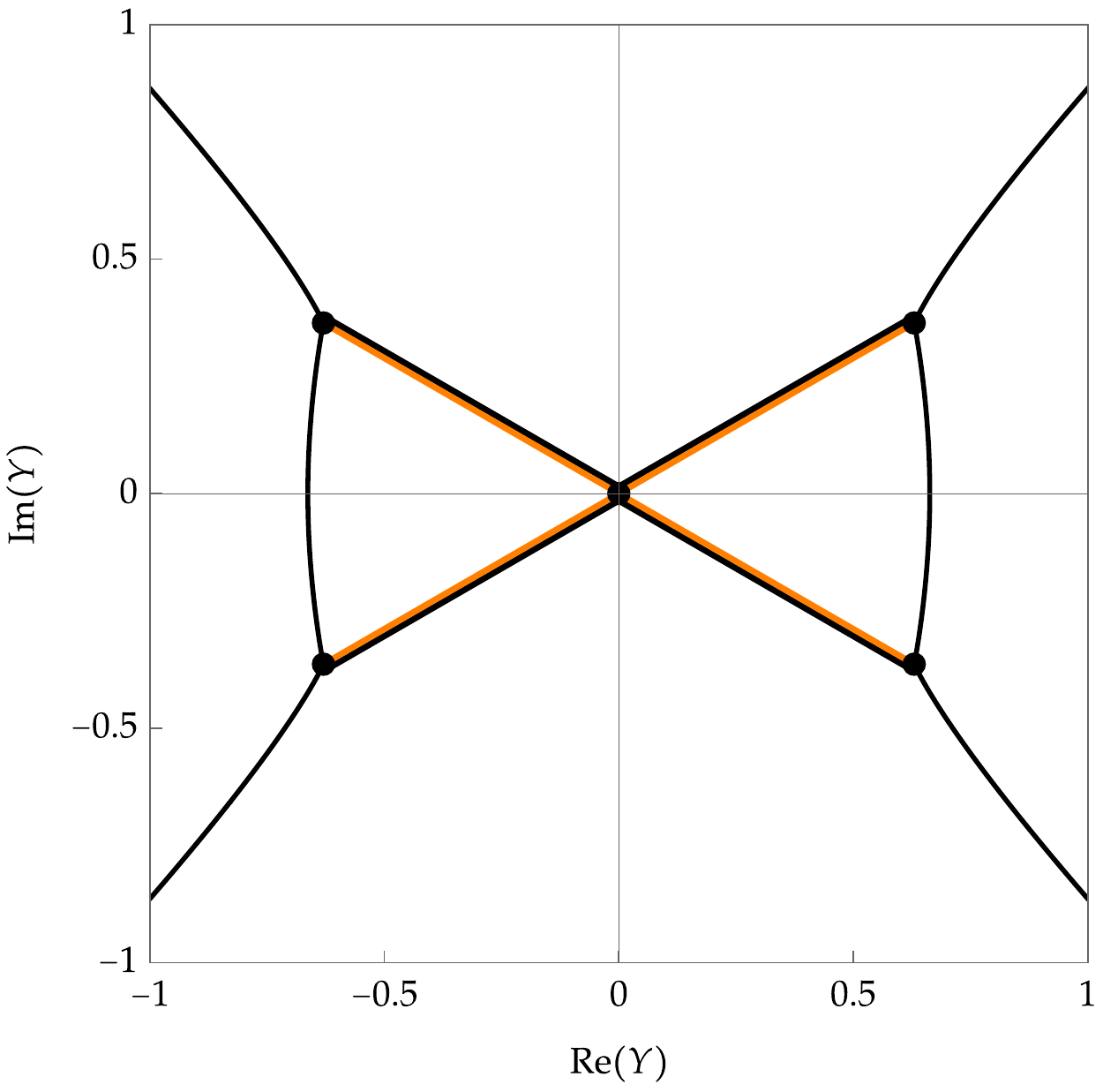}
\caption{Left:  the boundary curve in the $s$-plane.  Right:  the boundary curve in the $Y$-plane ($Y=y^{1/3}$); the crossing orange segments form the branch cut for the function $s(Y)$.}
\label{fig:boundary-curves}
\end{figure}

Given the symmetries in \eqref{eq:symmetries}, since $\zeta$ is proportional to $Y$, it is sufficient to assume that $0\le\arg(Y)\le\pi/2$.  To explain how to deform the jump contour for $\mathbf{Z}^{(n)}(\eta,y)$ as $\arg(Y)$ increases from zero for $Y\in\mathcal{E}$, we first observe that since the jump matrices on $\Sigma_0^+$ and $C^-$ are inverses we may reverse the orientation of $\Sigma_0^+$ and consider it joined with $C^-$ as a single contour from $\eta=0$ with a vertical tangent (when $\arg(Y)=0$) to the common endpoint of $C^-$ and $\Sigma_0^-$.  Then one sees that after reversing the orientation of $\Sigma_\infty^+$ it may be similarly combined with $C^+$ as a single contour from $\eta=\infty$ in the upper half-plane that terminates at the common endpoint of $C^+$ and $\Sigma_0^-$.  We call these combined contours $C^-$ and $C^+$ respectively.  The jump contour for $\mathbf{Z}^{(n)}(\eta,y)$ is therefore simplified to a union of four arcs:  $\Sigma_\infty^-\cup\Sigma_0^-\cup C^+\cup C^-$.  

For $\arg(Y)>0$ small with $Y\in\mathcal{E}$, there is an arc joining $\eta=0$ to $\eta=\ii s$ along which $h_\eta(\eta,y)^2\,\dd\eta^2<0$, and we choose this arc to be the contour $\Sigma_0^-$ which is also the branch cut for $h_\eta(\eta,y)$.  Since $h_\eta(\eta,y)$ has a residue of $-\frac{1}{2}$ at $\eta=\infty$, integration yields $h(\eta,y)$ as the function analytic for $\eta\in\mathbb{C}\setminus(\Sigma_0^-\cup\Sigma_\infty^-)$ and we choose the integration constant so that $h_+(\eta,y)+h_-(\eta,y)=0$ for $\eta\in\Sigma_0^-$.  This implies that also $g(0,y)=0$.  Moreover, $\mathrm{Re}(h(\eta,y))$ is harmonic except along $\Sigma_0^-$, which is an arc of its zero level curve albeit one across which it does not change sign.  The topological configuration of the zero level curve of $\mathrm{Re}(h(\eta,y))$ is common to all points $Y\in\mathcal{E}$ lying in the first quadrant and below the unbounded arc of the boundary curve emanating from the corner point $Y=2^{1/3}3^{-1/2}\ee^{\ii\pi/6}$.  The jump contour $\Sigma_\infty^-\cup\Sigma_0^-\cup C^+\cup C^-$ can be positioned relative to the level curve as illustrated in the representative plot shown in green (and orange, for $\Sigma_0^-$) in panel $\boxed{3}$ of Figure~\ref{fig:SignaturePlots}.
\begin{figure}[h]
\includegraphics[width=0.24\linewidth]{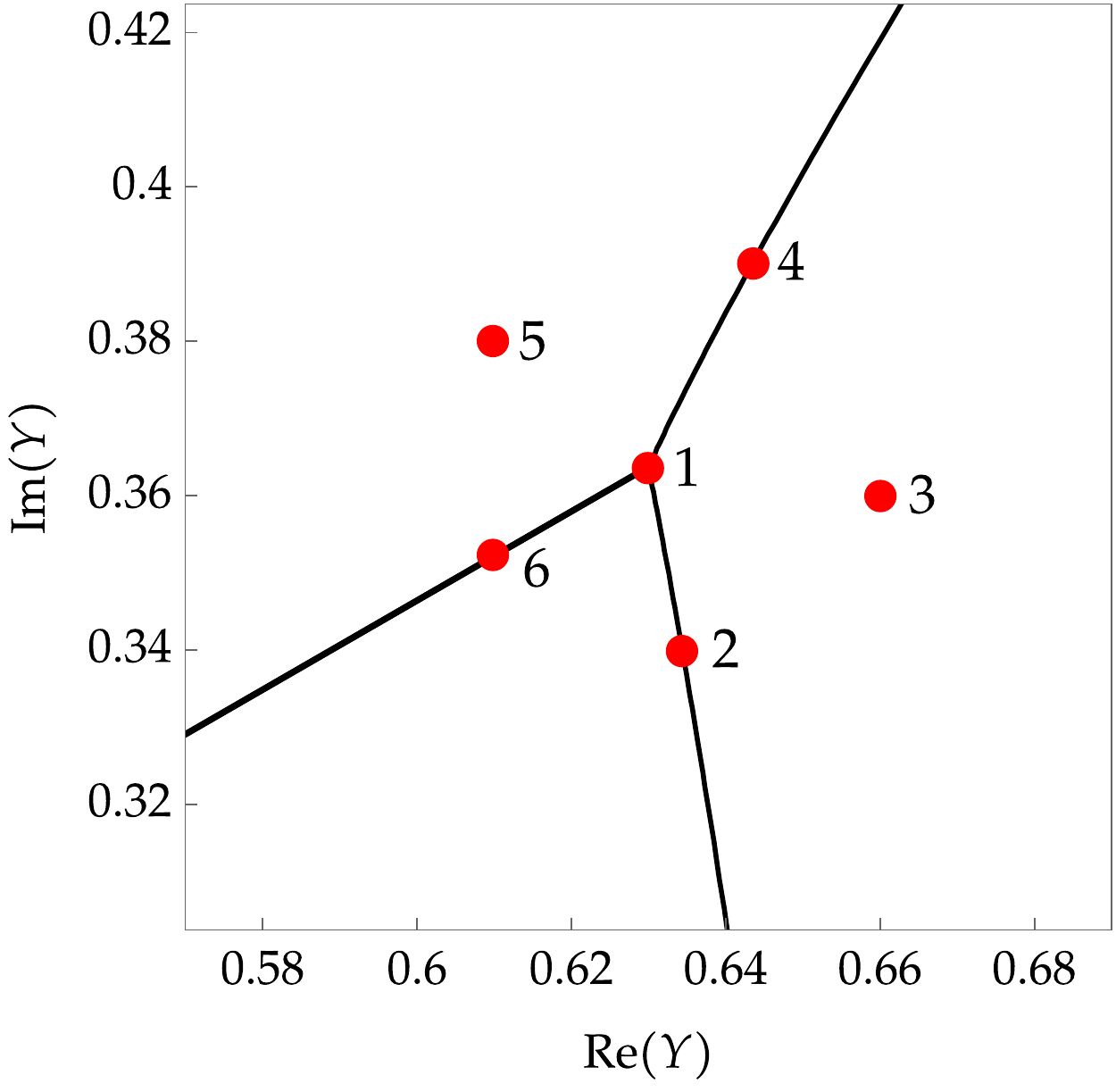}\hfill%
\includegraphics[width=0.24\linewidth]{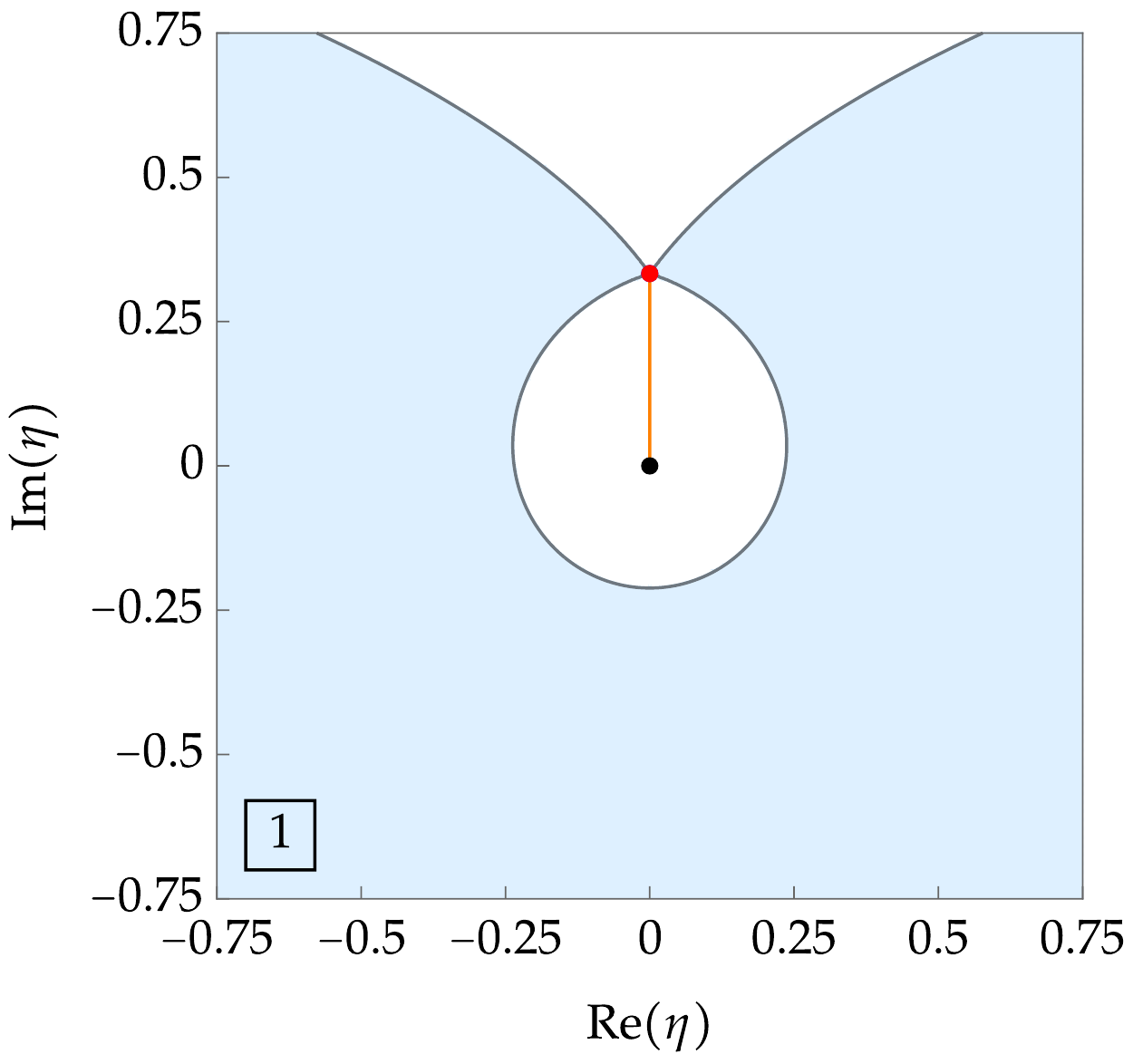}\hfill%
\includegraphics[width=0.24\linewidth]{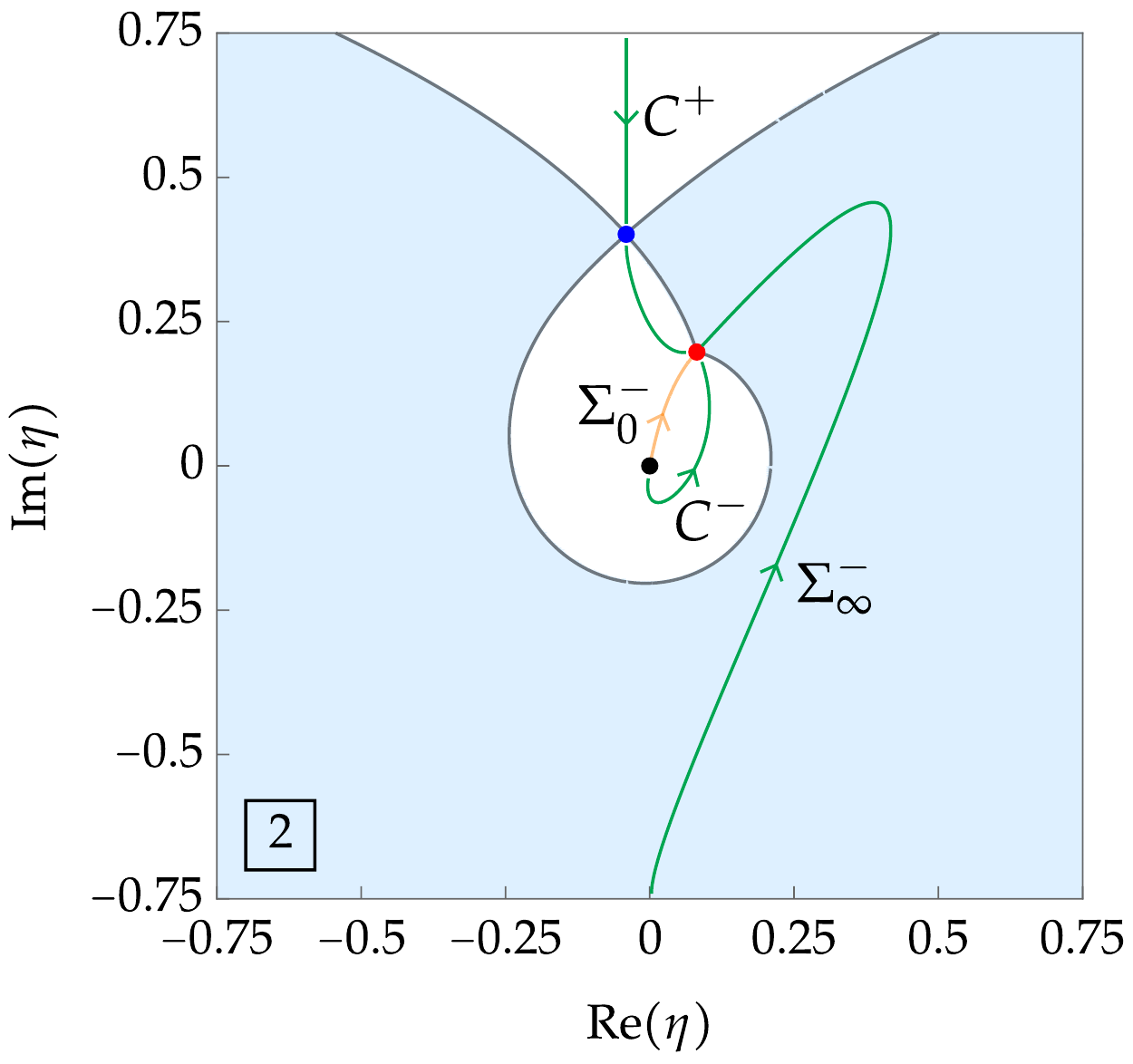}\hfill%
\includegraphics[width=0.24\linewidth]{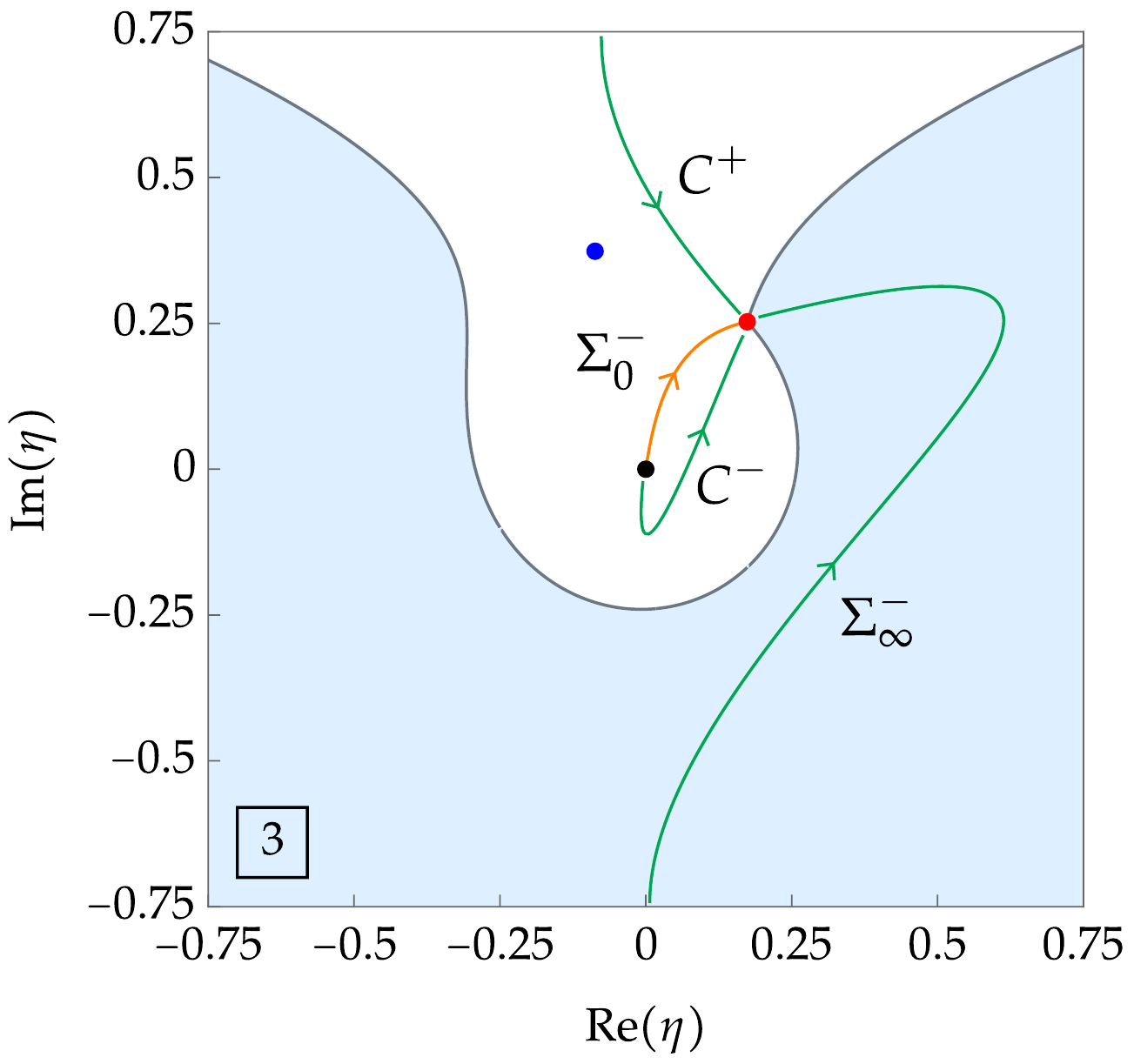}\\
\includegraphics[width=0.24\linewidth]{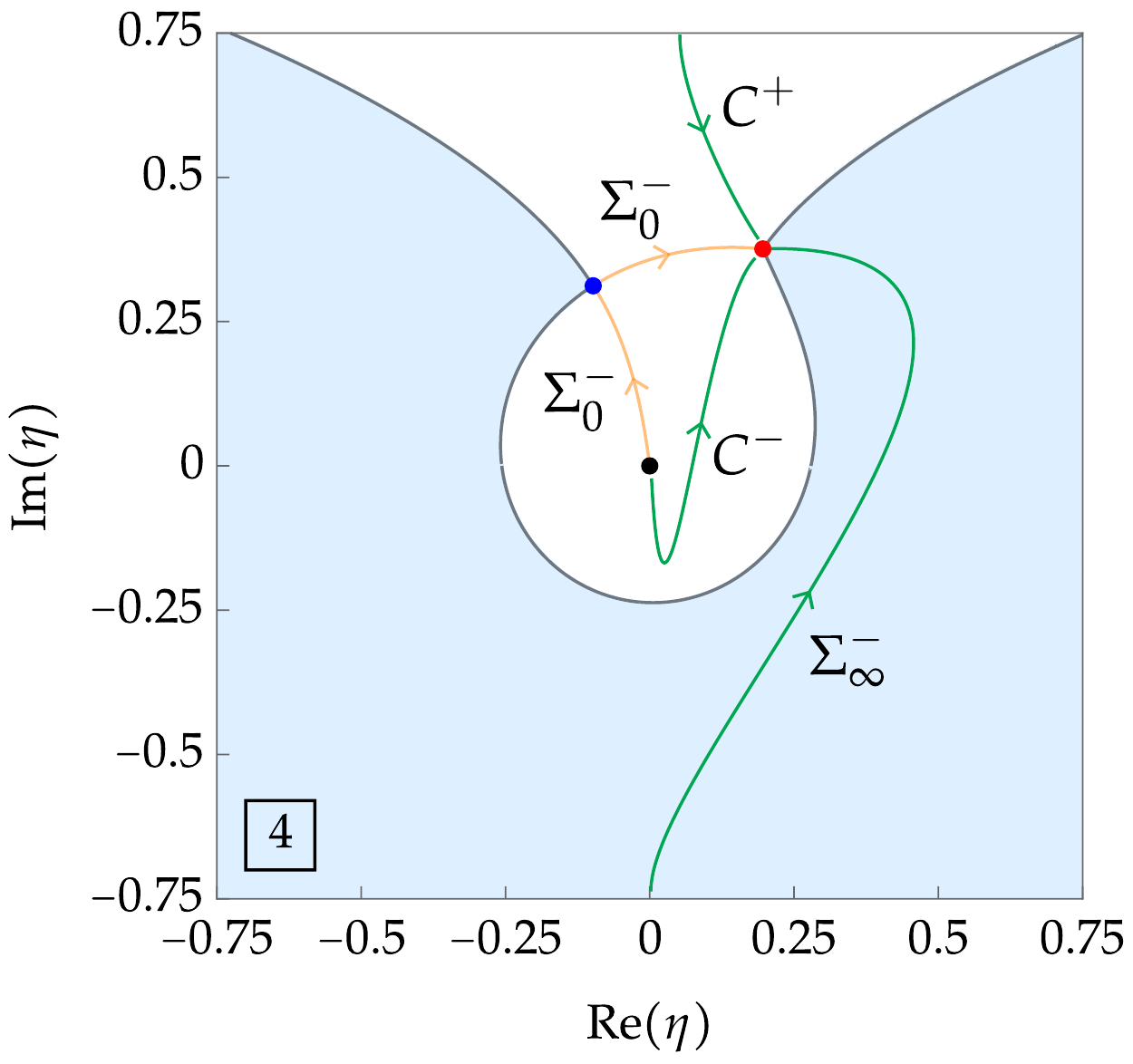}\hfill%
\includegraphics[width=0.24\linewidth]{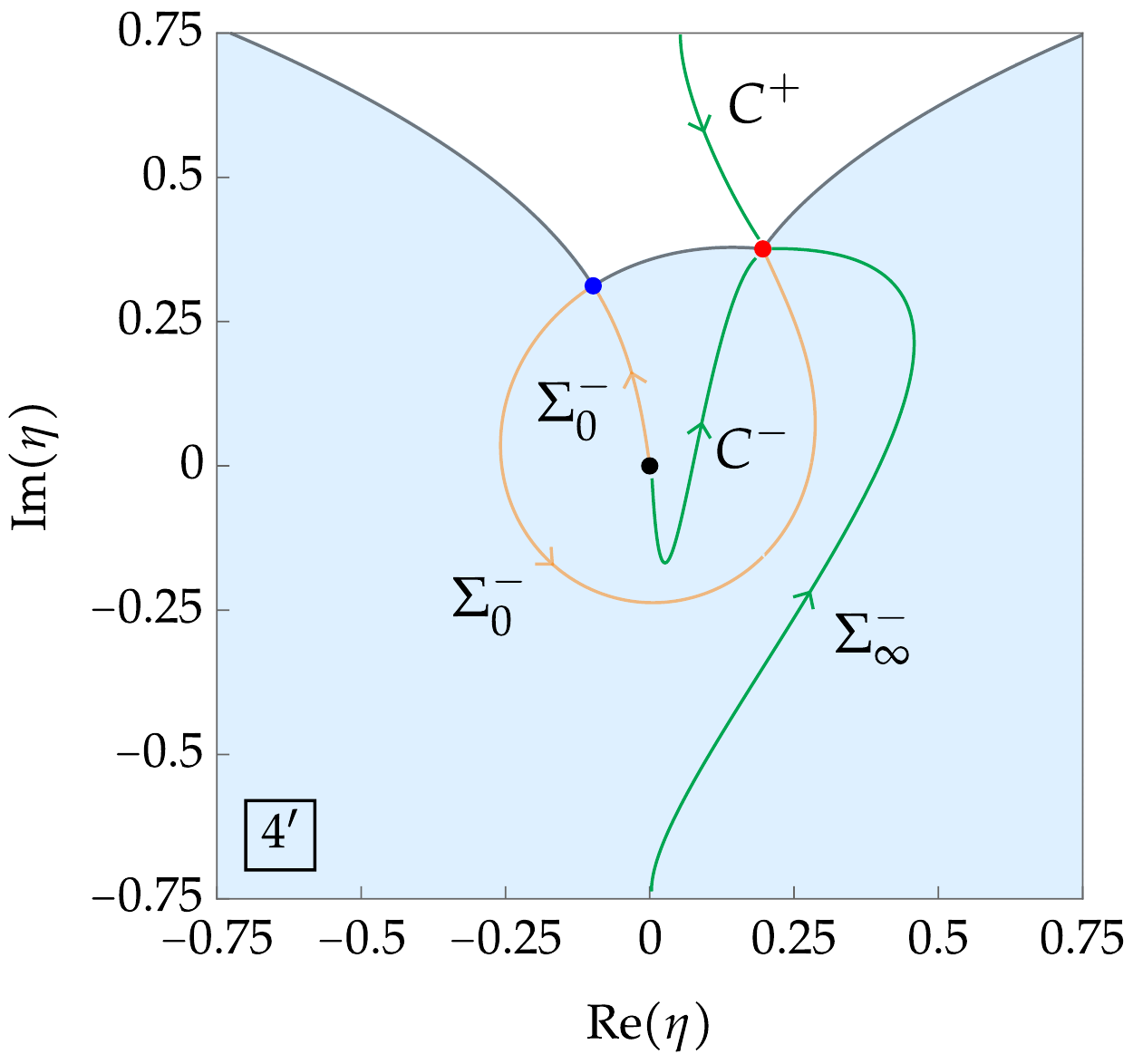}\hfill%
\includegraphics[width=0.24\linewidth]{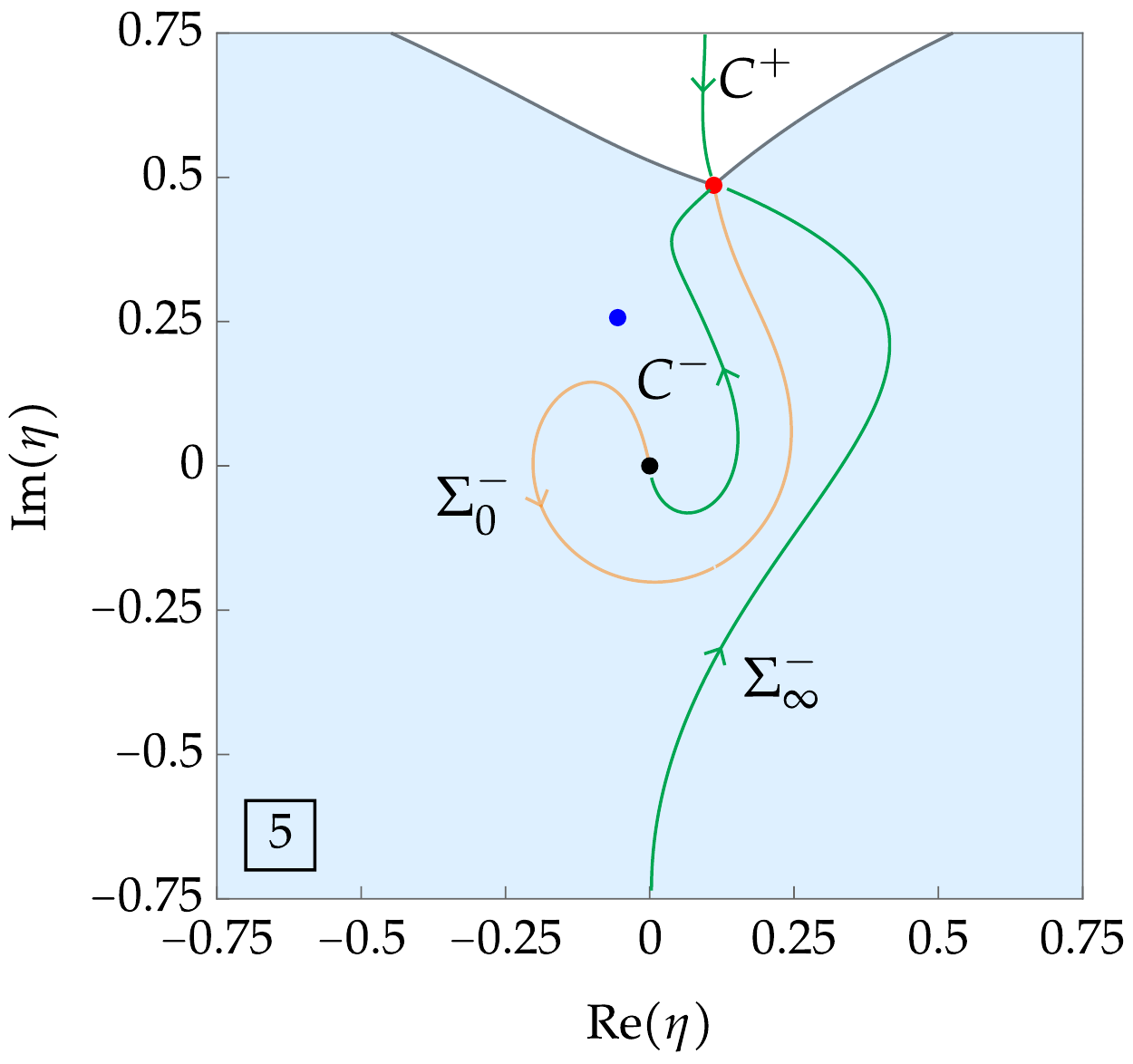}\hfill%
\includegraphics[width=0.24\linewidth]{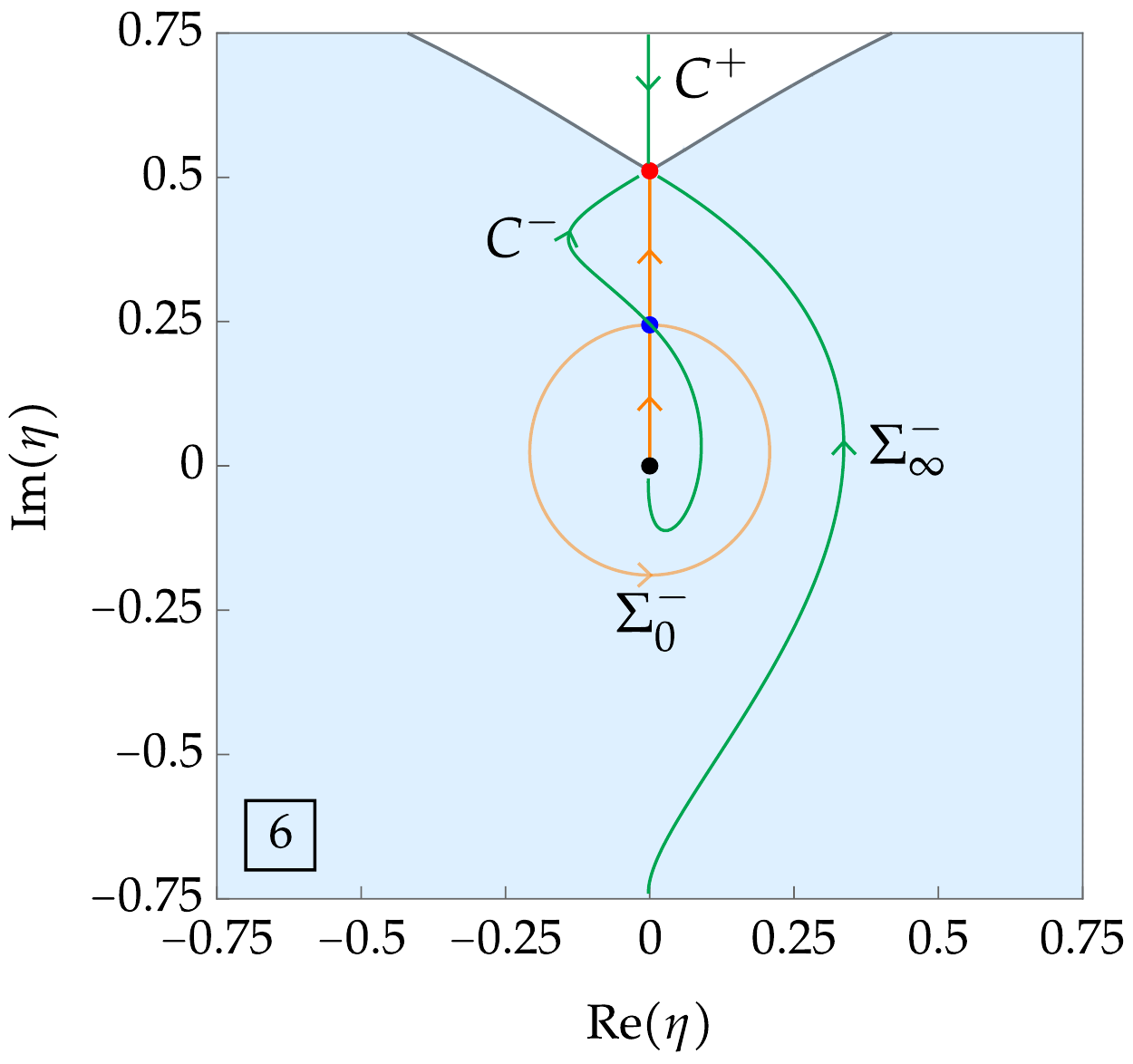}
\caption{Upper left-hand panel:  the $Y$-plane near the upper right-hand corner point indicating values of $Y=y^{1/3}$ corresponding to sign charts of $\mathrm{Re}(h(\eta,y))$ in the remaining panels (with the same annotation scheme as in Figure~\ref{fig:sign-charts}).  Superimposed on plots \framebox{2}--\framebox{6} are the arcs of the jump contour for $\mathbf{M}^{(n)}(\eta,y)$, with orange for $\Sigma_0^-$ (where $\mathrm{Re}(h(\eta,y))$ vanishes but fails to be harmonic) and green for the remaining arcs that are somewhat arbitrary.  Plots \framebox{4'}--\framebox{6} use a modified version of $h(\eta,y)$ and jump contour for $\mathbf{M}^{(n)}(\eta,y)$.}
\label{fig:SignaturePlots}
\end{figure}

When $Y$ moves onto the unbounded arc of the boundary curve from below, the point $\eta=\ii d$ descends from $\mathrm{Re}(h(\ii d,y))>0$ to the zero level curve $\mathrm{Re}(h(\ii d,y))=0$.  From the representative plot in panel $\boxed{4}$ of Figure~\ref{fig:SignaturePlots} one sees that the arc $\Sigma_0^-$ has developed a corner at $\eta=\ii d$.  However, this phenomenon has \emph{no effect} on the parametrix construction or error analysis, because the jump matrix for $\mathbf{M}^{(n)}(\eta,y)$ is independent of $\eta$ on $\Sigma_0^-$ because $h_+(\eta,y)+h_-(\eta,y)=0$.  On the other hand, to move $Y$ into the region above this ``phantom'' arc of the boundary curve, it is convenient to modify $\mathbf{M}^{(n)}(\eta,y)$ as follows.  In panel $\boxed{4}$ of Figure~\ref{fig:SignaturePlots} one can see that the arc of $\Sigma_0^-$ joining $\eta=\ii d$ with $\eta=\ii s$ forms part of the boundary of a ``bubble'' $B$ containing $\eta=0$ on which $\mathrm{Re}(h(\eta,y))\ge 0$ holds (strict inequality except on the other arc of $\Sigma_0^-$ joining the origin with $\eta=\ii d$).  For $\eta\in B$ we define $\widetilde{\mathbf{M}}^{(n)}(\eta,y):=\mathbf{M}^{(n)}(\eta,y)(-1)^n\ii\sigma_1$, while in the exterior of $B$ we set $\widetilde{\mathbf{M}}^{(n)}(\eta,y):=\mathbf{M}^{(n)}(\eta,y)$.    Dropping tildes, the jump contour for the modified $\mathbf{M}^{(n)}(\eta,y)$ is illustrated with green and orange arcs in panel $\boxed{4'}$ of Figure~\ref{fig:SignaturePlots}.  One can easily check that on both arcs of the new version of $\Sigma_0^-$, the original jump condition $\mathbf{M}^{(n)}_+(\eta,y)=\mathbf{M}^{(n)}_-(\eta,y)(-1)^n\ii\sigma_1$ holds.  The jump condition on $C^-$ is modified to read
\begin{equation}
\mathbf{M}^{(n)}_+(\eta,y)=\mathbf{M}^{(n)}_-(\eta,y)\ee^{-nh(\eta,y)\sigma_3}\begin{bmatrix}1& -\ii\\0 & 1\end{bmatrix}\ee^{nh(\eta,y)\sigma_3},\quad\eta\in C^-.
\end{equation}
Since the property that $\mathrm{Re}(h(\eta,y))$ is harmonic except on $\Sigma_0^-$ is a useful one to maintain, for convenience we redefine $h(\eta,y)$ by setting $\widetilde{h}(\eta,y):=-h(\eta,y)$ for $\eta\in B$ and $\widetilde{h}(\eta,y):=h(\eta,y)$ elsewhere.  The condition $h_+(\eta,y)+h_-(\eta,y)=0$ on the arc of $\Sigma_0^-$ joining $\eta=\ii d$ with $\eta=\ii s$ then shows that $\widetilde{h}(\eta,y)$ is analytic across that arc, and dropping tildes, the jump contour for $h(\eta,y)$ is once again a subset of that of $\mathbf{M}^{(n)}(\eta,y)$.  When we use this modified form of $\mathbf{M}^{(n)}(\eta,y)$, we need to account for the artificial change of sign of $h(\eta,y)$ by noting that $g(\eta,y)=-h(\eta,y)+\Phi(\eta,y)$ holds for $\eta\in B$.  In particular, the condition $g(0)=0$ still holds.  Panels $\boxed{4'}$, $\boxed{5}$, and $\boxed{6}$ of Figure~\ref{fig:SignaturePlots} show the jump contour for the modified matrix $\mathbf{M}^{(n)}(\eta,y)$ on the landscape of the modified $\mathrm{Re}(h(\eta,y))$.  With this modification, the same proof applicable for $Y$ below the unbounded arc of the boundary curve also works \emph{mutatis mutandis} for $Y$ on and above this curve, with the same resulting asymptotic formula \eqref{eq:un-asymp-formula} for $u_n(x)$.  Therefore, we have the following generalization of Theorem~\ref{thm:positive-exterior}.
\begin{theorem}
$u_n(n^{3/2}Y^3)=n^{1/2}U+\mathcal{O}(n^{-1/2})$ holds uniformly for $Y$ in compact subsets of $\mathcal{E}$, where $U=U(Y)$ satisfies $8U^3+2U-Y^3=0$ and is analytic for $Y\in \mathcal{E}$ with $U=\frac{1}{2}Y$ as $Y\to\infty$.  In particular, $u_n(n^{3/2}Y^3)$ is pole- and zero-free on $\mathcal{E}$ for $n$ large.  The compact set $\mathbb{C}\setminus\mathcal{E}$ has a ``bow-tie'' shape with boundary consisting of two straight line segments joining the pairs $\pm (2^{1/3}/3^{1/2})\ee^{\ii\pi/6}$ and $\pm (2^{1/3}/3^{1/2})\ee^{5\pi\ii/6}$ and two curved arcs satisfying $L(s(Y))=0$.
\label{thm:general-exterior}
\end{theorem}

\end{document}